\documentclass{LMCS}

\def\eqalign#1{\null\,\vcenter{\openup\jot\mathsurround=0 pt
  \ialign{\strut\hfil$\displaystyle{##}$&$\displaystyle{{}##}$\hfil
      \crcr#1\crcr}}\,}

\usepackage{amsmath,amssymb}
\usepackage{bcprules}
\usepackage{macros,symbols,rules}
\usepackage{url}
\usepackage{paralist}

\usepackage{graphicx,color}
\usepackage{multicol}
\newif\iffull\fullfalse
\renewcommand{\inflabel}[1]{%
  \ifsuppressrulenames\else
    \def\lab{#1}%
    \ifx\lab\empty
      \relax
    \else
      \rn{(}\rn{\lab}\rn{)}%
  \fi\fi
}
\usepackage{enumerate,hyperref}

\def\doi{6 (4:8) 2010}
\lmcsheading%
{\doi}
{1--43}
{}
{}
{Nov.~16, 2009}
{Dec.~18, 2010}
{}

\begin{document}
\title{A Logical Foundation for Environment Classifiers}

\author[T.~Tsukada]{Takeshi Tsukada\rsuper a}
\address{{\lsuper a}Graduate School of Information Science, Tohoku University%,
%6-3-09, Aramaki-Aza-Aoba, Aoba-ku, Sendai city, Miyagi prefecture, Japan
}
\email{tsukada@kb.ecei.tohoku.ac.jp}

\author[A.~Igarashi]{Atsushi Igarashi\rsuper b}
\address{{\lsuper b}Graduate School of Informatics, Kyoto University%,
%36-1, Yoshidahonmachi, Sakyo-ku, Kyoto city, Kyoto prefecture, Japan
}
\email{igarashi@kuis.kyoto-u.ac.jp}

\keywords{Curry-Howard correspondence, Environment classifiers, Modal logic, Multi-stage calculus}
\subjclass{D.3.3, F.3.3, F.4.1}

\begin{abstract}
  Taha and Nielsen have developed a multi-stage calculus
  \(\lambda^{\alpha}\) with a sound type system using the notion of
  \emph{environment classifiers}.  They are special identifiers, with
  which code fragments and variable declarations are annotated, and
  their scoping mechanism is used to ensure statically that certain
  code fragments are closed and safely runnable.

  In this paper, we investigate the Curry-Howard isomorphism for
  environment classifiers by developing a typed \(\lambda\)-calculus
  \(\sname\).  It corresponds to multi-modal logic that allows
  quantification by transition variables---a counterpart of
  classifiers---which range over (possibly empty) sequences of labeled
  transitions between possible worlds.  This interpretation will
  reduce the ``run'' construct---which has a special typing rule in
  \(\lambda^{\alpha}\)---and embedding of closed code into other code
  fragments of different stages---which would be only realized by the
  cross-stage persistence operator in \(\lambda^{\alpha}\)---to merely
  a special case of classifier application.  \(\sname\) enjoys not
  only basic properties including subject reduction, confluence, and
  strong normalization but also an important property as a multi-stage
  calculus: time-ordered normalization of full reduction.  

  Then, we develop a big-step evaluation semantics for an ML-like
  language based on \(\sname\) with its type system and prove that the
  evaluation of a well-typed \(\sname\) program is properly staged.
  We also identify a fragment of the language, where erasure
  evaluation is possible.

  Finally, we show that the proof system augmented with a classical
  axiom is sound and complete with respect to a Kripke semantics of
  the logic.
\end{abstract}

\maketitle

\section{Introduction}
A number of programming languages and systems that support
manipulation of programs as
data~\cite{JonesGomardSestoft93PEbook,ConselLawallLeMeur04SCP,WicklineLeePfenning98PLDI,PolettoHsiehEnglerKaashoek99TOPLAS,TahaSheard00TCS}
have been developed in the last two decades.  A popular language
abstraction in these languages consists of the Lisp-like
\emph{quasiquotation} mechanism to create and compose code fragments
and a function to run them like \texttt{eval} in Lisp.  For those
languages and systems, a number of type systems for so-called
``multi-stage'' calculi have been
studied~\cite{TahaSheard00TCS,GlueckJorgensen95PLILP,Davies-Atemporallogicappro,Davies2001,604134,Yuse2006,1111060}
to guarantee safety of generated programs \emph{even before the
  generating program runs}.
%% A multi-stage calculus, whose execution is divided into stages, is a
%% framework of partial evaluation, run-time code generation, program
%% generation and macro expansion
%% \finish{cite: need citations}.
%% A multi-stage calculus can handle a code fragment
%% as further stage computatiton
%% and has constructors for code fragments.
%% Famous constructors for handling code fragments are
%% quoting, unquoting and evaluation,
%% which are implemented in List like languages.
%% A challenge of a type system for a multi-stage calculus
%% is to ensure safety not only of the execution,
%% but also of a generated code fragment.

Among them, some seminal work on the principled design of type systems
for multi-stage calculi is due to
Davies~\cite{Davies-Atemporallogicappro} and Davies and
Pfenning~\cite{DBLP:conf/popl/DaviesP96,Davies2001}.
They discovered the Curry-Howard
isomorphism between modal/temporal logics and multi-stage calculi by
identifying (1) modal operators in modal logic with type constructors
for code fragments treated as data and, in the case of temporal logic,
(2) the notion of time with computation stages.  For example, the
calculus $ \lambda^{\bigcirc} $~\cite{Davies-Atemporallogicappro},
which can be thought as a reformulation of Gl{\"u}ck and
J{\o}rgensen's calculus for multi-level generating
extensions~\cite{GlueckJorgensen95PLILP} by using explicit quasiquote
and unquote in the language, corresponds to a fragment of linear-time
temporal logic (LTL) with the temporal operator ``next'' (written $
\bigcirc $)~\cite{StirlingHandbook}.  Here, linearly ordered time
corresponds to the level of nesting of quasiquotations, and a modal
formula \(\bigcirc A\) to the type of \emph{code} of type \(A\).  It,
however, does not treat \texttt{eval}; in fact, the code type in
\(\lambda^{\bigcirc}\) represents code values, whose bodies are open,
that is, may have free variables, so simply adding \texttt{eval} to
the calculus does not work---execution may fail by referencing free
variables in the code.  The calculus developed by Davies and
Pfenning~\cite{DBLP:conf/popl/DaviesP96,Davies2001}, on the other
hand, corresponds to (intuitionistic) modal logic S4 (only with the
necessity operator $ \square $), in which a formula \(\square A\) is
considered the type of \emph{closed code values} of type \(A\).  It supports
safe \texttt{eval} since every code is closed, but inability to deal
with open code hampers generation of efficient code.  The following
work by Taha and
others~\cite{TahaSheard00TCS,MoggiTahaBenaissaSheard99ESOP,BenaissaMoggiTahaSheard99IMLA,604134,Calcagno2004}
sought various forms of combinations of the two systems above to
develop expressive type systems for multi-stage calculi.

% a possible world is indexed by a natural number,
% and can deal with open code fragments,
% which are code fragments with free variables.
% But $ \lambda^{\bigcirc} $ has no construct
% for runtime code execution,
% because execution of a code with free variable
% causes an unbound variable error.
% The calculus $ \lambda^{\square} $ corresponds to
% the $ \square $ fragment of the S4 modal logic
% and has only closed code fragments,
% which are code fragments with no free variable.
% Because of this restriction,
% every code fragments are safely runnable
% and can be embedded into any future stage.

Finally, Taha and Nielsen~\cite{604134} developed a multi-stage
calculus $ \lambda^{\alpha} $, which was later modified to make type
inference possible~\cite{Calcagno2004} and implemented as a basis of
MetaOCaml.  The calculus $ \lambda^{\alpha} $ has a strong type system
while supporting open code, \textbf{run}, (which corresponds to \texttt{eval}), and
a mechanism called cross-stage persistence (CSP), which allows a
value to be embedded in a code fragment evaluated later.  They
introduced the notion of \emph{environment classifiers} (or, simply,
classifiers), which are special identifiers with which code fragments
and variable declarations are annotated, to the type system.  A key
idea is to reduce the closedness checking of a code fragment (which is
useful to guarantee the safety of \textbf{run}) to the freshness
checking of a classifier.  Unfortunately, however, correspondence to a
logic is not clear for \(\lambda^{\alpha}\) any longer, resulting in
somewhat ad-hoc typing rules and complicated operational semantics,
which would be difficult to adapt to different settings.

% Using this notion, Taha and Nielsen
% developed a multi-stage typed calculus which can deal with both open
% and closed code fragments and has a construct for safe-running.
% However, the intuitive meaning of classifiers and its binder are not
% clear and the typing rule for the \texttt{run} construct is unnatural.
% Moreover, reduction rules are less understandable.

In this paper, we investigate a Curry-Howard isomorphism for
environment classifiers by developing a typed \(\lambda\)-calculus
\(\sname\).  As a computational calculus, \(\sname\) is equipped with
quasiquotation (annotated with environment classifiers) and
abstraction over environment classifiers just like \(\lambda^\alpha\),
with application of a classifier abstraction to a \emph{possibly empty
  sequence of} environment classifiers, which makes \(\sname\)
different from \(\lambda^\alpha\).  Intuitively, (the type system of)
\(\sname\) can be considered a proof system of a multi-modal logic to
reason about deterministic labeled transition systems.  Here, modal
operators are indexed with transition labels, and so the logic is
multi-modal.  One notable feature of the logic is that it has
quantification that allows one to express ``for any state
transitions,'' where a state transition is a possibly empty sequence
of labels.  This quantifier corresponds to types for classifier
abstractions, used to ensure freshness of classifiers, which correspond
to transition labels (and variables ranging over their sequences).

% The new calculus corresponds to an (intuitionistic)
% multi-modal logic that allows quantification over state transitions in
% the underlying Kripke structure.  Intuitively, a Kripke structure for
% this modal logic would be an LTS where transitions between states, or
% possible worlds, are annotated with different labels.  Modal operators
% are indexed with these labels, hence the logic is multi-modal.  The
% quantification allows one to express ``for any state transitions''
% where a state transition is a possibly empty sequence of labels.  From
% a computational point of view, multiple modalities correspond to
% indexing of code types by classifiers and quantifiers to types for
% classifier abstractions, used to ensure freshness of classifiers.

% to set, in the
% Kripke semantics, classifiers to range over possibly empty
% \emph{sequences} of labels, attached to the transition function on
% possible worlds.  

A pleasant effect of this logical interpretation---in particular,
interpreting environment classifiers as variables ranging over 
sequences of transition labels---is that it will reduce the
\textbf{run} construct---which has a peculiar typing rule in
\(\lambda^{\alpha}\)---and embedding of closed code into other code
fragments of different stages---which would be only realized by the
CSP operator in \(\lambda^{\alpha}\)---to merely a special case of
classifier application.

Our technical contributions can be summarized as follows:
\begin{enumerate}[$\bullet$]
\item Identification of a modal logic that corresponds to (a
  computational calculus with) environment classifiers;
\item Development of a new typed \(\lambda\)-calculus \(\sname\),
  naturally emerged from the correspondence, with its syntax,
  operational (small-step reduction and big-step evaluation)
  semantics, and type system;
\item Proofs of basic properties, which a multi-stage calculus is expected to enjoy; and 
\item Proofs of soundness and completeness of the proof system (augmented with a
classical axiom) with respect to a Kripke semantics of the logic.
\end{enumerate}
Our calculus \(\sname\) not only enjoys the basic properties such as
subject reduction, confluence, and strong normalization but also
time-ordered normalization~\cite{Davies-Atemporallogicappro,Yuse2006},
which says (full) reduction to a normal form can always be performed
according to the order of stages.  We extend \(\sname\) with base
types and recursion, define a big-step evaluation semantics as a basis
of a multi-stage programming language such as MetaOCaml, and prove the
evaluation of a well-typed program is safe and staged, i.e., if a
program of a code type evaluates to a result, it is a code value whose
body is a well-typed program, again.  We also develop erasure
semantics, where information on classifiers is (mostly) discarded, and
identify a subset of the language, where the original and erasure
semantics agree, by an alternative type system.  It turns out that the
subset is rather similar to \(\lambda^i\)~\cite{Calcagno2004}, whose
type system is used in the current implementation of MetaOCaml.

One missing feature in \(\sname\) is CSP for all types of values but
we do not think it is a big problem.  First, CSP for primitive types
such as integers is easy to add as a primitive; CSP for function types
is also possible as long as they do not deal with open code, which, we
believe, is usually the case.
% (if one gives up printing code representation of functional values
% as in MetaOCaml).  \finish{I find that the CSP for function values
%   is difficult to add because function may return open code
%   fragments.  For example, \( x : \tau @ \alpha \vdash^\varepsilon
%   \lambda y : \sigma.  \Next{\alpha}{x} : \sigma \to
%   \Seal{\alpha}{\tau} \).}
Second, as mentioned above, embedding closed code into code fragments
of later stages is supported by a different means.  It does not seem
very easy to add CSP for open code to \(\sname\), but we think it is
rarely needed.  For more detail, see Section~\ref{sec:compare-alpha}.

% To give a foundation and an intuitive explanation for
% environment classifiers,
% we develop a logic which corresponds a multi-stage calculus
% with environment classifiers,
% following the suggestion by Daveis and Davies and Pfenning.
% The logic is a multi-modal logic
% because a calculus with environment classifiers has
% code type constructors as many as classifiers,
% and possible worlds of its Kripke semantics are
% related to the sequences of classifiers
% because a stages of a calculus
% is naturally indicated by sequences of classifiers.
% From this principle,
% we construct a Kripke semantics
% in which the structure of possible worlds is
% a labeled transition system.
% Then a code type named a classifier corresponds to a labeled transition
% and a stage indicated by a sequence of classifiers
% to a state which is reachable through a labeled path.
% These two correspondence are key idea of this paper
% and give an intuitive explanation of classifiers.

% We also propose a multi-stage calculus $ \sname $,
% which corresponds to the logic by Curry-Howard correspondence
% and this calculus differs from $ \lambda^{\alpha} $
% in the treatment of classifiers.
% In what follows,
% we call the counterpart of classifiers in our system \emph{transitions}
% which comes from the meaning in Kripke semantics,
% and use the name classifier only for the system
% given by Taha and Nielsen to avoid confusion.
% \finish{explanation about difference}

We can obtain a natural deduction proof system of a new logic that
corresponds to the calculus \( \sname \) just by removing terms from
typing rules, as usual.  It is also easy to see that terms and
reduction in the calculus correspond to proofs and proof normalization
in the logic, respectively.
% (Note that there is one-to-one correspondence between
% typable terms and type derivations in the calculus.)

Of course, we should answer an important question: ``What does this
logic really mean?''  We will elaborate the intuitive meaning of
formulae in Section~\ref{sec:overview} and proof rules can be
understood according to this informal interpretation but, to answer
this question more precisely, one has to give a semantics and prove
the proof system is sound and complete with respect to the semantics.
However, the logic is intuitionistic and it is not straightforward to
give (Kripke) semantics~\cite{PlotkinStirling86TARK}.  So, instead of
Kripke semantics of the logic directly corresponding to \(\sname\), we
give that of a classical version of the proof system, which has a
proof rule for double negation elimination and prove that the proof
system is sound and complete with it in Section~\ref{sec:logic}.  Even
though the semantics does not really correspond to \(\sname\), it
justifies our informal interpretation.

\subsection*{Organization of the Paper.}

In Section~\ref{sec:overview}, we review \(\lambda^\alpha\) and
informally describe how the features of its type system correspond to
those of a logic.  In Section~\ref{sec:calculus}, we define the
multi-stage calculus $ \sname $ and prove basic properties including
subject reduction, strong normalization, confluence, and time-ordered
normalization.  In Section~\ref{sec:miniml}, we define \miniML, an
extension of $ \sname $ with base types and recursion, with its
big-step semantics and prove that the big-step semantics implements
staged execution.  We also investigate erasure semantics of a subset
of \miniML{} here.  In Section~\ref{sec:logic}, we formally define (a
classical version of) the logic and prove soundness and completeness
of the proof system (augmented with a classical rule) with respect to
a Kripke semantics.  Lastly, we discuss related work and conclude.

\section{Interpreting Environment Classifiers in a Modal Logic}
\label{sec:overview}
In this section, we informally describe how environment classifiers
can be interpreted in a modal logic.  We start with
reviewing Davies'
\(\lambda^{\bigcirc}\)~\cite{Davies-Atemporallogicappro} to get an
intuition of how notions in a modal logic correspond to those in
a multi-stage calculus.  Then, along with reviewing main ideas of
environment classifiers, we describe our logic informally and how our
calculus \(\sname\) is different from \(\lambda^\alpha\) by Taha and
Nielsen~\cite{604134}.

\subsection{\(\lambda^\bigcirc\): Multi-Stage Calculus Based on LTL}
\label{subsec:bigcirc}
Davies has developed the typed multi-stage calculus
\(\lambda^{\bigcirc}\), which corresponds to a fragment of
intuitionistic LTL by the Curry-Howard isomorphism.  It can be
considered the \(\lambda\)-calculus with a Lisp-like quasiquotation
mechanism.  We first review linear-time temporal logic and the
correspondence between the logic and the calculus.

In LTL, the truth of propositions may depend on discrete and linearly
ordered time, i.e., a given time has a unique time that follows it.
Some of the standard temporal operators are \(\bigcirc\) (to mean
``next''), \(\Box\) (to mean ``always''), and \(U\) (to mean
``until'').  The Kripke semantics of (classical) LTL can be given by
taking the set of natural numbers as possible worlds;%
\footnote{%
  Note that this is equivalent to another, perhaps more standard
  presentation as a sublogic of CTL${}^*$~\cite{StirlingHandbook}.  }
then, for example, the semantics of \(\bigcirc\) is given by: $ n
\Vdash \bigcirc \tau $ if and only if $ n + 1 \Vdash \tau $, where \(
n \Vdash \tau \) is the satisfaction relation, which means ``\(\tau\)
is true in world---or, at time---\(n\).''

In addition to the usual Curry-Howard correspondence between
propositions and types and between proofs and terms, Davies has
pointed out additional correspondences between time and computation
stages (i.e., levels of nested quotations) and between the temporal
operator \(\bigcirc\) and the type constructor meaning ``the type of
code of''.  So, for example, \(\bigcirc \tau_1 \rightarrow \bigcirc
\tau_2\), which means ``if \(\tau_1\) holds at next time, then
\(\tau_2\) holds at next time,'' is considered the type of functions
that take a piece of code of type \(\tau_1\) and return another piece
of code of type \(\tau_2\).  According to this intuition, he has
developed \(\lambda^\bigcirc\), corresponding to the fragment of LTL
only with \(\bigcirc\).

$ \lambda^{\bigcirc} $ has two new term constructors $ \textbf{next }
M$ and $ \textbf{prev } M$, which correspond to the introduction and
elimination rules of \(\bigcirc\), respectively.  The type judgment of
\(\lambda^\bigcirc\) is of the form $ \Gamma \vdash^{n} M : \tau $,
where $ \Gamma $ is a context, $ M $ is a term, $ \tau $ is a type (a
proposition of LTL, only with \(\bigcirc\)) and $ n $ is a natural
number indicating a stage.  A context, which corresponds to
assumptions, is a mapping from variables to pairs of a type and a
natural number, since the truth of a proposition depends on time.  The
key typing rules are those for \textbf{next} and \textbf{prev}:
\vspace{1ex}
\begin{center}
\infrule{
  \Gamma \vdash^{n + 1} M : \tau
}{
  \Gamma \vdash^{n} \textbf{next } M : \bigcirc \tau
}
\qquad
\infrule{
  \Gamma \vdash^{n} M : \bigcirc \tau
}{
  \Gamma \vdash^{n + 1} \textbf{prev } M : \tau
} .
\end{center}
\vspace{1ex}
The former means that, if \(M\) is of type \(\tau\) at stage \(n+1\),
then, at stage \(n\), \(\textbf{next }M\) is code of type \(\tau\);
the latter is its converse.  Computationally, \textbf{next} and
\textbf{prev} can be considered quasiquote and unquote, respectively.
So, in addition to the standard \(\beta\)-reduction,
\(\lambda^{\bigcirc}\) has the reduction rule \(\textbf{prev
}(\textbf{next }M) \longrightarrow M\), which cancels \textbf{next} by
\textbf{prev}.

The code types in \(\lambda^\bigcirc\) are often called open code
types, since the quoted code may contain free variables, so naively
adding the construct to ``run'' quoted code does not work, since
it may cause unbound variable errors.

Although the logic is considered intuitionistic, Davies has only shown
that the proof system augmented with double negation elimination is
equivalent to a standard axiomatic
formulation~\cite{StirlingHandbook}, which is known to be sound and
complete with the Kripke semantics described above.  Kojima and
Igarashi~\cite{KojimaIgarashi08IMLA,KojimaIgarashi10IC} have studied the semantics of
intuitionistic LTL and shown that the proof system obtained from
\(\lambda^\bigcirc\) is sound and complete with the given semantics.
Even though the Kripke semantics discussed here does not really
correspond to the logic obtained from the calculus, it certainly helps
understand intuition behind the logic and we will continue to use
Kripke semantics in what follows for an explanatory purpose.

\subsection{Multi-Modal Logic for Environment Classifiers}
\label{subsec:quoting-with-classifiers}

Taha and Nielsen~\cite{604134} have introduced environment classifiers
to develop \(\lambda^\alpha\), which has quasiquotation, \textbf{run},
and CSP with a strong type system.  We explain how \(\lambda^\alpha\)
can be derived from \(\lambda^\bigcirc\).\footnote{%
  Unlike the original presentation, classifiers do not appear
  explicitly in contexts here.  The typing rules shown are accordingly
  adapted.}  Environment classifiers are a special kind of identifiers
with which code types and quoting are annotated: for each classifier $
\alpha $, there are a type constructor $\langle \tau \rangle^\alpha$
for code and a term constructor $ \langle M \rangle^{\alpha} $ to
quote \(M\).  Then, a stage is naturally expressed by a sequence of
classifiers, and a type judgment is of the form \(\Gamma \vdash^{A} M
: \tau\), where natural numbers in a \(\lambda^\bigcirc\) type
judgment are replaced with sequences \(A\) of classifiers.  So, the
typing rules of quoting and unquoting (written \(\tilde\; M\)) in $
\lambda^{\alpha} $ are given as follows:
\vspace{1ex}
\begin{center}
\infrule{
  \Gamma \vdash^{A \alpha} M : \tau
}{
  \Gamma \vdash^{A} \langle M \rangle^{\alpha} : \langle \tau \rangle^{\alpha}
}
\qquad
\infrule{
  \Gamma \vdash^{A} M : \langle \tau \rangle^{\alpha}
}{
  \Gamma \vdash^{A \alpha} \; \tilde{} \; M : \tau
} .
\end{center}
\vspace{1ex}
% Intuitively, code fragments in incomparable
% stages can be computed independently. 
Obviously, this is a generalization of $ \lambda^{\bigcirc} $: if only
one classifier is allowed, then the calculus is essentially $
\lambda^{\bigcirc} $.  
% As a result of this extension, however, the set
% of stages has non-linear ordering; it means that there are
% incomparable stages, e.g., \( \alpha \) and \( \beta \).

The corresponding logic would also be a generalization of LTL, in
which there are several ``dimensions'' of linearly ordered time.  A
Kripke frame for the logic is given by a transition system in which
each transition relation is a map.  More formally, a frame is a triple
\((S, L, \{\stackrel{\alpha}{\longrightarrow} \mid \alpha \in L\})\)
where \(S\) is the (non-empty) set of states, \(L\) is the set of
labels, and \(\mathord{\stackrel{\alpha}{\longrightarrow}} \in
S\rightarrow S\) for each \(\alpha \in L\).  Then, the semantics of
\(\langle \tau \rangle^\alpha\) is given by: $ s \Vdash \langle \tau
\rangle^{\alpha} $ if and only if $ s' \Vdash \tau $ for any \( s'
  \) such that \(s \stackrel{\alpha}{\longrightarrow} s'\), where $s$
and $s'$ are states.

The calculus $ \lambda^{\alpha} $ has also a scoping mechanism for
classifiers and it plays a central role to guarantee safety of \textbf{run}.
The term $ (\alpha) M $, which binds \(\alpha\) in \(M\), declares
that \(\alpha\) is used locally in \(M\) and such a local classifier
can be instantiated with another classifier by term $ M [\beta] $.
We show typing rules for them with one for \textbf{run} below:
\vspace{1ex}
\begin{center}
\infrule{
  \Gamma \vdash^{A} M : \tau \AND
    \alpha \notin \mathop{\textrm{FV}}(\Gamma, A)
}{
  \Gamma \vdash^{A} (\alpha) M : (\alpha) \tau
}
\quad
\infrule{
  \Gamma \vdash^{A} M : (\alpha) \tau
}{
  \Gamma \vdash^{A} M [\beta] : \tau \Subst{\alpha := \beta}
}
\quad
\infrule{
  \Gamma \vdash^{A} M : (\alpha) \langle \tau \rangle^{\alpha}
}{
  \Gamma \vdash^{A} \textbf{run} \; M : (\alpha) \tau
} .
\end{center}
\vspace{1ex}
The rule for $ (\alpha) M $ requires that $ \alpha $ does not occur in
the context---the term $ M $ has no free variable\footnote{%
It is important to distinguish labels of free variables from
free occurrences of classifiers.  For example, the term \( \langle \lam x : b. x \rangle^\alpha \)
has free occurrence of the classifier \( \alpha \), but no free variable labeled by \( \alpha \)
because there is no free variable at all.} labeled $ \alpha
$---and gives a type of the form \((\alpha)\tau\), which Taha and
Nielsen called \emph{$ \alpha $-closed type}, which characterizes a relaxed
notion of closedness.  For example,
the term \( \langle \lam x : b. x \rangle^\alpha \) is a closed term, 
so this term is \( \alpha \)-closed and the judgment
\( \emptyset \vdash^\varepsilon (\alpha)\langle \lam x : b. x \rangle^\alpha : (\alpha)\langle b \to b \rangle^\alpha \)
is valid.
The term \( \langle x \rangle^\alpha \), however, is not \( \alpha \)-closed because this
term has free variable \( x \) in the stage \( \alpha \),
but \( \beta \)-closed (if \( \beta \neq \alpha \)) because there is no free variable
in the stage containing the classifier \( \beta \).
The rule for $ \textbf{run} \; M $ says that
an $ \alpha $-closed code fragment annotated
with $ \alpha $ can be run.  Note that \(\langle \cdot \rangle^\alpha\)
(but not \((\alpha)\cdot\)) is removed in the type of \(\textbf{run
}M\).  Taha and Nielsen have shown that $\alpha$-closedness is
sufficient to guarantee safety of 
\textbf{run}.

When this system is to be interpreted as logic, it is fairly clear
that \((\alpha)\tau\) is a kind of universal quantifier, as Taha and
Nielsen have also pointed out~\cite{604134}.  Then, the question is
``What does a classifier range over?'', which has not really been
answered so far.  Another interesting question is ``How can the
typing rule for \textbf{run} be read logically?''

One plausible answer to the first question is that ``classifiers
range over the set of transition labels''.  This interpretation
matches the rule for \(M[\beta]\) and it seems that the typing rules
without \textbf{run} (with a classical axiom) are sound and complete
with the Kripke semantics that defines \(s \Vdash (\alpha)\tau \)
by $ s \Vdash \tau[\alpha:=\beta] $ for all \(\beta \in L\).  However,
it is then difficult to explain the rule for \textbf{run}.

The key idea to solve this problem is to have classifiers range over
the set of finite (and possibly empty) \emph{sequences} of transition
labels and to allow a classifier abstraction \((\alpha)M\) to be
applied to also sequences of classifiers.  Then, \textbf{run} will be
unified to a special case of application of a classifier abstraction
to the \emph{empty} sequence.  More concretely, we change the term
\(M[\beta]\) to \(M[B]\), where \(B\) is a possibly empty sequence of
classifiers (the left rule below).  When \(B\) is empty and \(\tau\)
is \(\langle \tau_0 \rangle^\alpha\) (assuming \(\tau_0\) do not
include \(\alpha\)), the rule (as shown as the right rule below) can
be thought as the typing rule of (another version of) \textbf{run},
since \(\alpha\)-closed code of \(\tau_0\) becomes simply \(\tau_0\)
(without \((\alpha)\cdot\) as in the original \(\lambda^\alpha\)).
\vspace{1ex}
\begin{center}
\infrule{
  \Gamma \vdash^{A} M : (\alpha) \tau
}{
  \Gamma \vdash^{A} M [B] : \tau \Subst{\alpha := B}
} 
\qquad
\infrule{
  \Gamma \vdash^{A} M : (\alpha) \langle \tau_0 \rangle^\alpha
}{
  \Gamma \vdash^{A} M [\varepsilon] : \tau_0
}
\end{center}
\vspace{1ex}

Another benefit of this change is that CSP for closed code (or
embedding of persistent code~\cite{Yuse2006}) can be easily expressed.
For example, if \(x\) is of the type \((\alpha)\langle \textbf{int}
\rangle^\alpha\), then it can be used as code computing an integer at
\emph{different} stages as in, say, \(\langle \cdots (\;\tilde{} \; x
[\alpha]) + 3 \cdots \langle \cdots 4 + (\;\tilde{}\,\;\tilde{}\;
x[\alpha\beta]) \cdots \rangle^\beta \cdots \rangle^\alpha\).  So,
once a programmer obtains closed code, she can use it at any later
stage.  While our calculus \( \sname \) does not have a primitive
  of CSP for all types, we can express CSP in many cases. In
  Section~\ref{sec:compare-alpha}, we discuss this subject in more
  detail.

Correspondingly, the semantics is now given by \(v, \rho; s \Vdash \tau\)
where \( v \) is a valuation for propositional variables
and \(\rho\) is a mapping from classifiers to sequences of
transition labels.  Then, \(v, \rho; s \Vdash \langle \tau
\rangle^\alpha\) is defined by \(v, \rho; s' \Vdash \tau\) where \(s'\)
is reachable from \(s\) through the sequence \(\rho(\alpha)\) of
transitions and \(v, \rho; s \Vdash (\alpha)\tau\) by: \(v,
\rho[A/\alpha]; s \Vdash \tau\) for any sequence \(A\) of labels
(\(\rho[A/\alpha]\) updates the value of \(\alpha\) to be \(A\)).  In
Section~\ref{sec:logic}, we give a formal definition of the Kripke
semantics and show that the proof system, based in the ideas above,
with double negation elimination is sound and complete with respect to it.

\section{The Calculus $ \sname $}
\label{sec:calculus}
In this section, we define the calculus $ \sname $, based on the ideas
described in the previous section: we first define its syntax, type
system, and small-step full reduction semantics and states some basic
properties; then we prove the time-ordered normalization property.
Finally, we give an example of programming in \(\sname\).  We
intentionally make notations for type and term constructors different
from \(\lambda^\alpha\) because their precise meanings are different;
it is also to avoid confusion when we compare the two calculi.

\subsection{Syntax}
Let $ \Gene $ be a countably infinite set of \emph{transition
  variables}, ranged over by $ \alpha $ and $ \beta $.  A
\emph{transition}, denoted by $ A $ and $ B $, is a finite sequence of
transition variables; we write $ \varepsilon $ for the empty sequence
and \(AB\) for the concatenation of the two transitions.  We write $
\Gene^* $ for the set of transitions.  A transition is often called a
\emph{stage}.  We write $ \FMV(A) $ for the set of transition
variables in $ A $, defined by $ \FMV(\alpha_1 \alpha_2 \dots
\alpha_n) = \{ \alpha_i ~|~ 1 \le i \le n \} $.

Let $ \PropVar $ be the set of \emph{base types} (corresponding to
propositional variables), ranged over by $ b $.  The set $ \TypeSet $
of \emph{types}, ranged over by $ \tau $ and $ \sigma $, is defined by
the following grammar:
\[\begin{array}{@{}lr@{}r@{}l}
    \textrm{\it Types}\quad
      & \tau &{}::={}&
        b \OR
	\tau \to \tau \OR
	\Seal{\alpha}{\tau} \OR
	\forall \alpha . \tau
  \end{array}\ .
\]
A type is a base type, a function type, a code type, which corresponds
to \(\langle \cdot \rangle^\alpha\) of \(\lambda^\alpha\), or an
\(\alpha\)-closed type, which corresponds to \((\alpha)\tau\). 
The
transition variable $ \alpha $ of $ \forall \alpha . \tau $ is bound
in $ \tau $.
In what follows, we assume tacit renaming of bound variables in types.
The type constructor $ \Seal{\alpha}{} $ connects tighter than $ \to $
and $\to$ tighter than \(\forall\): for example, $ \Seal{\alpha}{\tau}
\to \sigma $ means $ (\Seal{\alpha}{\tau}) \to \sigma $ and $ \forall
\alpha . \tau \to \sigma $ means $ \forall \alpha . (\tau \to \sigma) $.
We write $ \FMV(\tau) $ for the set of free transition variables,
which is defined in a straightforward manner.

Let $ \VarSet $ be a countably infinite set of \emph{variables},
ranged over by $ x $ and $ y $.  The set of \emph{terms}, ranged over
by $ M $ and $ N $, is defined by the following grammar:
\[
  \begin{array}{@{}lr@{}r@{}l}
    \textit{\it Terms}\quad
    & M &{}::={}&
        x \OR
        M \app M \OR
        \lam x : \tau . M \OR
	\Next{\alpha} M \OR 
	\Prev{\alpha} M \OR
	\Gen{\alpha}{M} \OR
        \Ins{M}{A}
  \end{array}\ .
\]
In addition to the standard \(\lambda\)-terms, there are four more
terms, which correspond to \(\langle M \rangle^\alpha\),
\(\tilde\; M\), \((\alpha)M\), and \(M [\beta]\) of \(\lambda^\alpha\)
(respectively, in the order presented).  Note that, unlike \(\tilde\; M\)
in \(\lambda^\alpha\), the term \(\Prev{\alpha} M\) for
unquote is also annotated.  This annotation is needed because
a single transition variable can be instantiated with a sequence,
in which case a single unquote has to be duplicated accordingly.
The variable $ x $ in $ \lam x : \tau . M
$ and the transition variable $ \alpha $ in $ \Gen{\alpha}{M} $ are
bound in $ M $.  Bound variables are tacitly
renamed to avoid variable capture in substitution.
We write \( \FMV(M) \) for the set of free transition variables,
which is defined in a straightforward manner, e.g.,
\( \FMV(\Next{\alpha}{M}) = \FMV(\Prev{\alpha}{M}) = \FMV(M) \cup \{ \alpha \} \),
\( \FMV(\Lambda \alpha. M) = \FMV(M) - \{ \alpha \} \) and
\( \FMV(M \app A) = \FMV(M) \cup \FMV(A) \).

\subsection{Type System}
\label{sec:type_system}
As mentioned above, a type judgment and variable declarations in a context are
annotated with stages.
A \emph{context} \(\Env\) is a finite set $ \{ x_1 : \tau_1 @ A_1,
\dots, x_n : \tau_n @ A_n \}$, where $x_i$ are distinct variables.  We
often omit braces $\{\}$.
We write $ \FMV(\Env) $ for the set of free transition
variables in $ \Env $, defined by: $ \FMV(\{ x_i : \tau_i @ A_i ~|~ 1
\le i \le n \}) = \bigcup_{i = 1}^{n} (\FMV(\tau_i) \cup \FMV(A_i)) $.

A \emph{type judgment} is of the form $ \Env \vdash^{A} M : \tau $, read
``term $M$ is given type $\tau$ under context $\Env$ at stage \(A\).''
Figure \ref{fig:typing_rules} presents the typing rules to derive type
judgments.
%
%
% Typing Rules
%
\begin{figure}[t]
\begin{multicols}{2}
  \typicallabel{Axiom}
  \RTypeVar\\ \vspace{.5em}
  \RTypeAbs\\ \vspace{.5em}
  \RTypeApp\\ \vspace{3.5em}
  \RTypeNext\\ \vspace{.5em}
  \RTypePrev\\ \vspace{.5em}
  \RTypeGen\\ \vspace{.5em}
  \RTypeIns
\end{multicols}
  \vspace{-1.4em}
  \caption{Typing rules.}
  \label{fig:typing_rules}
\end{figure}
The notation $ \tau \Subst{\alpha := B} $, used in the rule
(\textsc{Ins}), is capture-avoiding substitution of transition \(B\)
for \(\alpha\) in \(\tau\).  When \(\alpha\) in \(\Seal{\alpha}{}\) is
replaced by a transition, we identify $ \Seal{\varepsilon}{\tau} $
with $ \tau $ and $ \Seal{AB}{\tau} $ with $ \Seal{A}{\Seal{B}{\tau}}
$.  For example, \( (\Seal{\alpha}{\forall \alpha . \Seal{\alpha}{b}})
\Subst{\alpha := \varepsilon} = \forall \alpha . \Seal{\alpha}{b}\)
and \( (\forall \alpha . \Seal{\beta}{b}) \Subst{\beta := \alpha
  \alpha} = \forall \alpha' . \Seal{\alpha}{\Seal{\alpha}{b}}\).

The first three rules on the left are mostly standard except for stage
annotations.  The conditions on stage annotations are similar to those
in most multi-stage calculi: The rule (\textsc{Var}) means that
variables can appear only at the stage in which those variables are
declared.  and the rule (\textsc{Abs}) requires the stage of the
parameter to be the same as that of the body and, correpondingly, the
rule (\textsc{App}) requires \(M\) and \(N\) are typed at the same
stage.  The next two rules ($\Next{}{}$) and ($\Prev{}{}$) are for
quoting and unquoting and already explained in the previous section.
The last two rules (\textsc{Gen}) and (\textsc{Ins}) are for
generalization and instantiation of a transition variable,
respectively.  They resemble the introduction and elimination rules of
\(\forall x.A(x)\) in first-order predicate logic: the side condition
of the rule (\textsc{Gen}) ensures that the choice of $ \alpha $ is
independent of the context.  Computationally, this side condition
expresses $ \alpha $-closedness of $ M $, that means $ M $ has no free
variable which has annotation $ \alpha $ in its type or its stage.
This is a weaker form of closedness, which means $ M $ has no free
variable at all.

\subsection{Reduction}
We will introduce full reduction
 $ M \red N $, read ``\(M\) reduces to \(N\)
in one step,''
and prove basic properties including subject reduction,
confluence and strong normalization.
% This is a smallest step of calculation in this calculus.

Before giving the definition of reduction, we define substitution.
Since the calculus has binders for term variables and transition
variables, we need two kinds of substitutions for both kinds of
variables.  Substitution $ M \Subst{x := N} $ for a term variable is
the standard capture-avoiding one, and its definition is omitted here.
Substitution $ M \Subst{\alpha := A} $ of \(A\) for \(\alpha\) is
defined similarly to $ \tau \Subst{\alpha := A} $.
We show representative cases below:
\begin{eqnarray*}
  (\lam x : \tau . M) \Subst{\alpha := A} & = &
    \lam x : (\tau \Subst{\alpha := A}) . (M \Subst{\alpha := A}) \\
  (M \app B) \Subst{\alpha := A} & = &
    (M \Subst{\alpha := A}) \app (B \Subst{\alpha := A}) \\
  (\Next{\beta}{M}) \Subst{\alpha := A} & = &
    \Next{\beta \Subst{\alpha := A}}{(M \Subst{\alpha := A})} \\
  (\Prev{\beta}{M}) \Subst{\alpha := A} & = &
    \Prev{\beta \Subst{\alpha := A}}{(M \Subst{\alpha := A})},
\end{eqnarray*}
%where $ \Next{\alpha_1 \dots \alpha_n}{M} = \Next{\alpha_1}{\cdots
%  \Next{\alpha_n}{M}} $ and $ \Prev{\alpha_1 \dots \alpha_n}{M} =
%\Prev{\alpha_n}{\cdots \Prev{\alpha_1}{M}} $.
where $ \Next{\alpha_1 \dots \alpha_n}{M} $ is an abbreviation for
$ \Next{\alpha_1}{\cdots \Next{\alpha_n}{M}} $, and
$ \Prev{\alpha_1 \dots \alpha_n}{M} $ for
$ \Prev{\alpha_n}{\cdots \Prev{\alpha_1}{M}} $.  In particular, $
(\Next{\alpha}{M}) \Subst{\alpha := \varepsilon} = (\Prev{\alpha}{M})
\Subst{\alpha := \varepsilon} = M \Subst{\alpha := \varepsilon} $.
Note that, when a transition variable in \(\Prev{}{}\) is
replaced, the order of transition variables is reversed, because this
is the inverse operation of $ \Next{}{} $.  This is similar to
the inversion operation in group theory: $ (a_1 a_2 \dots a_n)^{-1} =
a_n^{-1} a_{n-1}^{-1} \dots a_1^{-1} $.

The \emph{reduction relation} \(M \red N\) is the least relation closed under
the following three computation rules
\[
  (\lam x :\tau . M) \app N  \red  M \Subst{x := N} \qquad
  \Prev{\alpha}{(\Next{\alpha}{M})} \red M \qquad
  (\Lambda \alpha . M) \app A  \red  M \Subst{\alpha := A}
\]
and congruence rules, which are omitted here.  In addition to the
standard \(\beta\)-reduction, there are two rules: the second one,
which is already explained previously, cancels quote by unquote and
the last one, instantiation of a transition variable, is similar to
polymorphic function application in System F.  Note that the reduction
is full---reduction occurs under any context in any stage.  This
reduction relation can be thought as (non-deterministic) proof
normalization, which should preserve types, be confluent and strongly
normalizing.  Then, we will define another reduction relation as a
triple $ M \stackrel{T}{\red} N $, with $ T $ standing for the stage
of reduction in Section~\ref{sec:annotated}, as done in $
\lambda^\bigcirc $~\cite{Davies-Atemporallogicappro} and $
\lambda^{\bigcirc\square} $~\cite{Yuse2006}, to prove time-ordered
normalization.

\subsection{Basic Properties}

We will prove three basic properties, namely, subject reduction,
strong normalization and confluence.

The key lemma is, as usual, Substitution Lemma, which says
substitution preserves typing.  We will prove such a property for each
kind of substitution.  We define substitution $ \Env \Subst{\alpha
  := A} $ for contexts as follows:
\begin{equation*}
  \{ x_i : \tau_i @ A_i \} \Subst{\alpha := A}
    = \{ x_i : \tau_i \Subst{\alpha := A} @ A_i \Subst{\alpha := A} \}
\end{equation*}
\begin{lem}[Substitution Lemma]
\label{lem:subst}
\quad
\begin{enumerate}[\em(1)]
\item If $ \Env, x : \sigma @ B \vdash^{A} M : \tau $ and
  $ \Env \vdash^{B} N : \sigma $,
	then $ \Env \vdash^{A} M \Subst{x := N} : \tau $.
\item If $ \Env \vdash^{A} M : \tau $,
  then $ \Env \Subst{\alpha := B} \vdash^{A \Subst{\alpha := B}}
    M \Subst{\alpha := B} : \tau \Subst{\alpha := B} $.
\end{enumerate}
\proof
Easy induction on the typing rules.
We only show main cases.

The proof of (1):
\begin{asparaitem}
\item Case \( M = x \): It is the case that \( \tau = \sigma \) and \(
  A = B \).  So, what we have to show is \( \Env \vdash^B N : \tau \),
  which is already assumed.
  \\
\item Case \( M = M_1 M_2 \): By the typing rules, we know that
\( \Env, x : \sigma @ B \vdash^A M_1 : \tau_0 \to \tau \) and
\( \Env, x : \sigma @ B \vdash^A M_2 : \tau_0 \) for some \( \tau_0 \).
By the induction hypothesis, \( \Env \vdash^A (M_1 \Subst{x := N}) : \tau_0 \to \tau \) and
\( \Env \vdash^A (M_2 \Subst{x := N}) : \tau_0 \). By the rule \textsc{(App)},
we have \( \Env \vdash^A (M_1 \app M_2) \Subst{x := N} : \tau \).
\\
\item Case \( M = \Next{\alpha}{M_0} \): By the typing rules, we know that
\( \tau = \Seal{\alpha}{\tau_0} \) and \( \Env, x : \sigma @ B \vdash^{A \alpha} M_0 : \tau_0 \).
By the induction hypothesis, \( \Env \vdash^{A \alpha} M_0 \Subst{x := N} : \tau_0 \).
So, we obtain \( \Env \vdash^A (\Next{\alpha}{M_0}) \Subst{x := N} : \Seal{\alpha}{\tau_0} \).
\end{asparaitem}
\ 

The proof of (2):
\begin{asparaitem}
\item Case \( M = x \): By the typing rules, \( x : \tau @ A \in \Env \).
So, \( x : (\tau \Subst{\alpha := B}) : (A \Subst{\alpha := B}) \in \Env \Subst{\alpha := B} \).
Therefore, \( \Env \Subst{\alpha := B} \vdash^{A \Subst{\alpha := B}} x : \tau \Subst{\alpha := B} \).
\\
\item Case \( M = \Next{\beta}{M_0} \): By the typing rules, we have
\( \tau = \Seal{\beta}{\tau_0} \) and \( \Env \vdash^{A \beta} M_0 : \tau_0 \).
By the induction hypothesis, \( \Env \Subst{\alpha := B} \vdash^{(A \Subst{\alpha := B}) (\beta \Subst{\alpha := B})} M_0 \Subst{\alpha := B}
: \tau_0 \Subst{\alpha := B} \).  By applying the rule \textsc{\((\Next{}{})\)} as many times as required,
we obtain \( \Env \Subst{\alpha := B} \vdash^{A \Subst{\alpha := B}} \Next{\beta \Subst{\alpha := B}}{(M_0 \Subst{\alpha := B})}
: \Seal{\beta \Subst{\alpha := B}}{(\tau_0 \Subst{\alpha := B})} \).
By the definition of substitution, this is equal to what we have to show.
\\
\item Case \( M = \Gen{\beta}{M_0} \): By the typing rules, we have
\( \tau = \forall \beta. \tau_0 \) and \( \Env \vdash^A M_0 : \tau_0 \).
Moreover we can assume without loss of generality that \( \beta \notin \FMV(\Env) \cup \FMV(A) \cup \FMV(B) \cup \{ \alpha \} \).
By the induction hypothesis, we have \( \Env \Subst{\alpha := B} \vdash^{A \Subst{\alpha := B}} M_0 \Subst{\alpha := B} :
\tau_0 \Subst{\alpha := B} \).  Because \( \beta \notin \FMV(\Env \Subst{\alpha := B}) \cup \FMV(A \Subst{\alpha := B}) \),
we obtain \( \Env \Subst{\alpha := B} \vdash^{A \Subst{\alpha := B}} (\Gen{\beta}{(M_0 \Subst{\alpha := B})}) :
\forall \beta. (\tau_0 \Subst{\alpha := B}) \).  This judgement is equal to what we have to show. \qed
\end{asparaitem}
%\end{proof}
\end{lem}

\begin{thm}[Subject Reduction]
\label{thm:subject-reduction}
If $ \Env \vdash^{A} M : \tau $ and
$ M \red M' $, then
$ \Env \vdash^{A} M' : \tau $.
\proof
  By straightforward induction on the derivation of $ M \red M' $,
  using Substitution Lemma (Lemma~\ref{lem:subst}).  We only show
  three base cases and omit induction steps.

\begin{asparaitem}
\item Case \( M = (\lambda x : \tau_0. M_0) \app M_1 \red M_0 \Subst{x := M_1} = M' \):
Because \( \Env \vdash^A (\lam x : \tau_0. M_0) \app M_1 : \tau \),
we have \( \Env, x : \tau_0 @ A \vdash^A M_0 : \tau \) and \( \Env \vdash^A M_1 : \tau_0 \).
By Substitution Lemma (1), we obtain \( \Env \vdash^A M_0 \Subst{x := M_1} : \tau \).
\\
\item Case \( M = (\Gen{\alpha}{M_0}) \app B \red M_0 \Subst{\alpha := B} = M' \):
Because \( \Env \vdash^A (\Gen{\alpha}{M_0}) \app B : \tau \),
we have \( \alpha \notin \FMV(\Env) \cup \FMV(A) \) and \( \Env \vdash^A M_0 : \tau_0 \) such that \( \tau_0 \Subst{\alpha := B} = \tau \).
By the second statement of Substitution Lemma (Lemma~\ref{lem:subst}), we obtain
\( \Env \Subst{\alpha := B} \vdash^{A \Subst{\alpha := B}} M_0 \Subst{\alpha := B} : \tau_0 \Subst{\alpha := B} \).
Because \( \alpha \notin \FMV(\Env) \cup \FMV(A) \), \( \Env \Subst{\alpha := B} = \Env \) and \( A \Subst{\alpha := B} = A \).
Therefore, \( \Env \vdash^A M_0 \Subst{\alpha := B} : \tau \).
\\
\item Case \( M = \Prev{\alpha}{\Next{\alpha}{M_0}} \red M_0 = M' \):
By replacing \( M \) in the assumption, we obtain
\( \Env \vdash^A \Prev{\alpha}{\Next{\alpha}{M_0}} : \tau \).
So we have \( \Env \vdash^A M_0 : \tau \) as required.
\qed
\end{asparaitem}
\end{thm}

\begin{thm}[Strong Normalization]
\label{thm:strong-normalization}
Let $ M $ be a typable term.
There is no infinite reduction sequence
$ M \red N_1 \red N_2 \red \cdots $.
\begin{proof}
\newcommand{\lambdaarrow}{\lambda^\rightarrow}
  We construct a term $ \natural(M) $ of the simply typed
  \(\lambda\)-calculus ($\lambdaarrow$) as follows:
\begin{eqnarray*}
  \natural(x) & = & x \\
  \natural(\lam x : \tau . N) & = & \lam x . \natural (N) \\
  \natural(N_1 \app N_2) & = & \natural(N_1) \app \natural(N_2) \\
  \natural(\Next{\alpha}{N}) & = & \natural(N) \\
  \natural(\Prev{\alpha}{N}) & = & \natural(N) \\
  \natural(\Lambda \alpha . N) & = & \natural(N) \\
  \natural(N \app A) & = & \natural(N)
\end{eqnarray*}
We can easily prove the following propositions
by induction of the structure of $ M $:
\begin{enumerate}[(1)]
\item If $ M $ is typable in $ \sname $,
then $ \natural(M) $ is typable in $ \lambdaarrow$.
\item Suppose $ M \red M' $.  If this is a \(\beta\)-reduction step,
  then $ \natural(M) \red \natural(M') $ in $\lambdaarrow$.
  Otherwise, $ \natural(M) = \natural(M') $.
\end{enumerate}
Now, assume there exists an infinite reduction sequence from a typable
term $ M $.  It is clear that there are infinitely many
$\beta$-reduction.  By (2), there
exists an infinite reduction sequence from $ \natural(M) $, which 
is typable by (1).  This contradicts the strong normalization property
of \(\lambdaarrow\).
\end{proof}
\end{thm}

The last property we will show is confluence.  We prove this by
using parallel reduction and complete development~\cite{DBLP:journals/iandc/Takahashi95}.  We define the
\emph{parallel reduction} relation \( M \pred N \) as in
Figure~\ref{fig:par-red}.
Notice that, the rule \textsc{(P-\(\Prev{}{}\Next{}{}\))}
allows more than one nested pairs of quoting and unquoting to be cancelled in one step: For example,
\( \Prev{\alpha_1}{\Prev{\alpha_2}{\dots \Prev{\alpha_n}{\Next{\alpha_n}{\dots \Next{\alpha_2}{\Next{\alpha_1}{x}}}}}}
\pred x \).  It is not very standard in the sense that parallel reduction usually
does not allow ``hidden'' redices (that is, redices that appear only
after some other reduction steps) to be contracted in one step.
We require this definition because a transition variable \(\alpha\) can be 
replaced with a sequence \(A\) of transition variables during reduction.  
If \(A\) in \textsc{(P-\(\Prev{}{}\Next{}{}\))} were \(\alpha\),
Lemma~\ref{lem:par-red-prop} (2) below would not hold any longer.

\begin{figure}[t]
\begin{multicols}{2}
  \typicallabel{Axiom}
\infrule[P-Var]{
  \mathstrut
}{
  x \pred x
}
\\ \vspace{.5em}
\infrule[P-Lam]{
  M \pred N
}{
  \lam x:\tau.M \pred \lam x:\tau.N
}
\\ \vspace{.5em}
\infrule[P-App]{
  M_1 \pred N_1 \qquad M_2 \pred N_2
}{
  M_1 \app M_2 \pred N_1 \app N_2
}
\\ \vspace{.5em}
\infrule[P-Beta]{
  M_1 \pred N_1 \qquad M_2 \pred N_2
}{
  (\lam x:\tau.M_2) \app M_2 \pred N_1 \Subst{x := N_2}
}
\\ \vspace{.5em}
\infrule[P-Gen]{
  M \pred N
}{
  \Lambda \alpha. M \pred \Lambda \alpha. N
}
\\ \vspace{.5em}
\infrule[P-TApp]{
  M \pred N
}{
  M \app A \pred N \app A
}
\\ \vspace{.5em}
\infrule[P-Ins]{
  M \pred N
}{
  (\Lambda \alpha. M) \app A \pred N \Subst{\alpha := A}
}
\\ \vspace{.5em}
\infrule[P-\(\Next{}{}\)]{
  M \pred N
}{
  \Next{\alpha}{M} \pred \Next{\alpha}{N}
}
\\ \vspace{.5em}
\infrule[P-\(\Prev{}{}\)]{
  M \pred N
}{
  \Prev{\alpha}{M} \pred \Prev{\alpha}{N}
}
\\ \vspace{.5em}
\infrule[P-\(\Prev{}{}\Next{}{}\)]{
  M \pred N
}{
  \Prev{A}{\Next{A}{M}} \pred N
}
\end{multicols}
  \caption{Rules for Parallel Reduction.}
  \label{fig:par-red}
\end{figure}

The following lemma relates the reduction relation and the parallel
reduction relation.
\begin{lem}
 \label{lem:par-red-basic}
 \( (\red) \subseteq (\pred) \subseteq (\red^*) \)
\end{lem}
\begin{proof}
 \( (\red) \subseteq (\pred) \)
 can be shown by induction on the derivation \( M \red N \).
 We can prove \( (\pred) \subseteq (\red^*) \) by induction on the
 structure of \( M \pred N \).
\end{proof}

Thanks to this lemma, we know that confluence of \( \red \) is
equivalent to confluence of \( \pred \).  We prove confluence of \(
\pred \) by showing that \(\pred\) enjoys the diamond property.  The
following properties of parallel reduction are useful.

\begin{lem}
 \label{lem:par-red-prop}\ 
 \begin{enumerate}[\em(1)]
  \item If \( M_1 \pred N_1 \) and \( M_2 \pred N_2 \), then
	\( M_1 \Subst{x := M_2} \pred N_1 \Subst{x := N_2} \).
  \item If \( M \pred N \), then
	\( M \Subst{\alpha := A} \pred N \Subst{\alpha := A} \).
 \end{enumerate}
\end{lem}
\begin{proof}
 Easy induction on the structure of the derivation
 \( M_1 \pred N_1 \) and \( M \pred N \), respectively.
\end{proof}

Now, we define the notion of complete development and show its key property.
The \emph{complete development} \( \compdev{M} \) of \( M \) is defined by
 induction as in Figure~\ref{fig:def-compdev}.

\begin{figure}[t]
\begin{eqnarray*}
 \compdev{x} & = & x \\
 \compdev{(\lam x:\tau. M)} & = & \lam x:\tau. \compdev{M} \\
 \compdev{((\lam x:\tau. M) \app N)} & = &
  \compdev{M} \Subst{x := \compdev{N}} \\
 \compdev{(M \app N)} & = & \compdev{M} \app \compdev{N}
  \textrm{ \qquad\qquad(if \( M \neq \lam x:\tau. M' \))} \\
 \compdev{(\Next{\alpha}{M})} & = & \Next{\alpha}{\compdev{M}} \\
 \compdev{(\Prev{A}{\Next{A}{M}})} & = & \compdev{M} \\
 \compdev{(\Prev{\alpha}{M})} & = &
  \Prev{\alpha}{\compdev{M}}
  \textrm{ \qquad\qquad(if \( M \neq \Prev{A}{\Next{A \alpha}{M'}}\))} \\
 \compdev{(\Gen{\alpha}{M})} & = & \Gen{\alpha}{\compdev{M}} \\
 \compdev{((\Gen{\alpha}{M}) \app A)} & = &
  \compdev{M} \Subst{\alpha := A} \\
 \compdev{(M \app A)} & = & \compdev{M} \app A
  \textrm{ \qquad\qquad\hspace{.5em}(if \( M \neq \Gen{\alpha}{M'} \))}
\end{eqnarray*}
 \caption{Definition of Complete Development.}
 \label{fig:def-compdev}
\end{figure}

\begin{lem}\label{lem:complete-development}
 If \( M \pred N \), then \( N \pred \compdev{M} \).
\end{lem}
\proof
By induction on the derivation of \( M \pred N \) with case analysis on
\( M \).  We show only interesting cases.

\begin{asparaitem}
\item Case \textsc{(P-Lam)}: We have \( M = \lam x : \tau. M_0 \) and
\( N = \lam x : \tau. N_0 \) with \( M_0 \pred N_0 \).
By the induction hypothesis, we have \( N_0 \pred
\compdev{M_0} \).  So, by applying \textsc{(P-Lam)} rule, we obtain \( \lam
 x:\tau. N_0 \pred \lam x:\tau. \compdev{M_0} \), as required.
\\
\item Case \textsc{(P-TApp)}: We have
 \( M = M_0 \app A \) and \( N = N_0 \app A \) with \( M_0 \pred N_0 \).  There are
 two subcases.

 Assume \( M_0 \neq \Gen{\alpha}{M_1} \).
 By the induction hypothesis, we have \( N_0 \pred \compdev{M_0} \).  By applying
 \textsc{(P-TApp)} rule, we obtain \( N_0 \app A \pred \compdev{M_0} \app A
 \).
 
 Assume \( M_0 = \Gen{\alpha}{M_1} \) for some \( M_1 \).
 By the induction hypothesis, we have \( N_0 \pred \compdev{M_0} =
 \Gen{\alpha}{\compdev{M_1}} \).  By the definition of parallel reduction,
 we have that \( N_0 = \Gen{\alpha}{N_1} \) for some \( N_1 \) and
\[
% \Gen{\alpha}{M_1} \pred \Gen{\alpha}{N_1} \pred \Gen{\alpha}{\compdev{M_1}},
N_1 \pred \compdev{M_1}.
\]
%which implies \( N_0 \pred \compdev{M_0} \).
So, by applying \textsc{(P-TIns)} rule, we obtain \( (\Gen{\alpha}{N_1})
 \app A \pred \compdev{M_1} \Subst{\alpha := A}  = \compdev{M}\).
\\
\item Case \textsc{(P-Ins)}: We have \( M = (\Gen{\alpha}{M_0}) \app A \)
and \( N = N_0 \Subst{\alpha := A} \) with \( M_0 \pred N_0 \).
 By the induction hypothesis, we have \( N_0 \pred
 \compdev{M_0} \).  By applying Lemma~\ref{lem:par-red-prop}, we obtain \(
 N_0 \Subst{\alpha := A} \pred \compdev{M_0} \Subst{\alpha := A}\)
\qed
\end{asparaitem}

It is easy to show diamond property of \( \pred \) by using Lemma~\ref{lem:complete-development}.
\begin{lem}
 \label{lem:diamond}
 If \( M \pred N_1 \) and \( M \pred N_2 \), then there exists \( N \)
 which satisfies \( N_1 \pred N \) and \( N_2 \pred N \).
\end{lem}
\begin{proof}
 Choose \( \compdev{M} \) as \( N \) and use the previous lemma.
\end{proof}

\begin{thm}[Confluence]
\label{thm:confluence}
If $ M \red^{*} N_1 $ and $ M \red^{*} N_2 $,
then there exists $ N $ such that
$ N_1 \red^{*} N $ and $ N_2 \red^{*} N $.
\begin{proof}
 By Lemma~\ref{lem:par-red-basic}, we have
\[
 (\red) \subseteq (\pred) \subseteq (\red^*).
\]
So we obtain
\[
 ( \red^* ) \subseteq (\pred^*) \subseteq (\red^*),
\]
 which implies \( (\red^*) = (\pred^*) \).  Therefore what we should show
 is confluence of \( \pred \), which is an easy consequence of Lemma~\ref{lem:diamond}.
\end{proof}
\end{thm}

\subsection{Annotated Reduction and Time-Ordered Normalization}
\label{sec:annotated}
We introduce the notion of stages into reduction and prove the
property called \emph{time-ordered
  normalization}~\cite{Davies-Atemporallogicappro,Yuse2006}.
Intuitively, it says that normalization can be done in the increasing
order of stages and does not need to `go back' to earlier stages.  In
other words, once all redices at some stage are contracted, subsequent
reductions never yield a new redex at that stage.  To state
time-ordered normalization formally, we first introduce the notion of
\emph{paths} from one stage to another and a new reduction relation,
annotated with paths to represent the stage at which reduction occurs.

%  can be normalized as
% \( M = N_0 \red N_1 \red \dots \red N_m \) where the sequence of stages of the redexes \( N_i \)
% is in increasing order.

A path represents how the stage of a subterm is reached from the stage
of a given term.  For example, if \( \Gamma \vdash^{\alpha} M \) and
\( \Gamma' \vdash^{\alpha\beta} N \) for a subterm \( N \) of \( M \),
then we say the path from (the stage of) \( M \) to (that of) \( N \)
is \( \beta \).  The stage of a subterm may not be able to be
expressed by a transition (a sequence of transition variables),
however: For example, consider the path from \( \Prev{\alpha}{M} \) to
\(M\).  We introduce \emph{formal inverses} \(\alpha^{-1}\) to deal
with such cases: the path from the stage of \(\Prev{\alpha} M\) to
that of \(M\) is represented by \(\alpha^{-1}\).  Similarly, the path
from \(\Next{\alpha}\Prev{\beta} M\) to \(M\) will be
\(\alpha\beta^{-1}\).

% For the technical convenience,
% we introduce a notion \emph{relative path},
% which is a path from the current stage to the stage of a subterm.
% The contrary notion of the relative path is the \emph{absolute path},
% which indicate the stage of a term.
% For example, if \( \Gamma \vdash^A M \) and \( \Gamma \vdash^{AB} N \) for
% a subterm \( N \) of \( M \),
% we say the absolute path of \( N \) is \( A B \) and
% the relative path of \( N \) from \( M \) is \( B \).

Formally, the set of paths, ranged over by \(T\) and \(U\), is the
free group generated by the set of transition variables \( \Gene \).
% In other words,
% the set of relative path is the set of sequences of
% \( \alpha \) and \( \alpha^{-1} \), where \( \alpha \in \Gene \),
% with the equation \( \alpha \alpha^{-1} = \alpha^{-1} \alpha = \varepsilon \)
In other words, a path is a finite sequence $ \xi_1 \xi_2 \dots \xi_n
$, where $ \xi_i = \alpha $ or $ \alpha^{-1} $, such that it includes
no subsequence of the form \(\alpha\alpha^{-1}\) or
\(\alpha^{-1}\alpha\) for any \(\alpha\).
% Especially, any transition is a relative path.
We define $ (\alpha^{-1})^{-1} = \alpha $ and
\[
  (\xi_1 \cdots \xi_n) \cdot (\xi_{n+1} \dots \xi_m) =
    \left\{
      \begin{array}{ll}
          (\xi_1 \cdots \xi_{n-1}) \cdot (\xi_{n+2} \dots \xi_m) & \textrm{(if $ \xi_n = \xi_{n+1}^{-1} $)} \\
          (\xi_1 \cdots \xi_{n-1}) \cdot (\xi_n \xi_{n+1} \dots \xi_m) & \textrm{(if $ \xi_n \neq \xi_{n+1}^{-1} $).}
      \end{array}
    \right.
\]
The empty sequence \(\varepsilon\) is the unit element for
the operation \(T\cdot U\).  We simply write $ T U $ for $T \cdot U$.  We define $
(\xi_1 \xi_2 \dots \xi_n)^{-1} = \xi_n^{-1} \xi_{n - 1}^{-1} \dots
\xi_1^{-1} $.
% The formal negation \( \alpha^{-1} \) is needed to describe
% relative path of un-quoting:
% for example, the relative path of \( x \) from \( \Prev{\alpha}{x} \)
% is \( \alpha^{-1} \).

% Let \( T \) and \( U \) range over the set of relative paths.
We say a path \( T \) is \emph{positive} if %(the canonical form of)
\( T \) does not contain formal inverses, in other words, the
canonical form of \( T \) is in \( \Gene^* \).  We can naturally
identify the positive paths with transitions and use metavariables \( A
\) and \( B \) for positive paths.  We write $ T \le U $ when there
exists a positive path $ A $ which satisfies $ T A = U $.  Clearly, $
\varepsilon \le T $ if and only if $ T $ is positive.

The \emph{annotated reduction} relation is a triple of the form $ M
\stackrel{T}{\red} N $, where $ M $ and $ N $ are terms and $ T $ is a
path from the stage of $ M $ to that of its redex---more precisely,
that of the constructor destructed by the reduction, since the stage
of a redex and that of its constructor may be different as in \(
\Next{\alpha}{} \) in redex \( \Prev{\alpha}{\Next{\alpha}{M}}\).
% that of its redex.
The definition of the annotated reduction, presented in
Figure~\ref{fig:annotated_reduction}, is mostly straightforward.  For
example, \(\alpha^{-1}\) is given to \(
\Prev{\alpha}{\Next{\alpha}{M}} \red M \) (the rule \(
\textsc{(AR-Quote)} \)), because the path to the constructor \(
\Next{\alpha}{} \) is \( \alpha^{-1} \).  As for the rule
\textsc{(AR-\( \Next{}{} \))}, the path from \( M \) to the
constructor destructed by the reduction is \( T \) and the path from
\( \Next{\alpha}{M} \) to \( M \) is \( \alpha \), hence the path from
\( \Next{\alpha}{M} \) to the constructor is given by their
concatenation \( \alpha T \).  The rule \textsc{(AR-\( \Prev{}{} \))}
is similar.  The rule \textsc{(AR-Gen)} is the most interesting.
First of all, \(\alpha\) is bound here, so, we cannot propagate \(T\)
in the premise to the conclusion to prevent \(\alpha\) from escaping
its scope.  We have found that replacing \(\alpha\) with
\(\varepsilon\), which is the earliest possible stage, is a reasonable
choice, especially for time-ordered normalization.
% because \(\alpha\) can be instantiated by any transition.

Annotated reduction is closely related to reduction defined in the previous section.
It is easy to see that \( M \red N \) if and only if there exists \( T \) such that
\( M \stackrel{T}{\red} N \).  Furthermore, such \( T \) is unique.
% , replacing to the following simpler
% rule is problematic because \( T \) may contain \( \alpha \) or \(
% \alpha^{-1} \), which is invisible outside its scope \( M \).
% \begin{center}
% \infrule{
%   \vdash^A M \stackrel{T}{\red} N
% }{
%   \vdash^A \Lambda \alpha. M \stackrel{T}{\red} \Lambda \alpha. N
% }
% \end{center}
% So we should substitute \( \alpha \) in \( T \) to some other transition
% and we choose it \( \varepsilon \).
\begin{figure}[t]
  \begin{multicols}{2}
%  \centering
\leavevmode
  \RRBeta \\
  \RRIns \\
  \RRPrevNext \\
  \RRCAbs \\
  \RRCAppOne \\
  \RRCAppTwo \\
  \RRCNext \\
  \RRCPrev \\
  \RRCGen \\
  \RRCIns
  \end{multicols}
  \caption{Annotated Reduction.}
  \label{fig:annotated_reduction}
\end{figure}

The next theorem shows that any reduction occurs indeed at a positive stage.
\begin{thm}
If $ \Env \vdash^{A} M : \tau $ and
$ M \stackrel{T}{\red} N $, then  $ \varepsilon \le A T $.
\proof
Easy induction on the structure of $ M \stackrel{T}{\red} N $.
\qed
\end{thm}

We say $ M $ is \emph{\(T\)-normal} when there are no $ U \le T $ and $ N $
such that $ M \stackrel{U}{\red} N $.  Then, we can state time-ordered
normalization as follows:
\begin{thm}[Time Ordered Normalization]
\label{thm:TON}
Let $ M $ be a typable term. If $ M $ is $ T $-normal and 
$ M \red^{*} N $, then $ N $ is $T$-normal.
\end{thm}
\proof
See Appendix~\ref{sec:proof-of-TON}.
\qed

As its corollary, we know that for any reduction to a normal form from a
typable term \( M \) is ``rearranged'' according to an increasing order
between stages.  Moreover, this increasing order can be any total order
that respects \(\le\), i.e., includes \(\le\) as a subset.
% there is a reduction sequence reaching the normal form,
% in which the stages of reductions are in ``increasing order''.
% But there are incomparable stages like \( \alpha \) and \( \beta \),
% so how sholud we order them.  The answer is in arbitrarily order.
% More precisely, for any total ordering \( \preceq \) which respects \( \le \),
% i.e., \( \preceq \) satisfies that if \( T \le U \), then \( T \preceq U \) for any \( T \) and \( U \),
% there is a reduction sequence, in which stages of reductions are in increasing ordering
% with respect to \( \preceq \).
\begin{cor}
Let \( M \) be a typable term and
\( \preceq \) be a total order that respects \( \le \).
Then, there is a reduction sequence
\( M \stackrel{T_1}{\red} N_1 \stackrel{T_2}{\red} \dots
\stackrel{T_n}{\red} N_n \), which satisfies
\( T_1 \preceq T_2 \preceq \dots \preceq T_n \) and
\( N_n \) is a normal form.
\qed
\end{cor}

\subsection{Programming in $ \sname $}
\label{sec:power}
\newcommand{\clet}{\ensuremath{\mathop{\textbf{let}}}}
\newcommand{\cfix}{\ensuremath{\mathop{\textbf{fix}}}}
\newcommand{\cint}{\ensuremath{\mathop{\textbf{int}}}}
\newcommand{\cif}{\ensuremath{\mathop{\textbf{if}}}}
\newcommand{\cthen}{\ensuremath{\mathop{\textbf{then}}}}
\newcommand{\celse}{\ensuremath{\mathop{\textbf{else}}}}
\definecolor{gray96}{gray}{.9}
\newcommand{\gb}[1]{\!\colorbox{gray96}{\!\(#1\)\!}\!}
\newcommand{\wb}[1]{\!\colorbox{white}{\!\(#1\)\!}\!}
We give an example of programming in \(\sname\).  The example is the
power function, which is a classical example in multi-stage calculi
and partial evaluation.  We augment \(\sname\) with integers,
Booleans, arithmetic and comparison operators, \cif-\cthen-\celse, a
fixed point operator \cfix, and \clet.  In the next section, we will
formalize such a language (without \clet) as \miniML{} and study its
evaluation in more detail.  For readability, we often omit type
annotations and put terms under quotation in shaded boxes.

We start with the ordinary power function without staging.
\[
\begin{array}{l}
  \clet \texttt{power}_0 \colon \cint \to \cint \to \cint \\
  \qquad
  = \cfix f .\; %\colon \!\!\! \cint \to \cint \to \cint \! . \;\;
    \lam n . \; % \colon \!\!\! \cint \! . \;\; 
    \lam x . %\colon \!\!\! \cint \! .
    \cif\; n = 0 \; \cthen \; 1 \; \celse \; x * (f \app (n - 1) \app x)
\end{array}
\]
Our purpose is to get a code generator $ \texttt{power}_{\forall} $
that takes the exponent \(n\) and returns (closed, hence runnable)
code of $\lam x. x * x * \dots x * 1$, which computes \(x^n\) without
recursion.  Here, we follow the construction of code generators in the
previous
work~\cite{BenaissaMoggiTahaSheard99IMLA,MoggiTahaBenaissaSheard99ESOP}.

First, we construct a code manipulator $ \texttt{power}_1 : \cint \to
\Seal{\alpha}{\cint} \to \Seal{\alpha}{\cint} $, which takes an
integer \(n\) and a piece of integer code and then outputs a piece of
code which concatenates the input code by ``\(*\)'' \(n\) times.  It can be
obtained by changing type annotation and introducing quasiquotation.
\[
\begin{array}{l}
  \clet \texttt{power}_1
    \colon \cint \to \Seal{\alpha}{\cint} \to \Seal{\alpha}{\cint} \\
  \qquad
  = \cfix f . \; % \colon \!\!\! \cint \to \Seal{\alpha}{\cint} \to \Seal{\alpha}{\cint} . \;\;
    \lam n . \; % \colon \!\! \cint . \; % \;\;
    \lam x \colon \! \Seal{\alpha}{\cint} .\\
  \qquad \qquad
    \cif\; n = 0 \; \cthen \; (\Next{\alpha}{\gb{1}})
    \; \celse \; \Next{\alpha}{\gb{((\Prev{\alpha}{\wb{x}}) * 
      (\Prev{\alpha}{\wb{f \app (n - 1) \app x}}))}}
\end{array}
\]
Then, from $ \texttt{power}_1 $, we can construct a code generator $ \texttt{power}_\alpha $ of type
$ \cint \to \Seal{\alpha}{(\cint \to \cint)} $,
which means it takes an integer and returns code of a function.
\[
\begin{array}{l}
  \clet \texttt{power}_{\alpha}
    \colon \! \cint \to \Seal{\alpha}{(\cint \to \cint)} \\
  \qquad
  = \lam n . \; %\colon \!\!\! \cint \! . \;
    \Next{\alpha}{
    \gb{\lam x \colon \!\!\! \cint \! . \; \Prev{\alpha}{\wb{
      (\texttt{power}_1 \app n \app (\Next{\alpha}{\gb{x}}))}}\;}}
\end{array}
\]
It indeed behaves as a code generator: for example, $
\texttt{power}_{\alpha} \app 3 $ would evaluate to the code value
$
\Next{\alpha}{\gb{ \lam x \colon \! \cint . x * (x * (x * 1))}} $.

This construction is independent of the choice of the stage $ \alpha
$.  So, by abstracting \(\alpha\) at appropriate places in
\(\texttt{power}_1\) and 
\(\texttt{power}_{\alpha}\), we can obtain the desired code generator,
whose return type is a closed code type $ \forall \gamma
. \Seal{\gamma}{(\cint \to \cint)} $.
\[
\begin{array}{l}
  \clet \texttt{power}_2
    \colon \forall \beta. \cint \to \Seal{\beta}{\cint} \to \Seal{\beta}{\cint} \\
  \qquad
  = \Lambda \beta. \cfix f . \;
    \lam n . \;
    \lam x \colon \! \Seal{\beta}{\cint} .\\
  \qquad \qquad
    \cif\; n = 0 \; \cthen \; (\Next{\beta}{\gb{1}})
    \; \celse \; \Next{\beta}{\gb{((\Prev{\beta}{\wb{x}}) * 
      (\Prev{\beta}{\wb{f \app (n - 1) \app x}}))}} \\
  \clet \texttt{power}_{\forall}
    \colon \!\! \cint \to \forall \gamma. \;
      \Seal{\gamma}{(\cint \to \cint)} \\
  \qquad
  = \lam n . \; %\colon \!\!\! \cint \! . \;
    \Lambda \gamma. \Next{\gamma}{
    \gb{\lam x \colon \!\!\! \cint \! . \; \Prev{\gamma}{\wb{
      (\texttt{power}_2 \app \gamma \app n \app (\Next{\gamma}{\gb{x}}))}}\;}}
\end{array}
\]

The output from $ \texttt{power}_{\forall} $ is usable in any stage.
For example, if we want code of a cube function at the later stage,
say $ A $, then we write $ \texttt{power}_{\forall} \app 3 \app A $.
In particular, when $ A $ is the empty sequence $ \varepsilon $, $
\texttt{power}_{\forall} \app 3 \app \varepsilon : \cint \to \cint $
evaluates to a function closure which computes \(x * x * x * 1\) from
the input \(x\).  The former corresponds to CSP (of closed code) and
the latter to \textbf{run}.

\section{\miniML}
\label{sec:miniml}
\newcommand{\svdash}{\vdash_s}
\newcommand{\ered}[3]{\vdash_e^{#1} {#2} \Downarrow {#3}}

We extend \( \sname \) and define an ML-like functional language
\miniML, which has, in addition to the features of \(\sname\),
integers, arithmetic and comparison operations, Booleans, conditional
expressions, and the (call-by-value) fixed-point combinator
\textbf{fix}.  We define the type system and big-step evaluation
semantics for \miniML{} and prove type soundness.  In this semantics,
bindings of transition variables have to be maintained at run time.
So, we investigate a fragment of \miniML, in which programs can be
executed by mostly forgetting information on transition variables.  We
give another type system, which identifies such a fragment, and
erasure translation, which removes transitions from terms, and
alternative evaluation semantics for erased terms.  Then, we prove the
erasure property, which says program executions before and after
erasure agree.

\subsection{Syntax and Type System}
The syntax of types and terms of \miniML{} is defined as follows,
where \( n \) and \(bv\) are metavariables ranging over integers
and Boolean constants \textbf{true} and \textbf{false}.
\\
\begin{center}
\(\begin{array}{@{}lr@{}r@{}l}
    \textrm{\it Types}\quad
      & \tau &{}::={}&
        \textbf{int} \OR
        \textbf{bool} \OR
	\tau \to \tau \OR
	\Seal{\alpha}{\tau} \OR
	\forall \alpha .\tau \\
    \textrm{\it Terms}\quad
      & M &{}::={}&
        x \OR
        n \OR
        bv \OR
        M = M \OR
        M + M \OR
        M - M \OR
        M * M \\
% \OR
%         \textbf{true} \OR
%         \textbf{false} \\
      & & \OR &
        \eIf{M}{M}{M}\OR
        \textbf{fix}\ f : \tau \to \sigma. M \OR
        \lam x : \tau . M \OR
        M M \\
      & & \OR &
        \Next{\alpha}{M} \OR
        \Prev{\alpha}{M} \OR
        \Gen{\alpha}{M} \OR
        \Ins{M}{A}
  \end{array}
\)
\end{center}

The type system is given as a straightforward extension of that of \(\sname\).
We show typing rules for the additional constructs.
\begin{center}
\infrule[IntC]{
  n \in \mathbb{Z}
}{
  \Gamma \vdash^A n : \textbf{int}
}
\qquad
\infrule[BoolC]{
  bv \in \{ \textbf{true}, \textbf{false} \}
}{
  \Gamma \vdash^A bv : \textbf{bool}
}
\\[1ex]
\infrule[Eq]{
  \Gamma \vdash^A M : \textbf{int} \AND
  \Gamma \vdash^A N : \textbf{int}
}{
  \Gamma \vdash^A M = N : \textbf{bool}
}
\\[1ex]
\infrule[IntOp]{
  \Gamma \vdash^A M : \textbf{int} \AND
  \Gamma \vdash^A N : \textbf{int} \AND
  \Diamond \in \{ +, -, * \}
}{
  \Gamma \vdash^A M \mathop{\Diamond} N : \textbf{int}
}
\\[1ex]
\infrule[If]{
  \Gamma \vdash^A M : \textbf{bool} \AND
  \Gamma \vdash^A N_1 : \tau \AND
  \Gamma \vdash^A N_2 : \tau
}{
  \Gamma \vdash^A \eIf{M}{N_1}{N_2}
}
\\[1ex]
\infrule[Fix]{
  \Gamma, f : \tau \to \sigma @ A \vdash^A M : \tau \to \sigma
}{
  \Gamma \vdash^A \textbf{fix}\ f : \tau \to \sigma. M : \tau \to \sigma
}
\end{center}

We use the same notations for term substitution \( M \Subst{x := N}
\), and transition substitution \( \tau \Subst{\alpha := A} \) and \(
M \Subst{\alpha := A} \) and other auxiliary notions, which can be similarly defined.
It is easy to prove that \miniML{} also enjoys Substitution Lemma.
\begin{lem}[Substitution Lemma]
\label{lem:ml-subst}\ 
\begin{enumerate}[\em(1)]
\item If \( \Gamma, x : \sigma @ B \vdash^A M : \tau \) and
\( \Gamma \vdash^B N : \sigma \),
then \( \Gamma \vdash^A M \Subst{x := N} : \tau \).
\item If \( \Gamma \vdash^A M : \tau \),
then, % for any transition \( B \) and transition variable \( \alpha \),
\( \Gamma \Subst{\alpha := B} \vdash^{A \Subst{\alpha := B}}
M \Subst{\alpha := B} : \tau \Subst{\alpha := B} \).
\end{enumerate}
\begin{proof}
The proof is essentially the same as that of Substitution Lemma for \( \sname \)
(Lemma~\ref{lem:subst}).
\end{proof}
\end{lem}

\subsection{Evaluation and Type Soundness}
\label{sec:evaluation}
Now, we give a big-step semantics and prove that the execution of a
well-typed program is properly divided into stages.  The judgment
has the form $ \bred{A}{M}{R} $, read ``evaluating term \(M\) at stage
\(A\) yields result \(R\),'' where \(R\) is either \(\err\), which
stands for a run-time error, or a value \(v\), defined below.  \emph{Values}
are given via a family of sets \(V^{A}\) indexed by transitions, that
is, stages.
The family \(V^{A}\) is defined by the following grammar:
\[
  \begin{array}{@{}l@{}r@{}l}
    v^{\varepsilon} \in V^{\varepsilon} &{}::={}&
      n \OR \textbf{true} \OR \textbf{false} \OR
      \lam x : \tau . M \OR
      \Next{\alpha}{V^{\alpha}} \OR
      \Lambda \alpha . V^{\varepsilon} \\
    v^{A} \in V^{A} \; ( A \neq \varepsilon) &{}::={}&
      x \OR n \OR \textbf{true} \OR \textbf{false} \OR
      \lam x : \tau . V^{A} \OR
      \textbf{fix}\ f : \tau \to \sigma. V^A \\
    & \OR &
      V^{A} \app V^{A} \OR
      \Next{\alpha}{V^{A \alpha}} \OR
      \Lambda \alpha . V^{A} \OR
      \Ins{V^{A}}{B} \\
    & \OR &
      \Prev{\alpha}{V^{A'}}
        \quad \text{\rm (if $ A' \alpha = A $ and $ A' \neq \varepsilon $)}
  \end{array}
\]
The index \( A \) represents the current stage in which a value is
typed.  So, the index changes under quoting and unquoting.  Note that
a value at a higher stage (that is, under quotation) include free
variables, applications and instantiation since computation is
suspended.  For example, \(x \app y \in V^\alpha\) and so
\(\Next{\alpha}{x \app y} \in V^\varepsilon\).

Figure~\ref{fig:evaluation} shows the evaluation rules.  Notice that
metavariables \(M\) or \(N\) for terms (not values) are used on the
right side of \(\Downarrow\), since it is not immediately clear that 
a result is really a value of a proper form (or $\err$)---we will
prove such a property as a theorem.  The evaluation is left-to-right
and call-by-value.  The rules in Figure~\ref{fig:evaluation}(1) are
for ordinary evaluation.  The rule for \( \Prev{\alpha}{M} \) means
that quote is canceled by unquote; since the resulting term \(M'\)
belongs to the stage \(\alpha\) (inside quotation), \(\alpha\) is
attached to the conclusion.  As seen in the rule for \(\Lambda
\alpha.M\), \(\Lambda\) does \emph{not} delay the evaluation of the
body.  The rule about instantiation of a transition abstraction is
straightforward.  The rules for stages later than \(\varepsilon\),
which are in Figure~\ref{fig:evaluation}(2), are all similar: since
the term to be evaluated is inside quotation, each term constructor is
left as it is and only subterms of stage \(\varepsilon\) will be
evaluated.  We also need rules for handling erroneous terms, such as:
\begin{center}
\infrule{
  \bred{\varepsilon}{M}{M'} \AND
  M' \notin \mathbb{Z} \AND
  \Diamond \in \{ +, -, *, = \}
}{
  \bred{\varepsilon}{M \Diamond N}{\err}
}
\hfil
\infrule{
  \bred{A}{M}{\err} \AND
  \Diamond \in \{ +, -, *, = \}
}{
  \bred{A}{M \Diamond N}{\err}
}
\end{center}
They are shown in Appendix~\ref{sec:complete-rules}.

\begin{figure}[p]
  \centering
%(for the case $ A = \varepsilon $) \hfill
\infrule{
  \bred{\varepsilon}{M}{m} \AND
  \bred{\varepsilon}{N}{n} \AND
  m = n
}{
  \bred{\varepsilon}{M = N}{\textbf{true}}
}
\qquad
\infrule{
  \bred{\varepsilon}{M}{m} \AND
  \bred{\varepsilon}{N}{n} \AND
  m \neq n
}{
  \bred{\varepsilon}{M = N}{\textbf{false}}
}
\\[2ex]
  \EOprE \\[2ex]
  \EIfTrue \qquad
  \EIfFalse \\[2ex]
  \ELamE \qquad
  \EAppE \\[2ex]
  \ENextE \qquad
  \EPrevE \qquad
  \EGenE \\[2ex]
  \EInsE \qquad
  \EConstE \\[2ex]
  \EFixE \\[1ex]
\mbox{(1) Rules for ordinary evaluation.}\\[5ex]
%\vspace{1em}
%(for the case $ A \neq \varepsilon $) \hfill
  \EOprNE \\[2ex]
  \EIfNE \\[2ex]
  \EVarNE \qquad
  \ELamNE \qquad
  \EAppNE \\[2ex]
  \ENextNE \qquad
  \EPrevNE \qquad
  \EInsNE \\[2ex]
  \EGenNE \qquad
  \EConstNE \\[2ex]
  \EFixNE \\[1ex]
\mbox{(2) Rules for evaluation inside quotation. Here, \( A \neq \varepsilon \).}\\[3ex]
%\vspace{1em}
%(stage independent) \hfill
  \caption{Big-Step Semantics of \miniML.}
  \label{fig:evaluation}
\end{figure}

We show a few properties of the big-step semantics.  The first theorem
says that evaluation is deterministic.
% , and the second says that values
% evaluates to themselves and, in particular, a term, which is the body of
% a quoted value, evaluates to the same term under any quotation.
%
\begin{thm}
If \(\bred{A}{M}{R}\) and \(\bred{A}{M}{R'}\), then \(R = R'\).
\proof
By straightforward induction on the derivation of \(\bred{A}{M}{R}\).
\qed
\end{thm}
%
% \begin{thm}
%   For any \(A\) and \(v^A\), \(\bred{A}{v^A}{v^A}\).  Moreover, for
%   any \(B\neq \varepsilon\), \(\alpha\), and \(v^\alpha\),
%   \(\bred{B}{v^\alpha}{v^\alpha}\).
% \proof
% By induction on \(v^A\).
% \end{thm}
%
The second theorem below
says that, unless the result is \(\err\), the result must be a value
even though the rules do not say it is the case.
\begin{thm}
Suppose \(\bred{A}{M}{R} \).
Then, either $ R = \err  $ or 
$ R \in V^A$.
\begin{proof}
By easy induction on the derivation \( \bred{A}{M}{R} \).
\end{proof}
\end{thm}

The last property is type soundness and its corollary that if a
well-typed program of a code type yields a result, then the result is
a quoted term, whose body is also typable at stage \(\varepsilon\).
Unlike a usual setting where only closed terms are considered
programs, free variables at non-\(\varepsilon\) stages are considered
symbols and do not cause unbound variable errors in \miniML, so we
relax the notion of programs to include terms that contain such
symbolic variables.  We say that \(\Env\) is \emph{$\varepsilon$-free} if it
satisfies \(A \neq \varepsilon\) for any \(x:\tau@A \in \Env\); then,
a program is a term which is typed under an \(\varepsilon\)-free
environment.  In the statement of Type Soundness Theorem, we also use
the notation $ \Env^{-A} $, defined by: \( \Env^{-A} = \{ x : \tau @ B
~|~ x : \tau @ AB \in \Env \} \).
\begin{thm}[Type Soundness]
\label{thm:ml-sound}
If \(\Env\) is \(\varepsilon\)-free and $ \Env \vdash^{\varepsilon}
M : \tau $ and $ \bred{\varepsilon}{M}{R} $,
then $ R = v $ and $ v \in V^{\varepsilon} $ for some $ v $
and $ \Env \vdash^{\varepsilon} v : \tau $.
In particular, if \(\tau = \Seal{\alpha}{\tau_0}\), then
$ v = \Next{\alpha}{N} $ and
$ \Env^{-\alpha} \vdash^{\varepsilon} N : \tau_0 $.
\begin{proof}
See Section~\ref{sec:proof-sound}.
\end{proof}
\end{thm}

\subsection{Staged Transition Variables for Erasure Property}
\label{sec:erasure-semantics}
The evaluation of \miniML\ introduced above relies on the annotation
of transition variables.  For example, consider two terms
\begin{eqnarray*}
  M_1 & = & (\Gen{\alpha}{\Next{\alpha}{(\lam x. 1) (\Next{\beta}{((\textbf{fix }f. f) 2)})}})\ \varepsilon \\
  M_2 & = & (\Gen{\alpha}{\Next{\alpha}{(\lam x. 1) (\Next{\alpha}{((\textbf{fix }f. f) 2)})}})\ \varepsilon.
\end{eqnarray*}
The only difference is the annotation on \( \Next{}{((\textbf{fix
  }f. f) 2)} \), but \( \bred{\varepsilon}{M_1}{1} \) whereas there is no
term \( N \) such that \( \bred{\varepsilon}{M_2}{N} \). In other
words, the evaluation of \( M_1 \) terminates but that of \( M_2 \)
diverges.  Therefore, we must record how transition variables are bound
to transitions during evaluation.

From the implementation point of view, it is desirable that evaluation
is insensitive to the annotation as much as possible to avoid
overhead.  In \(\lambda^\alpha\)~\cite{604134},
%(as well as MetaOCaml)
environment classifiers can be regarded as completely static citizens
so that the evaluation does not require them, although the authors do
not explicitly state it.  The property that the evaluation goes well
even if we erase the annotations is called \emph{erasure property}.
The previous example shows that the erasure property does \emph{not}
hold for \miniML.  Since the argument \(B\), especially, its length,
in an instantiation \(M\app B\) is significant at run time, we cannot
erase transitions completely.  So, we consider a slightly weaker
notion of erasure, which removes transition variables only from
\(\Next{}{}, \Prev{}{}\) and \(\Lambda\) and replaces the transition
\(B\) in \(M\app B\) with its length.  The goal of this section is to
find a practically meaningful subset of \miniML{}, which enjoys the
erasure property under the translation sketched above.

The reasons why the erasure property is broken are (1) \( \Lambda \)-bound
transition variables are used ``too far'' from the binder, as is the case in \( M_2 \)
and (2) the ``depth'' of quoting \( \Next{\alpha}{} \) can be changed
by using instantiation with a transition, whose length is not \( 1 \).
In the case of \( M_2 \), there is an occurrence of transition variable \( \alpha \)
far from the binder and \( \alpha \) is instantiated by \( \varepsilon \), whose length is \( 0 \).
So, to ensure the erasure property, it is enough to prevent both (1) and (2) from holding
at once, in other words, to guarantee that \( \Lambda \)-bound transition variables occur near the binder
or to restrict instantiations to only transitions of length \( 1 \).

Based on this observation, we will introduce two instantiation rules.
The first rule is for instantiation of transition variables used only
near the binder.  We can change the depth of quoting by using this
rule, but this rule can be applied only in limited situations.  The
second rule is for instantiation of transition variables by
transitions, whose lengths are \( 1 \).  This rule can be applied to
any \(\forall\)-types, but we cannot change the depth of quoting.  We
introduce a new term constructor \( \SIns{M}{\alpha} \) to distinguish
from the former.

The first instantiation rule requires some control on the occurrences
of transition variables.  We enforce one additional
restriction, which requires that transition variables be also staged
like term variables.  This restriction rejects a type with nested
occurrences of \(\Seal{\alpha}{}\), such as \(\forall
\alpha. \Seal{\alpha}{\Seal{\alpha}{\tau}}\), whose term would have a
distant use of \(\Next{\alpha}{}\).  This restriction is closely
related to the distinction between open and closed code types in
\(\lambda^i\)~\cite{Calcagno2004}.

We define a new type system with staged transition variables.  We need
two changes to deal with the stages for transition variables.  First,
we introduce environments for transition variables.  A
\emph{transition environment} is a set of the form \( \{ \alpha_1 @
A_1, \dots, \alpha_n @ A_n \} \), where \( \alpha_i \) are distinct
transition variables.  An intuitive meaning of \( \alpha @ A \) is
that the valid occurrence of \( \alpha \) is always of the form \( A
\alpha \).  The second change is the annotation for the universal
quantifier.  The new syntax for universal quantification is \( \forall
\alpha @ A. \tau \), where \( A \) is the (positive) path from the
current stage to the stage in which \( \alpha \) is usable.
% This notation and meaning are very similar to the formula
% \( \forall x \in D. P \) in the predicate logic.

Next, we define well-formed transitions, transition environments,
types, and type environments to ensure every use of a transition
variable is valid.
We say a transition \( A = \alpha_1 \dots \alpha_n \) is well formed
under a transition environment \( \Delta \) if, for any \( i \le n \),
\( \alpha_i @ \alpha_1 \dots \alpha_{i - 1} \in \Delta \).  We say \(
\Delta \) is well formed if, for any \( \alpha @ A \in \Delta \), \( A
\) is well formed under \( \Delta \), i.e., all stages where
transition variables are declared are well formed.  This definition
avoids the circular definition of transition variables, e.g., \(
\alpha @ \beta, \beta @ \alpha \).  We write \( \Delta
\svdash A \) if \( A \) is well formed under \( \Delta \), and \(
\svdash \Delta \) if \( \Delta \) is a well-formed transition
environment.

The judgment of the form \( \Delta \svdash^A \tau \) means ``type \(
\tau \) is well formed at stage \( A \) under \( \Delta \)'', and
defined by the rules in Figure~\ref{fig:staged-type}.  The base types
\textbf{int} and \textbf{bool} are always well formed at any
well-formed stage.  The rules for \( \tau \to \sigma \) and \(
\Seal{\alpha}{\tau} \) resemble the typing rules \(\textsc{(Abs)}\)
and \textsc{(\(\Next{}{}\))}, respectively.  The type \( \forall \alpha @ B. \tau \),
which binds a new transition variable \( \alpha \), is well formed at
\( A \) under \( \Delta \) if \( \tau \) is well formed under the
transition variables environment extended by the new transition
variable declaration \( \alpha @ AB \).  Finally, we define
well-formedness of type environment \( \Gamma \) under \( \Delta \),
written \( \Delta \svdash \Gamma \), by: \( \Gamma \) is well formed
under \( \Delta \) if and only if \(\Delta\) is well formed and, for
any \( x : \tau @ A \in \Gamma\), \( \tau \) is well formed at \( A \)
under \( \Delta \) (i.e., \(\Delta \svdash^A \tau\)).

\begin{figure}[t]
\begin{multicols}{2}
  \centering
\infrule[ST-Base]{
  \Delta \svdash A \AND
  b \in \{ \textbf{int}, \textbf{bool} \}
}{
  \Delta \svdash^A b
}
\\[1ex]
\infrule[ST-Imp]{
  \Delta \svdash^A \tau \AND
  \Delta \svdash^A \sigma
}{
  \Delta \svdash^A \tau \to \sigma
}
\\[1ex]
\infrule[ST-Code]{
  \Delta \svdash^{A \alpha} \tau
}{
  \Delta \svdash^A \Seal{\alpha}{\tau}
}
\\[1ex]
\infrule[ST-Univ]{
  \Delta, \alpha @ AB \svdash^A \tau
}{
  \Delta \svdash^A \forall \alpha@B. \tau
}
\end{multicols}
\vspace{-1em}
  \caption{The definition of well-formed types under the transition environment \( \Delta \).}
  \label{fig:staged-type}
\end{figure}

Figure~\ref{fig:staged-typing} shows the typing rules that differ from
the previous type system (except the addition of \(\Delta\)). They
have additional premises about well-formedness.  The rule
\textsc{(S-Var)} requires the well-formedness of environment \(
\Gamma, x : \tau @ A \), which will require well-formedness of the
type \( \tau \) at \( A \) and the transition environment \( \Delta
\).  The rules \textsc{(S-Num)} and \textsc{(S-Bool)} require the
well-formedness of the environment \( \Gamma \) and the stage \( A \),
which ensures the well-formedness of the base types.
The typing rule \textsc{(S-Gen)} records the path
from the current stage to the stage in which \( \alpha \) is usable.
This information is used by the rules \textsc{(S-Ins1)} and \textsc{(S-Ins2)}.
As mentioned above, there are two kinds of transition instantiation rules and corresponding term constructors.
The first one \textsc{(S-Ins1)} is computationally meaningful, in other words it may change the depth of quoting,
but can be used only in limited situations.  The second one \textsc{(S-Ins2)} does not change the depth of
quoting, so this is computationally meaningless and we can use anytime.
Here, substitution for a transition variable \( \alpha \) in \( \forall
\alpha @ A. \tau \) (among other types) is defined as follows:
\[
  (\forall \alpha @ A. \tau) \Subst{\beta := B} = \forall \alpha @ (A \Subst{\beta := B}). (\tau \Subst{\beta := B})
\]
% So to keep the well-formedness, the well-formedness of \( A B
% \) under \( \Delta \) is required.
It is easy to see that \( \Gamma; \Delta \svdash^A M : \tau \) implies \( \Gamma \vdash^A M : \tau \).
% So the type soundness of \( \svdash \) is trivially holds.

\begin{figure}
\begin{multicols}{2}
  \centering
\infrule[S-Var]{
  \Delta \svdash \Gamma, x : \tau @ A
}{
  \Gamma, x : \tau @ A; \Delta \svdash^A x : \tau
}
\\[1ex]
\infrule[S-Num]{
  \Delta \svdash \Gamma \AND
  \Delta \svdash A
}{
  \Gamma; \Delta \svdash^A n : \textbf{int}
}
\\[1ex]
\infrule[S-Bool]{
  \Delta \svdash \Gamma \AND
  \Delta \svdash A \\
  bv \in \{ \textbf{true}, \textbf{false} \}
}{
  \Gamma; \Delta \svdash^A bv : \textbf{bool}
}
\\[1ex]
\infrule[S-Gen]{
  \alpha \notin \FMV(\Gamma) \cup \FMV(\Delta) \\
  \Gamma; \Delta, \alpha @ A B \svdash^A M : \tau
}{
  \Gamma; \Delta \svdash^A \Lambda \alpha. M : \forall \alpha @ B. \tau
}
\\[1ex]
\infrule[S-Ins1]{
  \Gamma; \Delta \svdash^A M : \forall \alpha @ \varepsilon. \tau \\
  \Delta \svdash A B \qquad \tau = \Seal{\alpha}{\sigma}
}{
  \Gamma; \Delta \svdash^A \Ins{M}{B} : \tau \Subst{\alpha := B}
}
\\[1ex]
\infrule[S-Ins2]{
  \Gamma; \Delta \svdash^A M : \forall \alpha @ B. \tau \\
  \beta @ AB \in \Delta
}{
  \Gamma; \Delta \svdash^A \SIns{M}{\beta} : \tau \Subst{\alpha := \beta}
}
\end{multicols}
  \vspace{-1em}
  \caption{The typing rules which differ from the previous type system.}
  \label{fig:staged-typing}
\end{figure}

Now, we define the syntax for \emph{erased terms}, terms after erasure
and the erasure translation from \(\sname\) terms to erased terms, and the
big-step semantics for erased terms.
%  the erasure property, i.e.,
% for any term \( M \) which is typable in \( \svdash \),
% we can evaluate \( M \) even if we forget the annotations in \( M \).
% To state formally, we define 
% \emph{erased types} \( \tau^\flat \), 
% \emph{erased terms} \( M^\flat \),
% which have no annotation of transition variables, and \emph{erasure semantics} \( \ered{n}{M^\flat}{N^\flat} \).
 The syntax of erased %types and 
terms, ranged over by \(M^\flat\), is as follows:
\begin{center}
\(\begin{array}{@{}lr@{}r@{}l}
%     \textrm{\it Erased Types}\quad
%       & \tau^\flat &{}::={}&
%         \textbf{int} \OR
%         \textbf{bool} \OR
% 	\tau^\flat \to \tau^\flat \OR
% 	\Seal{}{\tau^\flat} \OR
% 	\forall . \tau^\flat \\
    \textrm{\it Erased Terms}\quad
      & M^\flat &{}::={}&
        x \OR
        n \OR
        b \OR
        M^\flat = M^\flat \OR
        M^\flat + M^\flat \OR
        M^\flat - M^\flat \OR
        M^\flat * M^\flat 
% \OR
%         \textbf{true} \OR
%         \textbf{false} 
      \\
      & & \OR &
        \eIf{M^\flat}{M^\flat}{M^\flat} \OR
        \textbf{fix}\ f %: \tau \to \sigma
           . M^\flat \OR
        \lam x % : \tau^\flat 
           . M^\flat \OR
        M^\flat M^\flat \\
      & & \OR &
        \Next{}{M^\flat} \OR
        \Prev{}{M^\flat} \OR
        \Lambda {M^\flat} \OR
        M^\flat \app \unittrans \OR
        \Ins{M^\flat}{n} \textrm{ (\(n \ge 0\))}
  \end{array}
\)
\end{center}
The \emph{erasing function} \( \flat(\cdot) \) from % types and 
terms to erased % types and 
erased terms is defined as follows:
\begin{eqnarray*}
%   \flat(b) & = & b \quad \textrm{(if \( b \in \{ \textbf{int}, \textbf{bool} \} \))} \\
%   \flat(\tau \to \sigma) & = & \flat(\tau) \to \flat(\sigma) \\
%   \flat(\Seal{\alpha}{\tau}) & = & \Seal{}{\flat(\tau)} \\
%   \flat(\forall \alpha. \tau) & = & \forall \flat(\tau) \\
%   \\
  \flat(c) & = & c \qquad (c \in \mathbb{Z} \cup \{ \textbf{true}, \textbf{false} \})\\
  \flat(M \Diamond N) & = & \flat(M) \Diamond \flat(N)  \qquad (\Diamond \in \{+, -, *, =\}) \\
  \flat(\eIf{M}{N_1}{N_2}) & = & \eIf{\flat(M)}{\flat(N_1)}{\flat(N_2)} \\
  \flat(\textbf{fix}\ f : \tau \to \sigma . M) & = & \textbf{fix}\ f. \flat(M) \\
  \flat(\lam x : \tau. M) & = & \lam x %: \flat(\tau)
      . \flat(M) \\
  \flat(M \app N) & = & \flat(M) \app \flat(N) \\
  \flat(\Next{\alpha}{M}) & = & \Next{}{\flat(M)} \\
  \flat(\Prev{\alpha}{M}) & = & \Prev{}{\flat(M)} \\
  \flat(\Gen{\alpha}{M}) & = & \Lambda \flat(M) \\
  \flat(M [\beta]) & = & \flat(M) \app \unittrans \\
  \flat(\Ins{M}{A}) & = & \flat(M) \app n \quad \textrm{(\( n \) is the length of \( A \)).}
\end{eqnarray*}
So, \( \flat(\Next{A}{M}) = \overbrace{\Next{}{\dots \Next{}{}}}^n \flat(M) \) and
\( \flat(\Prev{A}{M}) = \overbrace{\Prev{}{\dots \Prev{}{}}}^n \flat(M) \) where \( n \) is the length of \( A \).

The \emph{erasure semantics} is essentially the same as the ordinary
evaluation semantics in Section~\ref{sec:evaluation}, except for the
two differences: one is that the stage \( A \) of \( \bred{A}{M}{N} \)
is replaced by the natural number \( n \), which is the length of \( A
\); and the other is that the rule for \( M \app n \) at the stage \(
0 \).  In this case, \( M \) must be evaluated to \( \Lambda
\Next{}{M'} \) and \(\Next{}{}\) at its head is duplicated by \( n \)
times.  We show only main rules below.
\begin{center}
\infrule{
  \ered{0}{M^\flat}{\Next{}{N^\flat}}
}{
  \ered{1}{\Prev{}{M^\flat}}{N^\flat}
}
\qquad
\infrule{
  \ered{n + 1}{M^\flat}{N^\flat}
}{
  \ered{n}{\Next{}{M^\flat}}{\Next{}{N^\flat}}
}
\qquad
\infrule{
  \ered{0}{M^\flat}{\Lambda N^\flat}
}{
  \ered{0}{M^\flat \app []}{N^\flat}
}
\\
\infrule{
  \ered{0}{M^\flat}{\Lambda \Next{}{M_0^\flat}} \AND
  \ered{0}{\overbrace{\Next{}{\Next{}{\cdots\Next{}{}}}}^n M_0^\flat}{N^\flat}
}{
  \ered{0}{M^\flat \app n}{N^\flat}
}
\end{center}

Finally, we state the erasure property: the erasure semantics is
equivalent to the semantics with transition variables for terms typed
under the new type system.
\begin{thm}[Erasure Property]
\label{thm:erasure}
Suppose \(\Gamma\) is \(\varepsilon\)-free and \( \Gamma; \Delta \svdash^\varepsilon M : \tau \).  Then,
\begin{enumerate}[\em(1)]
\item if \( \bred{\varepsilon}{M}{N} \), then
\( \ered{0}{\flat(M)}{\flat(N)} \); and 
\item if \( \ered{0}{\flat(M)}{N^\flat} \), then
there is some \( N' \) such that \( \bred{\varepsilon}{M}{N'} \)
and \( N^\flat = \flat(N') \).
\end{enumerate}
\proof
See Appendix~\ref{sec:proof-erasure}.
\qed
\end{thm}

We believe that the calculus with this new type system does not lose
much expressiveness for practical use; in fact, the example of power
functions in Section~\ref{sec:power} can be typed with the new type
system.

% This restriction, which resembles one in $ \lambda^i $~\cite{Calcagno2004},
%  still allows embedding of $ \lambda^\bigcirc $ and $ \lambda^\square $
%  and $ \texttt{power}_\forall $ ().

% The new type system is sufficiently expressive by integrating ML-style
% let polymorphism.  If we consider the following typing rule for \clet
% %
% \begin{center}
% \infrule{
%   \Gamma; \Delta \svdash M : \tau' \AND
%   \Gamma; \Delta \svdash N\Subst{x := M} : \tau
% }{
%   \Gamma; \Delta \svdash \clet x = M \mathop{\textbf{in}} N : \tau
% }
% \end{center}
% %
% then the example of the power function in Section~\ref{sec:power} will
% be typable.  Note that we use a typing rule using substitution and
% \clet-bound variable does not have a type declaration, because
% \(\forall \alpha. \cint \to \Seal{\alpha}{\cint} \to
% \Seal{\alpha}{\cint}\) is not a well-formed type.  Formalization of
% another type system using type schemes is left for future work.

\section{Kripke Semantics  for $ \sname $ and Logical Completeness}
\label{sec:logic}
In this section, we formalize the Kripke semantics discussed in
Section~\ref{sec:overview} and prove completeness of a classical
version of the proof system obtained from \(\sname\) to justify the
informal interpretation of types in \(\sname\) as formulae of a logic.
We augment the set of propositions (namely types) with falsity and the
proof rules with double negation elimination.  It is left for future
work to study the semantics of the intuitionistic version, of which
recent work on Kripke semantics for intuitionistic
LTL~\cite{KojimaIgarashi08IMLA,KojimaIgarashi10IC} can be a basis.

First, we (re)define the set of propositions and the natural deduction
proof system.  Then, we proceed to the formal definition of the Kripke
semantics and prove soundness and completeness of the proof system.
Finally, we examine another rule for the double negation elimination.

\subsection{Natural Deduction}
\label{sec:deduction-rule}
The set $ \TypeSet_\bot $, ranged over by \(\phi\) and \(\psi\), of
\emph{propositions} is given by the grammar for $ \TypeSet $ extended with a
new constant $ \bot $.

The natural deduction system can be obtained by forgetting variables
and terms in the typing rules.  We add the following new rule,
which is the ordinary double negation elimination rule,
adapted for this setting:
\begin{center}
 \DBotET\ .
\end{center}

\subsection{Kripke Semantics and Completeness}
\newcommand{\reachablevia}[1]{%
\mathbin{\stackrel{#1}{\longrightarrow}}}
\newcommand{\notreachablevia}[1]{%
%\mathbin{\stackrel{#1}{\not\longrightarrow}}}
\mathbin{\,\,\,\,\,\not\!\!\!\!\! \reachablevia{#1}}}

As mentioned in Section 2, the Kripke semantics for this logic is
based on a functional transition system \(\mathcal{T} = (S, L,
\{\reachablevia{a} \mid a \in L\})\) where \(S\) is the (non-empty)
countable set of states, \(L\) is the countable set of labels, and
\(\mathord{\reachablevia{a}} \in S \to S\) for each label \(a \in L\).
We write \(s \reachablevia{a_1 \cdots a_n} s'\) if there exist \(s_1,
\ldots, s_{n-1}\) such that \(s \reachablevia{a_1} s_1
\reachablevia{a_2} \cdots \reachablevia{a_{n-1}} s_{n-1}
\reachablevia{a_n} s'\).  Actually, given \(s, a_1, \ldots, a_n\),
\(s'\) always exists in this setting because \( \reachablevia{a_i} \)
is a total function for all \( 1 \le i \le n \).

To interpret a proposition, we need two \emph{valuations}, one for
propositional variables and the other for transition variables.  The
former is a total function $ v \in S \times \PropVar \to \{ 0, 1 \} $;
the latter is a total function $ \rho \in \Gene \to L^{*} $, where
\(L^{*}\) is the set of all finite sequences of labels.
Then, we define the \emph{satisfaction relation} $ {\mathcal T}, v, \rho; s
\Vdash \phi $, where $ s \in S $ is a state, as follows:
\newcommand{\ifff}{\;\;\textrm{iff}\;\;}
\[
\begin{array}{lcl}
  \mathcal{T}, v, \rho; s \Vdash p & \ifff & v(s , p) = 1 \\
  \mathcal{T}, v, \rho; s \Vdash \bot & & \textrm{never occurs} \\
  \mathcal{T}, v, \rho; s \Vdash \phi \to \psi
    & \ifff & \mathcal{T}, v, \rho; s \not \Vdash \phi
      \;\; \textrm{or} \;\;
      \mathcal{T}, v, \rho; s \Vdash \psi \\
  \mathcal{T}, v, \rho; s \Vdash \Seal{\alpha}{\phi}
    & \ifff & \mathcal{T}, v, \rho; s' \Vdash \phi \textrm{ where }
    s \reachablevia{\rho(\alpha)} s' \\
%    & \ifff & s \notreachablevia{\rho(\alpha)} \textrm{ or }
%      \mathcal{T}, v, \rho; s' \Vdash \phi \textrm{ where }
% s \reachablevia{\rho(\alpha)} s'\\
  \mathcal{T}, v, \rho;  s \Vdash \forall \alpha . \phi
    & \ifff & \textrm{for all} \; A \in L^{*}, \;\;
      \mathcal{T}, v, \rho[A/\alpha]; s \Vdash \phi
\end{array}
\]
Here, $ \rho[A/\alpha] $ is defined by: $ \rho[A/\alpha](\alpha) = A
$ and $ \rho[A/\alpha](\beta) = \rho(\beta) $ (for $ \beta \neq \alpha $).
The satisfaction relation is extended pointwise to contexts \(\Gamma\)
(possibly infinite sets of pairs of a proposition and a transition%
\footnote{We allow \( \Gamma \) to be infinite for technical convenience.
For an infinite context \( \Gamma \), we write \( \Gamma \vdash^A \phi \)
if there exists a finite context \( \Gamma' \subseteq \Gamma \) such that
\( \Gamma' \vdash^A \phi \).  The following argument holds if we restrict
\( \Gamma \) to be finite because the logic is compact, i.e., \( \Gamma \)
is unsatisfiable if and only if there exists a finite subset \( \Gamma' \subseteq \Gamma \)
such that \( \Gamma' \) is also unsatisfiable.  For more detail,
see Appendix~\ref{sec:proof-of-completeness}.
}%
) by:
\[
  \mathcal{T}, v, \rho; s \Vdash \Gamma  \ifff 
    \mathcal{T}, v, \rho; s \Vdash \Seal{A}{\phi}
    \textrm{ for all } \phi @ A \in \Gamma \ .
\]
The local consequence relation \(\Gamma \Vdash \phi\) is
defined by:
\[
 \Gamma \Vdash \phi \ifff
  \mathcal{T}, v, \rho; s \Vdash \Gamma \text{ implies }
  \mathcal{T}, v, \rho; s \Vdash \phi \text{ for any }
    \mathcal{T}, v, \rho, s \ . 
\]

Then, the natural deduction proof system is sound and complete with
respect to the local consequence relation.
The proof is similar to the one for first-order predicate logic:
we use the standard techniques of Skolemization and Herbrand structure.
\begin{thm}[Soundness]
If $ \Gamma \vdash^{\varepsilon} \phi $, then
$ \Gamma \Vdash \phi $.
\begin{proof}
By induction on the derivation of $ \Gamma \vdash^{\varepsilon} \phi $.
\end{proof}
\end{thm}

\begin{thm}[Completeness]
\label{thm:completeness}
If $ \Gamma \Vdash \phi $, then
$ \Gamma \vdash^{\varepsilon} \phi $.
\begin{proof}
  We give a proof sketch; more detailed proofs are found in
  Appendix~\ref{sec:proof-of-completeness}.

We assume $ \Gamma \nvdash^{\varepsilon} \phi $.
We construct a transition system $ {\mathcal T} $,
two valuations $ v $ and $ \rho $ and
a state $ s $
such that $ {\mathcal T}, v, \rho; s \Vdash \Gamma $
and $ {\mathcal T}, v, \rho; s \nVdash \phi $.
The construction is similar to the construction of counter models
in first-order predicate logic.

First, we prove the following proposition.
\begin{center}
$ \Gamma \nvdash^{\varepsilon} \bot $
implies there exists $ {\mathcal T} $, $ v $, $ \rho $ and $ s $
such that $ {\mathcal T}, v, \rho; s \Vdash \Gamma $.
\end{center}
We can reduce this proposition to
the completeness of quantifier-free logic
(the logic without the quantifier over transition variables)
by using Skolemization and Herbrand Universes.

Then, we prove this theorem by contraposition.
Because $ \Gamma \nvdash^{\varepsilon} \phi $,
we have $ \Gamma, \phi \to \bot @ \varepsilon \nvdash^{\varepsilon} \bot $.
Therefore, there exist $ {\mathcal T} $, $ v $, $ \rho $ and $ s $
which satisfy $ {\mathcal T}, v, \rho; s \Vdash \Gamma $ and
$ {\mathcal T}, v, \rho; s \Vdash \phi \to \bot $.
Then, $ {\mathcal T}, v, \rho; s \nVdash \phi $
because $ {\mathcal T}, v, \rho; s \Vdash \phi \to \bot $.
Therefore, $ \Gamma \nVdash \phi $.
\end{proof}
\end{thm}

\subsection{Alternative Semantics}
\label{sec:another-semantics}
We can give an alternative deduction rule for double negation elimination.
\begin{center}
\textrm{\DBotEP}
\end{center}
The difference is in the stage of the premise.  This rule requires
that the stage of $ \bot $ is equal to the stage of $ \phi \to \bot $,
but in the rule in Section~\ref{sec:deduction-rule} the stage of $
\bot $ is arbitrary.  This restriction makes the proof system weaker:
for example, in this setting $ \Seal{\alpha}{} $ is not self-dual,
that means $ \neg \Seal{\alpha}{\neg \phi} \leftrightarrow
\Seal{\alpha}{\phi} $ is not provable (here \( \neg \phi \) is an abbreviation
of \( \phi \to \bot \)), while under the previous rules
$ \Seal{\alpha}{} $ is self-dual.  The difference is equivalent to the
axiom $ \Seal{A}{\bot} \leftrightarrow \Seal{B}{\bot} $ (or a weaker
form $ \forall \alpha . (\Seal{\alpha}{\bot} \to \bot) $) in the sense
that this system with this axiom is equivalent to the previous proof
system.  This axiom corresponds to the axiom $ \bigcirc \neg A
\leftrightarrow \neg \bigcirc A $ in linear-time temporal logic, due
to Stirling~\cite{StirlingHandbook}.

There is a corresponding semantics, with respect to which the new
proof system is sound and complete.  In fact, this is achieved by a
minor change: we allow transition functions $ \reachablevia{a} $ to be
partial.  As a result, there can be no $ s' $ such that $ s
\reachablevia{\rho(\alpha)} s' $.
In this setting, the semantics for $ \Seal{\alpha}{\phi} $
has to be modified.  We define $ s \notreachablevia{a_1
  \dots a_n} $ as there is no $ s' $ which satisfies $ s
\reachablevia{a_1 \dots a_n} s' $, and $ {\mathcal T}, v, \rho; s
\Vdash \Seal{\alpha}{\phi} $ as follows:
\begin{eqnarray*}
  {\mathcal T}, v, \rho; s \Vdash \Seal{\alpha}{\phi}
    & \ifff &
      s \notreachablevia{\rho(\alpha)} \textrm{ or }
      {\mathcal T}, v, \rho; s' \Vdash \phi \textrm{ where }
      s \reachablevia{\rho(\alpha)} s'
\end{eqnarray*}
Completeness and soundness are proved similarly.

\section{Comparing with other multi-stage calculi}
\label{sec:embedding}
In this section, we will compare \( \sname \) with closely related
calculi \( \lambda^\bigcirc \)~\cite{Davies-Atemporallogicappro}, the
Kripke-style modal \(\lambda\)-calculus~\cite{Davies2001}, \(
\lambda^\alpha \)~\cite{604134} and \( \lambda^i
\)~\cite{Calcagno2004}.  The first two calculi are based on
Curry-Howard correspondence between multi-stage calculi and modal
logics and our work can be considered a generalization of them.  In
fact, there are embeddings from these two calculi to \(\sname\).  The
calculi \( \lambda^\alpha \) and \( \lambda^i \) are multi-stage
calculi with environment classifiers.  We discuss several differences
among these two calculi and \(\sname\).  Although it does not seem
possible to give (straightforward) embeddings from them, due to the
presence of CSP, we will show that an embedding from the CSP-free
fragment of \( \lambda^i\) to \(\sname\) is possible.

\subsection{Comparing with \( \lambda^\bigcirc \)}
\( \lambda^\bigcirc \)~\cite{Davies-Atemporallogicappro} is a
multi-stage calculus corresponding to linear-time temporal logic
(LTL).  As already mentioned in Section~\ref{sec:overview}, \(
\lambda^\bigcirc \) is obtained by using only one transition variable,
say \( \alpha \). Then the modal operator \( \bigcirc \) to mean
``next'' in LTL corresponds to the modal operator \( \Seal{\alpha}{}
\) and the stage \( n \) in LTL corresponds to \( \alpha^n \) and so
on.  We define an embedding \( \LV \cdot \RV \) from \( \lambda^\bigcirc
\) into \( \sname \) in Figure~\ref{fig:embedding-lambda-circle}.  
This embedding is essentially the same as that from
\(\lambda^\bigcirc\) calculus into \(\lambda^\alpha\), given by Taha
and Nielsen~\cite{604134}.

\begin{figure}%[h]
\begin{eqnarray*}
  \EmbCT{b} & = & b \\
  \EmbCT{\tau \to \sigma} & = & \EmbCT{\tau} \to \EmbCT{\sigma} \\
  \EmbCT{\bigcirc \tau} & = & \Seal{\alpha}{\EmbCT{\tau}} \\
  \\
  \EmbCT{x} & = & x \\
  \EmbCT{\lam x : \tau . M} & = & \lam x : \EmbCT{\tau} . \EmbCT{M} \\
  \EmbCT{M \app N} & = & \EmbCT{M} \app \EmbCT{N} \\
  \EmbCT{\textbf{next } M} & = & \Next{\alpha}{\EmbCT{M}} \\
  \EmbCT{\textbf{prev } M} & = & \Prev{\alpha}{\EmbCT{M}} \\
  \\
  \EmbCT{\cdot} & = & \cdot \\
  \EmbCT{\Gamma, x : A @ n} & = & \EmbCT{\Gamma}, x : \EmbCT{A} @ \alpha^n
\end{eqnarray*}
\caption{Embedding from \(\lambda^\bigcirc\) to \(\sname\).  \( \alpha
  \) is a fixed transition variable.  }
\label{fig:embedding-lambda-circle}
\end{figure}

The following two theorems show the correctness of the embedding.
\begin{thm}
If \( \Gamma \vdash^n M : \tau \) in \( \lambda^\bigcirc \),
then \( \EmbCT{\Gamma} \vdash^{\alpha^n} \EmbCT{M} : \EmbCT{\tau} \).
\proof
By induction on the type derivation.
\qed
\end{thm}

\begin{thm}
Let \( M \), \( N \) be \( \lambda^\bigcirc \) terms.
Then, \( M \red N \) iff \( \EmbCT{M} \red \EmbCT{N} \).
\proof
By induction on the structure of \( M \).
\qed
\end{thm}

Moreover, by giving a suitable definition of reduction annotated with
time for \(\lambda^\bigcirc\), we can easily show that the embedding
preserves the stage of reduction.

We can construct a reverse mapping, a type- and reduction-preserving embedding
from the quantifier-free fragment of \( \sname \) to \( \lambda^\bigcirc \),
by simply forgetting annotations of transition variables.  Moreover, the quantifier-free
fragment of \( \sname \) with only one transition variable is isomorphic to \( \lambda^\bigcirc \)
in the sense that there is a bijection that preserves typability and reduction.
%  We focus on $\beta$-reduction because the stage of
% reduction is defined in \( \lambda^\bigcirc \) only for
% $beta$-reduction.  In \( \lambda^\bigcirc \), the stage of the
% reduction \( M \red N \) is defined as the stage of the redex, not the
% path as in \( \sname \).  So, we should define the path of the
% reduction for \( \lambda^\bigcirc \).  We define the staged reduction
% \( M \stackrel{n}{\red} N \) of \( \lambda^\bigcirc \) by \( M
% \stackrel{m - n}{\red} N \) if the stage of \( M \) is \( n \) and the
% stage of the redex is \( m \).  This is well-defined and equivalent to
% the definition in \( \lambda^{\bigcirc \square} \), which includes \(
% \lambda^\bigcirc \) as a fragment.  Then we can state the preservation
% of the stage of reduction as \( M \stackrel{n}{\red} N \) if and only
% if \( \EmbCT{M} \stackrel{\alpha^n}{\red} \EmbCT{N} \).

\subsection{Comparing with Calculus based on S4}
Davies and Pfenning~\cite{Davies2001,DBLP:conf/popl/DaviesP96} develop
calculi that correspond to intuitionistic modal logic S4 (only with
\(\Box\)).  The type \(\Box \tau\) represents closed code values,
which thus can be run or embedded in code of any later stages, as is
possible in \(\sname\).  We compare \(\sname\) with one of their
calculi (what they call the Kripke-style modal \(\lambda\)-calculus in
Section~4 of \cite{Davies2001}), in which, there are \(\textbf{box}\)
and \(\textbf{unbox}_n\) for quoting and unquoting, respectively (see
Pfenning and Davies~\cite{Davies2001} for details).
An embedding \( \EmbBT{\cdot} \) from this calculus 
into \( \sname \) is given in Figure~\ref{fig:embedding-lambda-square}.
% We can also embed $ \lambda^{\square} $ into $ \sname $ as follows:
% We fix an infinite sequence of distinct classifiers \( \alpha_1 \alpha_2 \dots \),
% and define \( A_n = \alpha_1 \alpha_2 \dots \alpha_n \).
% (This is essentially same as the embedding of \( \lambda^\square \) into \( \lambda^i \)~\cite{604134}).

\begin{figure}
\begin{eqnarray*}
  \EmbBT{b} & = & b \\
  \EmbBT{\tau \to \sigma} & = & \EmbBT{\tau} \to \EmbBT{\sigma} \\
  \EmbBT{\square \tau} & = & \forall \alpha. \Seal{\alpha}{\EmbBT{\tau}} \\
  \\
  \EmbBT{x}^A & = & x \\
  \EmbBT{\lam x : \tau . M}^A & = & \lam x : \EmbBT{\tau} . \EmbBT{M}^A \\
  \EmbBT{M \app N}^A & = & \EmbBT{M}^A \app \EmbBT{N}^A \\
  \EmbBT{\textbf{box } M}^A & = & \Lambda \alpha. \Next{\alpha}{\EmbBT{M}^{A \alpha}}
    \quad \textrm{(where \( \alpha \notin \FMV(A) \))} \\
  \EmbBT{\textbf{unbox}_n\ M}^A & = & \Prev{B}{\EmbBT{M}^{A'} \app B} \\ &&
    \quad \textrm{(where \( A = A' B \) and the length of \( B \) is \( n \))}
  \\
  \EmbBT{x_1 : \tau_1, \dots x_n : \tau_n}^A & = &
    x_1 : \EmbBT{\tau_1} @ A, \dots, x_n : \EmbBT{\tau_n} @ A \\
  \\
  \EmbBT{\Gamma_0;\cdots; \Gamma_n}^{\alpha_1\cdots \alpha_n} & = & 
        \EmbBT{\Gamma_0}^{\varepsilon}, \ldots,
        \EmbBT{\Gamma_i}^{\alpha_1\cdots \alpha_i}, \ldots,
        \EmbBT{\Gamma_n}^{\alpha_1\cdots \alpha_n}
\end{eqnarray*}
  \caption{Embedding from the Kripke-style modal \(\lambda\)-calculus to \(\sname\).}
\label{fig:embedding-lambda-square}
\end{figure}

The following theorems state the correctness of the embedding.
\begin{thm}
  Consider a sequence $A$ of distinct transition variables of
  length \( n \).  Then \( \Gamma_0; \Gamma_1; \dots; \Gamma_n \vdash M
  : \tau \) implies \( \EmbBT{\Gamma_0; \dots; \Gamma_n}^A \vdash^A
  \EmbBT{M}^A : \EmbBT{\tau} \).  
  \proof %
  An easy induction on the type derivation \( \Gamma_0; \dots;
  \Gamma_n \vdash M : \tau \).  \qed
\end{thm}

\begin{thm}
  Let \( A \) be a sequence of distinct transition variables of 
  length \( n \).  Provided that \( \Gamma_0; \dots; \Gamma_n \vdash M : \tau \)
  is derivable in the Kripke-style modal \(\lambda\)-calculus, then \( M \red_\beta N \) iff
  \( \EmbBT{M}^A \red \EmbBT{N}^A \).  
  \proof %
  By induction on the structure of \( M \). \qed
\end{thm}

Taha and Nielsen have shown a similar embedding from the Kripke-style
calculus into \(\lambda^\alpha\).  In their embedding, the translation
of \textbf{unbox}, which corresponds to the elimination rule for
\(\square\), is slightly more involved than ours, since they use
\textbf{run} and the CSP operator.  In our embedding, on the other
hand, \textbf{unbox} is expressed uniformly by \(M \app B\), which
corresponds to the elimination rule for \(\forall\).

\subsection{Comparing with \( \lambda^\alpha \) and \( \lambda^i \)}
\label{sec:compare-alpha}

Comparing \( \sname \) with \( \lambda^\alpha \)~\cite{604134}, we can
point out two differences: the meaning of \textbf{run} and the absence
of CSP primitive.  

In \( \lambda^\alpha \), \textbf{run} is a primitive, while, in \( \sname \), \( \textbf{run } M \) is
defined as a syntax sugar for \( M \app \varepsilon \).
The following rules are typing rules for \textbf{run} in \( \lambda^\alpha \) and \( \sname \), respectively.
\begin{center}
\infrule{
  \Gamma \vdash^A M : (\alpha) \langle \tau \rangle^\alpha
}{
  \Gamma \vdash^A \textbf{run } M : (\alpha) \tau
}
\qquad
\infrule{
  \Gamma \vdash^A M : \forall \alpha. \tau
}{
  \Gamma \vdash^A M \varepsilon : \tau \Subst{\alpha := \varepsilon}
}
\end{center}
Aside from the presence of a binder \((\alpha)\), which is not essential,
an important difference is how code type constructors are removed in
the conclusion.  In \( \lambda^\alpha \), \textbf{run} removes only
the \emph{outermost} bracket annotated with \( \alpha \), while, in \(
\sname \), \( M \app \varepsilon \) removes \emph{all} code-type
constructor \( \Seal{\alpha}{} \) in \(\tau\).  This difference in
typing rules is also reflected in reductions.
\begin{eqnarray*}
  \textbf{run } (\alpha) \langle v \rangle^\alpha & \red & (\alpha) v \\
  (\Lambda \alpha. v) \app \varepsilon & \red & v \Subst{\alpha := \varepsilon}
\end{eqnarray*}
(Here we assume that \( v \) does not contain the CSP constructor to
simplify the argument---See \cite{604134} for the complete
definition.)  So, it does not seem very easy to give an embedding of
either direction.

From a practical point of view, we do not think this difference is
very significant.  It is not clear when one wants to share one
environment classifier or transition variable among different stages.
If the use of classifiers or transition variables is staged, as we
discussed in Section~\ref{sec:erasure-semantics}, then the difference
is very little.  In fact, the current implementation of MetaOCaml is
based on \(\lambda^i\), which can be considered a subset of both
\(\lambda^\alpha\) and \(\sname\) (if CSP is ignored), and this fact
shows that the difference is practically insignificant.

As another difference, \( \lambda^\alpha \) allows the CSP constructor
\( \% \) to be applied to any terms to embed the value of the term
inside a quotation.  It is easy to see that \( \phi @ \varepsilon
\vdash^\alpha \phi \) is not provable in general and so such a
universal CSP operator would \emph{not} be expressible in \( \sname
\).  However, we can support CSP for many types.  CSP for values of
base types such as integers and Booleans is easy.  CSP for function
closures is also possible if they do not contain open code in their
bodies or environment.  We can deal with CSP for \emph{closed code},
i.e., the terms which have types of the form \( \forall
\alpha. \Seal{\alpha}{\tau} \) with \( \alpha \notin \FMV(\tau) \), as
syntactic sugar in \( \sname \).  The following rules are for CSP in
\( \lambda^\alpha \) and for CSP as syntactics sugar in \( \sname \).
\begin{center}
\infrule{
  \Env \vdash^A M : \tau
}{
  \Env \vdash^{A \alpha} \%\ M : \tau
}
\qquad
\infrule{
  \Env \vdash^A M : \forall \beta. \Seal{\beta}{\tau} \AND
    \beta \notin \FMV(\tau)
}{
  \Env \vdash^{A \alpha} \Gen{\beta}{(\Prev{\alpha}{(M \app \alpha \beta)})} :
    \forall \beta. \Seal{\beta}{\tau}
}
\end{center}
The only problematic case is CSP for open code and, as mentioned
above, functions containing open code, but we think it is rarely
needed.

% In fact, we can embed \( \lambda^\square \) into \( \sname \)
% and use a power function \( \textrm{power}_\forall \) at further stages,
% both of which need CSP in \( \lambda^\alpha \).

\(\lambda^i\)~\cite{Calcagno2004} is developed as a subset of
\(\lambda^\alpha\) where type inference is possible (although there
are slight differences in syntax and typing).  The difference between
\( \sname \) and \( \lambda^i \), is smaller than that between \(
\sname \) and \( \lambda^\alpha \).  In fact, we can construct an
embedding that preserves typing from \( \lambda^i \) (without the CSP
operator) into \( \sname \), by observing that the executable code
type \( \langle \tau \rangle\) in \( \lambda^i \) corresponds to%
\( \forall \alpha. \Seal{\alpha}{\tau} \) (where \( \alpha \notin
\FMV(\tau) \)) in \(\sname\).  Figure~\ref{fig:embedding-lambda-i}
  shows the complete description of the embedding.  Precisely
  speaking, it takes a type derivation rather than terms: For example,
  the rule for \textbf{open} means
\[
  \Biggl[\!\!\Biggl[
    \frac{\Gamma \vdash^A M : \langle \tau \rangle}{\Gamma \vdash^A \textbf{open}\;M : \langle \tau \rangle^\alpha}
  \Biggr]\!\!\Biggr]
    = \EmbIT{\Gamma \vdash^A M : \langle \tau \rangle} \app \alpha,
\]
(note that \( \alpha \) does not appear in the term \(
\textbf{open}\;M \) and it comes from the derivation) and the
rule for \textbf{close} means
\[
  \Biggl[\!\!\Biggl[ \frac{\Gamma \vdash^A M : \langle \tau \rangle^\alpha \qquad \alpha \notin \textrm{FV}(\Gamma, A, \tau)}
              {\Gamma \vdash^A \textbf{close}\; M : \langle \tau \rangle} \Biggr]\!\!\Biggr]
    = \Lambda \alpha. \EmbIT{\Gamma \vdash^A M : \langle \tau \rangle^\alpha}.
\]
% The reason why such a treatment is needed is that terms of \( \sname \) are explicitly
% typed but terms of \( \lambda^i \) are not, i.e., there is one-to-one correspondence
% between terms and type derivation in \( \sname \) but not in \( \lambda^i \).

\begin{thm}
For any \(\lambda^i\) term \(M\), which does not contain \%,
if \( \Gamma \vdash^A M : \tau \) in \( \lambda^i \), then
\( \EmbIT{\Gamma} \vdash^A \EmbIT{M} : \EmbIT{\tau} \).
\end{thm}
\begin{proof}
Easy induction on the type derivation.
\end{proof}

\begin{figure}
\begin{eqnarray*}
  \EmbIT{b} & = & b \\
  \EmbIT{\tau \to \sigma} & = & \EmbIT{\tau} \to \EmbIT{\sigma} \\
  \EmbIT{\langle \tau \rangle^\alpha} & = & \Seal{\alpha}{\EmbIT{\tau}} \\
  \EmbIT{\langle \tau \rangle} & = & \forall \alpha. \Seal{\alpha}{\EmbIT{\tau}}
    \quad \textrm{(where \( \alpha \notin \textrm{FV}(\tau) \))} \\
  \\
  \EmbIT{x}^A & = & x \\
  \EmbIT{\lam x. M}^A & = & \lam x : \tau. \EmbIT{M}^A
    \quad \textrm{(\( \tau \) should be chosen appropriately)} \\
  \EmbIT{M \app N}^A & = & \EmbIT{M}^A \app \EmbIT{N}^A \\
  \EmbIT{\langle M \rangle} & = & \Next{\alpha}{\EmbIT{M}}
    \quad \textrm{(\( \alpha \) should be chosen appropriately)} \\
  \EmbIT{\;\tilde{}\; M} & = & \Prev{\alpha}{\EmbIT{M}}
    \quad \textrm{(\( \alpha \) should be chosen appropriately)} \\
  \EmbIT{\textbf{run}\; M} & = & \EmbIT{M} \; \varepsilon \\
  \EmbIT{\textbf{open}\; M} & = & \EmbIT{M} \; \alpha
    \quad \textrm{(\( \alpha \) should be chosen appropriately)} \\
  \EmbIT{\textbf{close}\; M} & = & \Lambda \alpha. \EmbIT{M}
    \quad \textrm{(\( \alpha \) should be chosen appropriately)} \\
  \\
  \EmbIT{x_1 : \tau_1^{A_1}, \dots x_n : \tau_n^{A_n}} & = &
    x_1 : \EmbIT{\tau_1} @ A_1, \dots, x_n : \EmbIT{\tau_n} @ A_n
\end{eqnarray*}
\caption{Embedding from \(\lambda^i\) without CSP primitive to
  \(\sname\).  Precisely speaking, in order to recover \( \tau \) and
  \( \alpha \), which appear only on the right hand side, this
  embedding takes type derivations of \( \lambda^i \) rather than
  terms as an input.}
\label{fig:embedding-lambda-i}
\end{figure}

Preservation of the semantics is hard to discuss precisely.  First
of all, the semantics of \(\lambda^i\) is not given
in~\cite{Calcagno2004} in spite of the subtle syntactical differences
between \(\lambda^\alpha\) and \(\lambda^i\).  However, as far as we
guess, the semantics of \(\lambda^i\) seems very close to the
erasure semantics in Section~\ref{sec:erasure-semantics}, and then we
expect to have preservation of the semantics.

\section{Related Work}
\subsection*{Multi-Stage Calculi Based on Modal Logics and Their Extensions.}
Our work can be considered a generalization of the previous work on
the Curry-Howard isomorphism between multi-stage calculi and modal
logics~\cite{Davies-Atemporallogicappro,Davies2001,Yuse2006}.
As we have seen in Section~\ref{sec:embedding},
there are embeddings from \( \lambda^\bigcirc \) and \( \lambda^\square \) to \( \sname \).

The restriction of \(\lambda^\square\) that all code be closed
precludes the definition of a code generator like
\(\texttt{power}_\forall\), which generates both efficient and
runnable code.  Nanevski and Pfenning~\cite{NanevskiPfenning05JFP}
have extended \(\lambda^\square\) with the notion of names, similar to
the symbols in Lisp, and remedied the defect of \(\lambda^{\square}\)
by allowing newly generated names (not variables) to appear in closed
code.

Taha and Sheard~\cite{TahaSheard00TCS} added \textbf{run} and CSP to
\(\lambda^\bigcirc\) and developed MetaML, but its type system was not
strong enough---\textbf{run} may fail at run-time.  Then, Moggi, Taha,
Benaissa, and Sheard~\cite{MoggiTahaBenaissaSheard99ESOP} developed
the calculus AIM (``An Idealized MetaML''), in which there are types
for both open and closed code; it was then simplified to
\(\lambda^{\textsf{BN}}\), which replaced closed code types with
closedness types for closed terms that are not necessarily code.  Both
calculi are based on categorical models and have sound type systems.
The notion of \(\alpha\)-closedness in \(\lambda^\alpha\) can be
considered a generalization of \(\lambda^{\textsf{BN}}\)'s closed
types.  In fact, the typing rule for \textbf{run} in
\(\lambda^{\textsf{BN}}\) is similar to the one in \(\lambda^\alpha\).
Although some of these calculi have sound type systems, it is hard to
regard them as logic, mainly due to the presence of CSP, which delays
the stage of the type judgment to any later stage, and the typing rule
for \textbf{run} (as discussed in Section~\ref{sec:compare-alpha}).

% One nice property of $ \lambda^{\alpha} $ is that a program can be
% executed without exploiting information on classifiers; in other
% words, classifiers can be erased after typechecking.  Although our
% calculus $ \sname $ does not have this ``erasure property,'' due to
% the presence of abstraction/instantiation of transition variables, by
% \changed{restricting transition variables to be staged,
% % restricting \(\forall\)-types to be of the form $ \forall \alpha
% % . \Seal{\alpha}{\tau} $ where $ \alpha \notin \FMV(\tau) $,
% information on transition variables can be mostly erased (See Section~\ref{sec:erasure-semantics}).}
% Under this restriction, the only information to be left after erasure is the
% length \(n\) of \(A\) in \(M \app A\), which only duplicates
% \(\Next{}{}\) at the head of the value of \(M\) \(n\) times.  This
% restriction, which resembles one in $ \lambda^i $~\cite{Calcagno2004},
% still allows embedding of $ \lambda^\bigcirc $ and $ \lambda^\square $
% and $ \texttt{power}_\forall $ (by inlining \(\texttt{power}_2\) into
% the body of it).

More recently, Yuse and Igarashi have proposed the calculus $
\lambda^{\bigcirc\square} $~\cite{Yuse2006} by combining
\(\lambda^\bigcirc\) and \(\lambda^\square\), while maintaining the
Curry-Howard isomorphism.  The main idea was to consider LTL with
modalities ``always'' (\(\square\)) and ``next'' (\(\bigcirc\)), which
represent closed and open code types, respectively.  It is similar to
AIM in this respect.  Although \(\lambda^{\bigcirc\square}\) is based
on logic, it cannot be embedded into \(\sname\) simply by combining
the two embeddings above.  In fact, in \(\lambda^{\bigcirc\square}\),
both directions of \(\square \bigcirc \tau \leftrightarrow \bigcirc
\square \tau \) are provable, whereas neither direction of \((\forall
\alpha. \Seal{\alpha}{\Seal{\beta}{\tau}}) \leftrightarrow
\Seal{\beta}{\forall \alpha. \Seal{\alpha}{\tau}} \) is provable in
\(\sname\).  However, in $ \lambda^{\bigcirc\square} $ it seems
impossible to program a code generator like
\(\texttt{power}_\forall\), which generates specialized code used at
any stage; the best possible one presented can generate specialized
code used only at any \emph{later} stage, so running the specialized
code is not possible.

It is considered not easy to develop a sound type system for staging
constructs \emph{with side effects}.  Calcagno, Moggi, and Sheard
developed a sound type system for a multi-stage calculus with
references using closed types~\cite{CalcagnoMoggiSheard03JFP}.  It is
interesting to study whether their closedness condition can be relaxed
by using \(\alpha\)-closedness.

\subsection*{Other Multi-Stage Calculi.}

Kim, Yi, and Calcagno's $ \lambda^{poly}_{open} $~\cite{1111060} is a
rather powerful multi-stage calculus with open and closed code
fragments, intentionally variable-capturing substitution, lifting
values into code, and even references and ML-style type inference.
The type structure of $ \lambda^{poly}_{open} $ is rather different: 
Since variables in code values can escape from its initial binder's 
scope and get captured by other binders, a
code type records the names of free variables and their types, as well
as the type of the whole code.  It is not clear how (a pure fragment
of) the calculus can be related to other foundational calculi;
possible directions may be to use the calculus of
contexts~\cite{SatoSakuraiKameyama01JFLP} by Sato, Sakurai, and
Kameyama, and the contextual modal type theory by Nanevski, Pfenning,
and Pientka~\cite{NanevskiPfenningPientka08TOCL}.

Viera and Pardo~\cite{VieraPardo06GPCE} have proposed a multi-stage
language with intensional code analysis, that is, pattern matching on
code.  The language requires typechecking at run-time.

More recently, Kameyama, Kiselyov, and
Shan~\cite{KameyamaKiselyovShan09PEPM} have developed an extension of
(a two-level version of) \(\lambda^\bigcirc\) with control operators
shift/reset~\cite{DanvyFilinski90LFP}, which enable an interesting
pattern of code generation such as
let-insertion~\cite{LawallDanvy94LFP} in the direct style.  It will be
interesting to investigate how this calculus extends to dynamic code execution
(i.e., \textbf{run}).

\subsection*{Modal Logics.}

As we discussed above, the $ \square $-fragment of modal logic,
the $ \bigcirc $-fragment of LTL can be
embedded into our logic, and the $ \square\bigcirc $-fragment
of LTL and our logic is incomparable.

Our logic has three characteristic features: (1) it is multi-modal,
(2) it has universal quantification over modalities and (3) modal
operators are ``relative'', meaning their semantics depends on the
possible world at which they are interpreted.  Most other logics
do not have all of these features.

Dynamic logic~\cite{dynamiclogic} is a multi-modal logic for reasoning
about programs.  Its modal operators are $ [\alpha] $ for each program
$ \alpha $, and $ [\alpha] \phi $ means ``when $ \alpha $ halts, $
\phi $ must stand after execution of $ \alpha $ from the current
state''.  Dynamic logic is multi-modal and its modal operators are
``relative'', but does not have quantification over programs.
Therefore, there is no formula in Dynamic logic which would correspond
to $ \forall \alpha . \Seal{\alpha} {\Seal{\alpha}{\phi}} $.  There
is, however, a formula which is expressive in Dynamic logic but not in
our logic: e.g., a Dynamic logic formula $ [\alpha^*] \phi $,
which means intuitively $ \phi \wedge [\alpha] \phi \wedge [\alpha]
[\alpha] \phi \wedge \dots $, cannot be expressed in our logic.

Hybrid logic~\cite{hybridlogic} is a modal logic with a new kind of
atomic formula called \emph{nominals}, each of which must be true
exactly one state in any model (therefore, a nominal names a state).
For each nominal $ i $, $ @_i $ is a modal operator and $ @_i \phi $
means ``$\phi$ stands at the state denoted by $ i $''.  Hybrid logic
has a universal quantifier over nominals (and another binder $
\downarrow $: $ \downarrow x . \phi $ means ``let \(x\) stand for the
nominal for the current world, then $ \phi $ stands'').  Hybrid logic
differs from our logic, in that modal operators $ @_i $ indicate
worlds directly, hence are not ``relative''.  In Hybrid logic, $ @_i
@_j \phi \leftrightarrow @_j \phi $, but $
\Seal{\alpha}{\Seal{\beta}{\phi}} $ and $ \Seal{\beta}{\phi} $ are not
equivalent in our logic.

\section{Conclusion and Future Work}
We have studied a logical aspect of environment classifiers by
developing a simply typed multi-stage calculus \(\sname\) with
environment classifiers.  This calculus corresponds to a multi-modal
logic with quantifier over transitions by the Curry-Howard isomorphism
and satisfies time-ordered normalization as well as basic properties
such as subject reduction, confluence, and strong normalization.  The
classical proof system is sound and complete with respect to the
Kripke semantics.
% based on automata, which gives a
% logical foundation for environment classifiers with two
% correspondences, classifiers to words and stages to states decided by
% words.  
Our calculus simplifies the previous calculus \(\lambda^\alpha\) of
environment classifiers by reducing \textbf{run} and some uses of CSP
to an extension of another construct.  We believe our work helps
clarify the semantics of environment classifiers.

We have also studied evaluation of (a slight extension of)
\(\sname\) and shown staged execution of a program is
possible.  Also, it is shown that erasure execution is possible for a
certain subset, which is close to \(\lambda^i\), the basis of
MetaOCaml.

% has four operators, which are quoting,
% un-quoting, closing for a transition variable and opening by a
% transition, and CSP for closed code fragment and \texttt{run} can be
% regarded as abbreviation using opening operator.

From a theoretical perspective, it is interesting to study the
semantics of the intuitionistic version of the logic, as mentioned
earlier, and also the calculus corresponding to the classical version
of the logic.  It is known that the naive combination of staging
constructs and control operators is problematic since bound variables
in quotation may escape from its scope by a control operator.  We
expect that a logical analysis, like the one presented here
and Reed and Pfenning~\cite{ReedPfenning07M4M}, will help
analyze the problem.

From a practical perspective, one feature missing from \(\sname\) is
CSP for all types.  As argued in the introduction, we think typical
use of CSP is rather limited and so easy to support.  Type inference
for full \(\sname\) would not be possible for the same reason as
\(\lambda^\alpha\)~\cite{Calcagno2004}.  However, it would be easy
to applying type inference for \(\lambda^i_{let}\)~\cite{Calcagno2004}
to a similar subset of \(\sname\).

\subsection*{Acknowledgments.}  This work was begun while the first
author was at Kyoto University.  We would like to thank Lintaro Ina,
Naoki Kobayashi, Ryosuke Sato, Naokata Shikuma, and anonymous
reviewers for useful comments.  This work was supported in
part by Grant-in-Aid for Scientific Research No.\ 21300005.

\bibliographystyle{splncs}
\bibliography{database,igarashi}

\newpage
\appendix
\section{Proof of Time Ordered Normalization (Theorem~\ref{thm:TON})}
\label{sec:proof-of-TON}
First, we give an inductive characterization of the class of
\(T\)-normal terms.  A judgment of the form \( \Delta \vdash
\Downarrow^T M \), where \( \Delta \) is a finite set of transition
variables, means ``\( M \) is \( T \)-normal under bound transition
variables \( \Delta \).''  We distinguish transition variables bound
outside of \(M\) since they are interpreted as the empty transition in
annotated reduction.  We also need an auxiliary judgment of the form
\( \Delta \vdash \bigtriangledown^T M \), read ``\( M \) is
\(T\)-neutral under bound transition variables \( \Delta \)''.  We
write \( T \Subst{\Delta := \varepsilon} \) for \( T \Subst{\alpha_1
  := \varepsilon} \dots \Subst{\alpha_n := \varepsilon} \), where \(
\Delta = \{ \alpha_1, \dots, \alpha_n \} \).  The inference rules for
these judgments are shown in Figure~\ref{fig:normal-form}.

\begin{figure}[b]
  \centering
  \typicallabel{Axiom}
\infrule{
  \varepsilon \nleq T \Subst{\Delta := \varepsilon}
}{
  \Delta \vdash \bigtriangledown^T M
}
\qquad
\infrule{
  M = x \textrm{ or } M_0 \app M_1 \textrm{ or } M_0 A
}{
  \Delta \vdash \bigtriangledown^T M
}
\\[1ex]
\infrule{
}{
  \Delta \vdash \Downarrow^T x
}
\qquad
\infrule{
  \Delta \vdash \Downarrow^T M
}{
  \Delta \vdash \Downarrow^T \lambda x : \tau. M
}
\qquad
\infrule{
  \Delta \vdash \Downarrow^T M \AND
  \Delta \vdash \bigtriangledown^T M \AND
  \Delta \vdash \Downarrow^T N
}{
  \Delta \vdash \Downarrow^T M \app N
}
\\[1ex]
\infrule{
  \Delta \vdash \Downarrow^T M
}{
  \Delta \vdash \Downarrow^{\alpha T} \Next{\alpha}{M}
}
\qquad
\infrule{
  \Delta \vdash \Downarrow^T M \AND
  \Delta \vdash \bigtriangledown^T M
}{
  \Delta \vdash \Downarrow^{\alpha^{-1} T} \Prev{\alpha}{M}
}
\\[1ex]
\infrule{
  \Delta, \alpha \vdash \Downarrow^T M \AND
  \alpha \notin \FMV(T) \cup \Delta
}{
  \Delta \vdash \Downarrow^T \Lambda \alpha. M
}
\qquad
\infrule{
  \Delta \vdash \Downarrow^T M \AND
  \Delta \vdash \bigtriangledown^T M
}{
  \Delta \vdash \Downarrow^T M \app A
}
  \caption{\(T\)-neutral terms and \(T\)-normal terms.}
  \label{fig:normal-form}
\end{figure}

We first prove that this proof system characterize \(T\)-normal forms
as in Lemma~\ref{lem:normal-form}, which is obtained as a special case of
the lemma below.  In what follows, \(\FMV(T)\) denotes the set of
transition variables in the path \(T\).

\begin{lem}\label{lem:normal-form-aux}
  Let \( M \) be a typable term, \( T \) be a path, \(\Delta\) be a
  finite set of transition variables.  Then, the following two
  conditions are equivalent:
\begin{enumerate}[\em(1)]
\item For any term \(N\) and path \(U\), if \( M \stackrel{U}{\red} N
  \), then \( U \Subst{\Delta := \varepsilon} \nleq T \Subst{\Delta :=
    \varepsilon} \).
\item \( \Delta \vdash \Downarrow^T M \).
\end{enumerate}
\Proof
We show \( (1) \implies (2) \) by induction on \( M \).\medskip

\begin{asparaitem}
\item Case \( M = x \): \( \Delta \vdash \Downarrow^T M \) trivially holds.
\medskip
\item Case \( M = \Next{\alpha}{M'} \): We first show that \(M'
  \stackrel{U}{\red} N' \) implies \( U \Subst{\Delta := \varepsilon}
  \nleq (\alpha^{-1} T) \Subst{\Delta := \varepsilon} \) for any
  \(N'\) and \(U\).  Assume \( M' \stackrel{U}{\red} N' \).  Then, we
  have \( \Next{\alpha}{M'} \stackrel{\alpha U}{\red}
  \Next{\alpha}{N'} \).  By (1), we have \( (\alpha U) \Subst{\Delta
    := \varepsilon} \nleq T \Subst{\Delta := \varepsilon} \) and then
  \( U \Subst{\Delta := \varepsilon} \nleq (\alpha^{-1} T)
  \Subst{\Delta := \varepsilon} \).
\noindent{m}
  By the induction hypothesis, we have \( \Delta \vdash
  \Downarrow^{\alpha^{-1} T} M' \), so \( \Delta \vdash \Downarrow^T
  \Next{\alpha}{M'} \).
  \medskip
\item Case \( M = \Prev{\alpha}{M'} \): Similarly to the case above,
  we have \( \Delta \vdash \Downarrow^{\alpha T} M' \).  What remains
  to show is \( \Delta \vdash \bigtriangledown^{\alpha T} M' \).

  Since \(M\) is typable, a possible form of \(M'\) is a variable, an
  application, an instantiation, or \(\Next{\alpha} M''\) for some
  \(M''\).  In the first three cases, \( \Delta \vdash
  \bigtriangledown^{\alpha T} M' \) is trivial.  If \(M =
  \Next{\alpha} M''\), then we have \( M' \stackrel{\alpha^{-1}}{\red}
  M''\).  Then, by (1), we have \( \alpha^{-1} \Subst{\Delta :=
    \varepsilon} \not \le T \Subst{\Delta := \varepsilon} \) and
  then \( \varepsilon \not \le (\alpha T) \Subst{\Delta :=
    \varepsilon} \).  By the first rule for \(\cdot \vdash
  \bigtriangledown^{(\cdot)} \cdot \), we have \( \Delta \vdash
  \bigtriangledown^{\alpha T} M' \).
%
% By the typability of \( M \), we have \( M' = \Next{\alpha}{M''} \).  Therefore \( M \stackrel{\alpha^{-1}}{\red} M'' \).
% Because \( \varepsilon \le (\alpha T) \Subst{\Delta := \varepsilon} \), it is the case that \( \alpha^{-1} \Subst{\Delta := \varepsilon}
% \le T \Subst{\Delta := \varepsilon} \).  This contradict to the assumption.
\medskip\

\item Case \( M = \Lambda \alpha. M' \): 

  Assume \( M' \stackrel{U}{\red} N' \) and \( \alpha \notin \FMV(T)
  \cup \Delta \).  Then, we have \( \Lambda \alpha. M' \stackrel{U
    \Subst{\alpha := \varepsilon}}{\red} \Lambda \alpha. N' \).  By
  (1), \( (U \Subst{\alpha := \varepsilon}) \Subst{\Delta :=
    \varepsilon} \nleq T \Subst{\Delta := \varepsilon} \).  We have \(
  (U \Subst{\alpha := \varepsilon}) \Subst{\Delta := \varepsilon} = U
  \Subst{\Delta, \alpha := \varepsilon} \) and \( T \Subst{\Delta,
    \alpha := \varepsilon} = T \Subst{\Delta := \varepsilon}\) because
  \( \alpha \notin \FMV(T) \).  So, \( U \Subst{\Delta, \alpha :=
    \varepsilon} \nleq T \Subst{\Delta, \alpha := \varepsilon} \).
  Then, we have \( \Delta, \alpha \vdash \Downarrow^{T} M' \) by
  the induction hypothesis.  So, \( \Delta \vdash \Downarrow^T \Lambda
  \alpha. M' \) as required.
\end{asparaitem}

Other cases are similar.

The proof of \( (2) \Rightarrow (1) \) is by easy induction on the structure of the derivation \( \Delta \vdash \Downarrow^T M \).
% (this is equivalent to the induction on the structure of \( M \)).
We show the case \( M = \Prev{\alpha}{M'} \) as a representative case.

Assume \( \Delta \vdash \Downarrow^T \Prev{\alpha}{M'} \).  By definition, we have \( \Delta \vdash \Downarrow^{\alpha T} M' \)
and \( \Delta \vdash \bigtriangledown^{\alpha T} M' \).  Assume \( \Prev{\alpha}{M'} \stackrel{U}{\red} N \).  There are two cases.
\medskip
\begin{asparaitem}
\item Case \( \Prev{\alpha}{M'} \stackrel{U}{\red} \Prev{\alpha}{N'} \) with \( M' \stackrel{\alpha U}{\red} N' \):
Because \( \Delta \vdash \Downarrow^{\alpha T} M' \), by induction hypothesis we have \( (\alpha U) \Subst{\Delta := \varepsilon} \nleq
(\alpha T) \Subst{\Delta := \varepsilon} \).
This implies that \( U \Subst{\Delta := \varepsilon} \nleq T \Subst{\Delta := \varepsilon} \).
\\
\item Case \( M' = \Next{\alpha}{M''} \) and \(
  \Prev{\alpha}{\Next{\alpha}{M''}} \stackrel{\alpha^{-1}}{\red} M''
  \): Since \( \Delta \vdash \bigtriangledown^{\alpha T}
  \Next{\alpha}{M''} \), we have \( \varepsilon \nleq (\alpha T)
  \Subst{\Delta := \varepsilon} \).  This equation is equivalent to \(
  \alpha^{-1} \Subst{\Delta := \varepsilon} \nleq T \Subst{\Delta :=
    \varepsilon} \). \qed
\end{asparaitem}
\end{lem}

\begin{lem}
\label{lem:normal-form}
Suppose  \( M \) is typable. Then, \( M \) is \(T\)-normal if and only if
\( \emptyset \vdash \Downarrow^T M \).
\end{lem}
\begin{Proof}
Immediate from Lemma~\ref{lem:normal-form-aux}.
\end{Proof}

Now, we prove that reduction preserves \(T\)-normality and
\(T\)-neutrality, as in
Lemma~\ref{lem:normality-neutrality-preservation}, from which
Theorem~\ref{thm:TON} immediately follows.  Before that, we prove that
(term and transition) substitution preserves \(T\)-normality and
\(T\)-neutrality under a certain condition.

\begin{lem}
\label{lem:staged-substitution}
Suppose \( \Gamma, x : \sigma @ B \vdash^A M : \tau \) and \( \Gamma
\vdash^B N : \sigma \) and \( B \Subst{\Delta := \varepsilon} \nleq
(AT) \Subst{\Delta := \varepsilon} \).  If \( \Delta \vdash
\Downarrow^T M \) and \( \Delta \vdash \Downarrow^{B^{-1} A T} N \),
then \( \Delta \vdash \Downarrow^T (M \Subst{x := N}) \).  Similarly,
if \( \Delta \vdash \bigtriangledown^T M \), then \( \Delta \vdash
\bigtriangledown^T (M \Subst{x := N}) \).
\proof
By induction on \( M \).  We show only the main cases.\medskip
\begin{asparaitem}
\item Case \( M = y \): The subcase \( x \neq y \) is trivial.  
Assume \( x = y \).
Because \( \Gamma, x : \sigma @ B \vdash^A x : \tau \), it is the case that \( A = B \).
So \( B^{-1} A T = T \) and \( \Delta \vdash \Downarrow^{T} N \).
Moreover, we have \( \Delta \vdash \bigtriangledown^T N \) because \( \varepsilon \nleq T \Subst{\Delta := \varepsilon} \),
which follows from \( B \Subst{\Delta := \varepsilon} \nleq (AT) \Subst{\Delta := \varepsilon} \) and \( B = A \).
\medskip
\item Case \( M = \Next{\alpha}{M'} \): From the type derivation of \( M \), we have \( \tau = \Seal{\alpha}{\tau_0} \) and \( \Gamma, x : \sigma @ B \vdash^{A \alpha} M' : \tau_0 \).  From \( \Delta \vdash \Downarrow^T \Next{\alpha}{M'} \),
we obtain \( \Delta \vdash \Downarrow^{\alpha^{-1} T} M' \).  Moreover, \( B \Subst{\Delta := \varepsilon} \nleq 
(AT)\Subst{\Delta := \varepsilon} = 
((A \alpha) (\alpha^{-1} T))
\Subst{\Delta := \varepsilon} \) holds.  By the induction hypothesis,
we now have \( \Delta \vdash \Downarrow^{\alpha^{-1} T} (M' \Subst{x := N}) \).
Therefore \( \Delta \vdash \Downarrow^T (\Next{\alpha}{M'}) \Subst{x := N} \).

If \( \Delta \vdash \bigtriangledown^T \Next{\alpha}{M'} \), it is the case that \( \varepsilon \nleq T \Subst{\Delta := \varepsilon} \).
Therefore \( \Delta \vdash \bigtriangledown^T (\Next{\alpha}{M'}) \Subst{x := N} \).
\medskip
\item Case \( M = \Lambda \alpha. M' \): Without loss of generality, we can assume that \( \alpha \notin \FMV(A) \cup \FMV(N) \cup
\FMV(\Gamma) \cup \FMV(T) \cup \FMV(B) \cup \Delta \).  From the type derivation of \( M \), we have \( \tau = \forall \alpha. \tau_0 \)
and \( \Gamma, x : \sigma @ B \vdash^A M' : \tau_0 \).  From \( \Delta \vdash \Downarrow^T M \), we also have \( \Delta, \alpha \vdash \Downarrow^T M' \).  Because \( \alpha \) is fresh and \( B \Subst{\Delta := \varepsilon} \nleq
(AT) \Subst{\Delta := \varepsilon} \), we have \( B \Subst{\Delta, \alpha := \varepsilon} \nleq
(AT) \Subst{\Delta, \alpha := \varepsilon} \).  By the induction hypothesis, we have \( \Delta, \alpha \vdash \Downarrow^T
M' \Subst{x := N} \), which implies \( \Delta \vdash \Downarrow^T (\Lambda \alpha. M') \Subst{x := N} \).

If \( \Delta \vdash \bigtriangledown^T \Lambda \alpha. M' \), it is the case that \( \varepsilon \nleq T \Subst{\Delta := \varepsilon} \).
Therefore \( \Delta \vdash \bigtriangledown^T (\Lambda \alpha. M') \Subst{x := N} \). \qed
\end{asparaitem}
\end{lem}

\begin{lem}
\label{lem:staged-transition-substitution}
Suppose \( \Gamma \vdash^A M : \tau \) and \( \alpha \notin \FMV(AT) \).
\begin{enumerate}[\em(1)]
\item If \( \Delta, \alpha \vdash \Downarrow^T M \),
then \( \Delta \vdash \Downarrow^{T \Subst{\alpha := B}} (M \Subst{\alpha := B}) \); and 
\item if \( \Delta, \alpha \vdash \bigtriangledown^T M \),
then \( \Delta \vdash \bigtriangledown^{T \Subst{\alpha := B}} (M \Subst{\alpha := B}) \).
\end{enumerate}

\proof
  We first show (2) by case analysis on the last rule used to derive
  \( \Delta, \alpha \vdash \bigtriangledown^T M \). \medskip
\begin{asparaitem}
\item Case \( \Delta, \alpha \vdash \bigtriangledown^T M \) with \( \varepsilon \nleq T \Subst{\Delta, \alpha := \varepsilon} \):
Because \( \alpha \notin \FMV(AT) \) and \( \varepsilon \le A \), we have \( A = A' \alpha^n \) and \( T = \alpha^{-n} T' \)
and \( \alpha \notin \FMV(A') \cup \FMV(T') \) for some \( n \ge 0 \).
Then, \( \varepsilon \nleq T \Subst{\Delta, \alpha := \varepsilon} = (\alpha^{-n} T') \Subst{\Delta, \alpha := \varepsilon}
= T' \Subst{\Delta := \varepsilon} \).  
Therefore \( B^{-n} \Subst{\Delta := \varepsilon} 
\nleq (B^{-n} \Subst{\Delta := \varepsilon})(T' \Subst{\Delta := \varepsilon})
= T \Subst{\alpha := B} \Subst{\Delta := \varepsilon}\).  It follows that
 \( \varepsilon \nleq T \Subst{\alpha := B} \Subst{\Delta := \varepsilon} \) because \( B^{-n} \Subst{\Delta :=\varepsilon} \le \varepsilon \).  So we have \( \Delta \vdash \bigtriangledown^{T \Subst{\alpha := B}}
(M \Subst{\alpha := B}) \).\medskip
\item The other case is easy.\medskip
\end{asparaitem}

\noindent Then, (1) is
  proved by induction on the derivation of \( \Delta, \alpha \vdash
  \Downarrow^T M \).  We show only the main cases.\medskip
\begin{asparaitem}
\item Case \( \Delta, \alpha \vdash \Downarrow^T N_0 N_1 \): 
%We have \( \Delta, \alpha \vdash \Downarrow^T N_0 \) and
%\( \Delta, \alpha \vdash \bigtriangledown^T N_0 \) and \( \Delta, \alpha \vdash \Downarrow^T N_1 \).  By the induction hypothesis and (2),
%we have \( \Delta \vdash \Downarrow^{T \Subst{\alpha := B}} (N_0 \Subst{\alpha := B}) \) and
%\( \Delta \vdash \bigtriangledown^{T \Subst{\alpha := B}} (N_0 \Subst{\alpha := B}) \) and
%\( \Delta \vdash \Downarrow^{T \Subst{\alpha := B}} (N_1 \Subst{\alpha
%:= B}) \). So we obtain
%\( \Delta \vdash \Downarrow^{T \Subst{\alpha := B}} (N_0 N_1) \Subst{\alpha := B} \).
Applying the induction hypothesis and (2) to, respectively,
\[\eqalign{
   \Delta, \alpha &\vdash \Downarrow^T N_0\cr
   \Delta, \alpha &\vdash \bigtriangledown^T N_0\cr
   \Delta, \alpha &\vdash \Downarrow^T N_1
  }
  \qquad\hbox{results in}\qquad
  \eqalign{
  \Delta &\vdash\Downarrow^{T \Subst{\alpha := B}} (N_0 \Subst{\alpha:= B})\cr
  \Delta &\vdash\bigtriangledown^{T \Subst{\alpha := B}} (N_0\Subst{\alpha:=B})\cr
  \Delta &\vdash\Downarrow^{T \Subst{\alpha:=B}}(N_1\Subst{\alpha:= B}).
  }
\]
  So we obtain
\( \Delta \vdash \Downarrow^{T \Subst{\alpha := B}} (N_0 N_1) \Subst{\alpha := B} \).
\medskip
\item Case \( \Delta, \alpha \vdash \Downarrow^T \Prev{\beta}{M'} \): It is the case that
\( \Delta, \alpha \vdash \Downarrow^{\beta T} M' \) and \( \Delta, \alpha \vdash \bigtriangledown^{\beta T} M' \).
Because \( \Gamma \vdash^A \Prev{\beta}{M'} : \tau \),
we have \( A = A' \beta \) and \( \Gamma \vdash^{A'} M' : \Seal{\beta}{\tau} \).  Since \( \alpha \notin \FMV(AT) = \FMV(A' \beta T) \),
by applying the induction hypothesis and (2), we have \( \Delta \vdash \Downarrow^{(\beta T) \Subst{\alpha := B}} M' \Subst{\alpha := B} \) and
\( \Delta \vdash \bigtriangledown^{(\beta T) \Subst{\alpha := B}} M' \Subst{\alpha := B} \).
Therefore \( \Delta \vdash \Downarrow^{T \Subst{\alpha := B}} (\Next{\beta \Subst{\alpha := B}}{(M' \Subst{\alpha := B})}) \).\qed
\end{asparaitem}
\end{lem}

\begin{lem}\label{lem:normality-neutrality-preservation}
Suppose \( \Gamma \vdash^A M : \tau \) and \( M \stackrel{U}{\red} N \) 
and  \( U \Subst{\Delta := \varepsilon} \nleq T \Subst{\Delta := \varepsilon} \).  If \( \Delta \vdash \Downarrow^T M \), then \( \Delta \vdash \Downarrow^T N \).  Similarly, if \( \Delta \vdash \bigtriangledown^T M \), then \( \Delta \vdash \bigtriangledown^T N \).
\proof
By induction on the derivation of \( M \stackrel{U}{\red} N \).\medskip
\begin{asparaitem}
\item Case \( M = (\lambda x : \sigma. P) \app Q \stackrel{\varepsilon}{\red} P \Subst{x := Q} = N \):
It is the case that \( U = \varepsilon \) and so \(\varepsilon \nleq
T \Subst{\Delta := \varepsilon}\).
Assume \( \Delta \vdash \Downarrow^T M \).
Then we have \( \Delta \vdash \Downarrow^T P \) and \( \Delta \vdash \bigtriangledown^T P \) and
\( \Delta \vdash \Downarrow^T Q \).
% The equation \( U = \varepsilon \) and the condition implies \( \varepsilon \nleq T \Subst{\Delta := \varepsilon} \).
From  \(\varepsilon \nleq
T \Subst{\Delta := \varepsilon}\), it follows that \( A \Subst{\Delta := 
\varepsilon} \nleq (A T) \Subst{\Delta := \varepsilon} \).  By Lemma~\ref{lem:staged-substitution},
we have \( \Delta \vdash \Downarrow^T (P \Subst{x := Q}) \).
Moreover, we have \( \Delta \vdash \bigtriangledown^T N \) because \( U \Subst{\Delta := \varepsilon} = \varepsilon \nleq
T \Subst{\Delta := \varepsilon} \).
\medskip
\item Case \( M = (\Lambda \alpha. P) \app B \stackrel{\varepsilon}{\red} P \Subst{\alpha := B} = N \):
We can assume without loss of generality that \( \alpha \notin \FMV(\Gamma) \cup \FMV(A) \cup \FMV(T) \cup \Delta \).
Assume \( \Delta \vdash \Downarrow^T M \). 
Then we have \( \Delta, \alpha \vdash \Downarrow^T P \).  Because \( \alpha \notin \FMV(AT) \),
by applying Lemma~\ref{lem:staged-transition-substitution},  we have
\( \Delta \vdash \Downarrow^{T \Subst{\alpha := B}} P \Subst{\alpha := B} \). Because \( \alpha \notin \FMV(T) \),
\( T \Subst{\alpha := B} = T \).  So we have \( \Delta \vdash \Downarrow^T (P \Subst{\alpha := B}) \).

The proof of \( \Delta \vdash \bigtriangledown^T N \) is similar to the prior case since \(U = \varepsilon\).
\medskip
\item Case \( M = \Next{\alpha}{P} \stackrel{U}{\red} \Next{\alpha}{P'} = N \) with \( P \stackrel{\alpha^{-1} U}{\red} P' \):
Assume \( \Delta \vdash \Downarrow^T \Next{\alpha}{P} \).  Then \( \Delta \vdash \Downarrow^{\alpha^{-1} T} P \).
The typability of \( M \) implies the typability of \( P \).  We have \( (\alpha^{-1} U) \Subst{\Delta := \varepsilon}
\nleq (\alpha^{-1} T) \Subst{\Delta := \varepsilon} \) because \( U \Subst{\Delta := \varepsilon} \nleq T \Subst{\Delta := \varepsilon} \).
So, by the induction hypothesis, we obtain \( \Delta \vdash \Downarrow^{\alpha^{-1} T} P' \).  Therefore \( \Delta \vdash \Downarrow^{T} N \).

Assume \( \Delta \vdash \bigtriangledown^T \Next{\alpha}{P} \).  It
must be the case that \( \varepsilon \nleq T \Subst{\Delta :=
  \varepsilon} \).  Therefore \( \Delta \vdash \bigtriangledown^T
\Next{\alpha}{P'} \).
\medskip
\item Case \( M = \Lambda \alpha. P \stackrel{U}{\red} \Lambda \alpha. P' = N \) with \( P \stackrel{U'}{\red} P' \) and
\( U = U' \Subst{\alpha := \varepsilon} \): We can assume without loss of generality that
\( \alpha \notin \FMV(\Gamma) \cup \FMV(A) \cup \FMV(T) \cup \Delta \).
Moreover \( \alpha \notin \FMV(U) \) because \( U = U' \Subst{\alpha := \varepsilon} \).
It follows that
\begin{eqnarray*}
  U' \Subst{\Delta, \alpha := \varepsilon}
    & = & (U' \Subst{\alpha := \varepsilon}) \Subst{\Delta := \varepsilon} \\
    & & \qquad \textrm{(by the definition of \( U \))} \\
    & = &  U \Subst{\Delta := \varepsilon} \\
    & & \qquad \textrm{(by assumption)} \\
    & \nleq & T \Subst{\Delta := \varepsilon} \\
    & & \qquad \textrm{(because \( \alpha \notin \FMV(T) \))} \\
    & = & T \Subst{\Delta, \alpha := \varepsilon}
\end{eqnarray*}

Assume that \( \Delta \vdash \Downarrow^T M \).  Then \( \Delta, \alpha \vdash \Downarrow^T P \).
The typability of \( M \) implies the typability of \( P \).
So, by the induction hypothesis, we have \( \Delta \vdash \Downarrow^T P' \), which implies \( \Delta \vdash \Downarrow^T N \).

The proof of the remaining part is similar to other cases. \qed
\end{asparaitem}
\end{lem}

\section{List of Error-generating and Error-propagating Rules}
\label{sec:complete-rules}
\subsection*{The List of Error-generating Rules}
\ \\
\begin{center}
\infrule{
  \bred{\varepsilon}{M}{M'} \AND
  M' \notin \mathbb{Z} \AND
  \Diamond \in \{ +, -, *, = \}
}{
  \bred{\varepsilon}{M \Diamond N}{\err}
}
\qquad
\infrule{
  \bred{\varepsilon}{M}{M'} \AND
  M' \in \mathbb{Z} \AND
  \bred{\varepsilon}{N}{N'} \AND
  N' \notin \mathbb{Z} \AND
  \Diamond \in \{ +, -, *, = \}
}{
  \bred{\varepsilon}{M \Diamond N}{\err}
}
\\[1ex]
\infrule{
  \bred{\varepsilon}{M}{M'} \AND
  M' \notin \{ \textbf{true}, \textbf{false} \}
}{
  \bred{\varepsilon}{\eIf{M}{N_1}{N_2}}{\err}
}
\qquad
\infrule{
}{
  \bred{\varepsilon}{x}{\err}
}
\\[1ex]
\infrule{
  \bred{\varepsilon}{M}{M'} \AND
  M' \neq \lam x:\tau.M''
}{
  \bred{\varepsilon}{M \app N}{\err}
}
\qquad
\infrule{
  \bred{\varepsilon}{M}{M'} \AND
  M' \neq \Lambda \alpha. M''
}{
  \bred{\varepsilon}{\Ins{M}{A}}{\err}
}
\\[1ex]
\infrule{
  \bred{\varepsilon}{M}{M'} \AND
  M' \neq \Next{\alpha}{M''}
}{
  \bred{\alpha}{\Prev{\alpha}{M}}{\err}
}
\qquad
\infrule{
  \mathstrut
}{
  \bred{\varepsilon}{\Prev{\alpha}{M}}{\err}
}
\end{center}

\subsection*{The List of Error-propagating Rules}
\ \\
\begin{center}
\infrule{
  \bred{A}{M}{\err} \AND
  \Diamond \in \{ +, -, *, = \}
}{
  \bred{A}{M \Diamond N}{\err}
}
\qquad
\infrule{
  \bred{A}{M}{M'} \AND
  \bred{A}{N}{\err} \AND
  \Diamond \in \{ +, -, *, = \}
}{
  \bred{A}{M \Diamond N}{\err}
}
\\[1ex]
\infrule{
  \bred{A}{M}{\err}
}{
  \bred{A}{\eIf{M}{N_1}{N_2}}{\err}
}
\qquad
\infrule{
  \bred{\varepsilon}{M}{\textbf{true}} \AND
  \bred{\varepsilon}{N_1}{\err}
}{
  \bred{\varepsilon}{\eIf{M}{N_1}{N_2}}{\err}
}
\\[1ex]
\infrule{
  \bred{\varepsilon}{M}{\textbf{false}} \AND
  \bred{\varepsilon}{N_2}{\err}
}{
  \bred{\varepsilon}{\eIf{M}{N_1}{N_2}}{\err}
}
\qquad
\infrule{
  \bred{A}{M}{M'} \AND
  \bred{A}{N_1}{\err} \AND
  A \neq \varepsilon
}{
  \bred{A}{\eIf{M}{N_1}{N_2}}{\err}
}
\\[1ex]
\infrule{
  \bred{A}{M}{M'} \AND
  \bred{A}{N_1'}{N_1'} \AND
  \bred{A}{N_2}{\err} \AND
  A \neq \varepsilon
}{
  \bred{A}{\eIf{M}{N_1}{N_2}}{\err}
}
\qquad
\infrule{
  \bred{A}{M}{\err} \AND
  A \neq \varepsilon
}{
  \bred{A}{\lam x:\tau .M}{\err}
}
\\[1ex]
\infrule{
  \bred{A}{M}{\err}
}{
  \bred{A}{M \app N}{\err}
}
\qquad
\infrule{
  \bred{A}{M}{M'} \AND
  \bred{A}{N}{\err}
}{
  \bred{A}{M \app N}{\err}
}
\qquad
\infrule{
  \bred{A \alpha}{M}{\err}
}{
  \bred{A}{\Next{\alpha}{M}}{\err}
}
\\[1ex]
\infrule{
  \bred{A}{M}{\err}
}{
  \bred{A \alpha}{\Prev{\alpha}{M}}{\err}
}
\qquad
\infrule{
  \bred{A}{M}{\err}
}{
  \bred{A}{\Lambda \alpha. M}{\err}
}
\qquad
\infrule{
  \bred{A}{M}{\err}
}{
  \bred{A}{\Ins{M}{B}}{\err}
}
\\[1ex]
\infrule{
  \bred{\varepsilon}{M \Subst{f := \textbf{fix }f : \tau \to \sigma. M}}{\err}
}{
  \bred{\varepsilon}{\textbf{fix }f : \tau \to \sigma. M}{\err}
}
\qquad
\infrule{
  \bred{A}{M}{\err} \AND
  A \neq \varepsilon
}{
  \bred{A}{\textbf{fix } f : \tau \to \sigma. M}{\err}
}
\end{center}

\section{Proofs of Properties about mini-ML}
\label{sec:proof-miniML}
\subsection{Proof of Theorem~\ref{thm:ml-sound} (Type Soundness of \miniML)}
\label{sec:proof-sound}

\begin{lem}
\label{lem:sound2}
If $ \Gamma \vdash^{\alpha A} v : \tau $ and
$ v \in V^{\alpha A} $,
then $ \Gamma^{- \alpha} \vdash^A v : \tau $.
\begin{proof}
By induction on the structure of $ v $.
We show only cases for $ v = x $ and $ v = \Prev{\beta}{v'} $.

\begin{enumerate}[$\bullet$]
\item Case $ v = x $: We have $ x : \tau @ \alpha A \in \Gamma $.
Because $ x : \tau @ A \in \Gamma^{- \alpha} $,
$ \Gamma^{- \alpha} \vdash^A x : \tau $.
\\
\item Case $ v = \Prev{\beta}{v'} $:
We have $ A \neq \varepsilon $
because $ \Prev{\beta}{v'} \in V^{\alpha A} $.
That means $ A = A' \beta $ for some $ A' $ and
$ \Gamma \vdash^{\alpha A' \beta} \Prev{\beta}{v'} : \tau $.
So we have \( \Gamma \vdash^{\alpha A'} v' : \Seal{\beta}{\tau} \).
By the induction hypothesis,
$ \Gamma^{- \alpha} \vdash^{A'} v' : \Seal{\beta}{\tau} $.
Therefore,
by ($\Prev{}{}$), we have $ \Gamma^{- \alpha} \vdash^{A' \beta} \Prev{\beta}{v'} : \tau $.
\end{enumerate}
\end{proof}
\end{lem}

\noindent\textit{Proof of Theorem~\ref{thm:ml-sound}}.
The first part of the theorem 
\begin{quotation}
  If $ \Gamma $ is $ \varepsilon $-free and
$ \Gamma \vdash^{A} M : \tau $ and
$ \bred{A}{M}{R} $,
then $ R = v $ and $ v \in V^{A} $ for some $ v $
and $ \Gamma \vdash^{A} v : \tau $.
\end{quotation}
is proved by 
induction on the derivation of $ \bred{A}{M}{R} $ with case analysis
on the form of $M$ and the last rule used to derive  $ \bred{A}{M}{R} $.

We only show representative cases here.
\begin{enumerate}[$\bullet$]
\item Case $ M = x $
	\begin{enumerate}[$-$]
	\item Subcase \infrule{}{ \bred{A}{x}{x} } ($ A \neq \varepsilon $)
    \begin{enumerate}[(1)]
    \item Immediate.
    \end{enumerate}
  \item Subcase \infrule{}{ \bred{\varepsilon}{x}{\err} }
    \begin{enumerate}[(1)]
    \item We have $ x : \tau @ \varepsilon \in \Gamma $ because \( \Gamma \vdash^\varepsilon x : \tau \)
      but this contradicts the assumption that $ \Gamma $ is $ \varepsilon $-free.
    \end{enumerate}
  \end{enumerate}
\item Case $ M = \Prev{\alpha}{N} $
  \begin{enumerate}[$-$]
  \item Subcase
    \infrule{
      \bred{\varepsilon}{N}{\Next{\alpha}{N'}}
    }{
      \bred{\alpha}{\Prev{\alpha}{N}}{N'}
    }
    \begin{enumerate}[(1)]
    \item We have
      $ \Gamma \vdash^\varepsilon N : \Seal{\alpha}{\tau} $
      because \( \Gamma \vdash^\varepsilon \Prev{\alpha}{N} \).
    \item By the induction hypothesis,
      $ \Gamma \vdash^\varepsilon \Next{\alpha}{N'} : \Seal{\alpha}{\tau} $ and
      $ \Next{\alpha}{N'} \in V^\varepsilon $ .
    \item So we have
      $ \Gamma \vdash^\alpha N' : \tau $ and
      $ N' \in V^\alpha $.
    \end{enumerate}
  \item Subcase
    \infrule{
      \bred{A}{N}{N'}
    }{
      \bred{A \alpha}{\Prev{\alpha}{N}}{\Prev{\alpha}{N'}}
    } ($ A \neq \varepsilon $)
    \begin{enumerate}[(1)]
    \item We have
      $ \Gamma \vdash^A N : \Seal{\alpha}{\tau} $
      because \( \Gamma \vdash^A \Prev{\alpha}{N} : \tau \).
    \item By the induction hypothesis,
      $ \Gamma \vdash^A N' : \Seal{\alpha}{\tau} $ and
      $ N' \in V^A $ .
    \item Therefore,
      $ \Gamma \vdash^{A \alpha} \Prev{\alpha}{N'} : \tau $ by ($\Prev{}{}$) and
      $ \Prev{\alpha}{N'} \in V^{A \alpha} $.
    \end{enumerate}
  \item Subcase
    \infrule{
      \bred{\varepsilon}{N}{N'} \AND
      N' \neq \Next{\alpha}{N''}
    }{
      \bred{\alpha}{\Prev{\alpha}{N}}{\err}
    }
    \begin{enumerate}[(1)]
    \item We have
      $ \Gamma \vdash^\varepsilon N : \Seal{\alpha}{\tau} $
      because \( \Gamma \vdash^\varepsilon \Prev{\alpha}{N} : \tau \).
    \item By the induction hypothesis,
      $ \Gamma \vdash^\varepsilon N' : \Seal{\alpha}{\tau} $ and
      $ N' \in V^\varepsilon $, which is a contradiction.
%     \item Because $ R \neq \Next{\alpha}{N'} $ and $ R \in V^\varepsilon $,
%       we have $ R = n, \textbf{true}, \textbf{false}, \lam x : \sigma . N', \Lambda \beta . N' $ or $ \Next{\beta}{N'} $
%       ($ \beta \neq \alpha $) but none of them 
%       can have the type $ \Seal{\alpha}{\tau} $.
    \end{enumerate}
   \end{enumerate}
\item Case $ M = \Lambda \alpha . N $
  \begin{enumerate}[$-$]
  \item Subcase
    \infrule{
      \bred{A}{N}{N'}
    }{
      \bred{A}{\Lambda \alpha . N}{\Lambda \alpha . N'}
    }
    \begin{enumerate}[(1)]
    \item We have
      $ \tau = \forall \alpha . \tau' $ and
      $ \Gamma \vdash^A N : \tau' $ for some $ \tau' $ and
      $ \alpha \notin \FMV(\Gamma) \cup \FMV(A) $
      because \( \Gamma \vdash^A \Gen{\alpha}{N} : \tau \).
    \item By the induction hypothesis,
      $ \Gamma \vdash^A N' : \tau' $ and
      $ N' \in V^A $.
    \item Therefore,
      $ \Gamma \vdash^A \Lambda \alpha . N' : \tau $ by (\textsc{Gen}) and
      $ \Lambda \alpha N' \in V^A $.
    \end{enumerate}
  \end{enumerate}
\item Case $ M = N \app B $
  \begin{enumerate}[$-$]
  \item Subcase
    \infrule{
      \bred{\varepsilon}{N}{\Lambda \alpha . N'} \AND
      \bred{\varepsilon}{N' \Subst{\alpha := B}}{M}
    }{
      \bred{\varepsilon}{N \app B}{M}
    }
    \begin{enumerate}[(1)]
    \item We have
      $ \tau = \sigma \Subst{\alpha := B} $ and
      $ \Gamma \vdash^\varepsilon N : \forall \alpha . \sigma $
      for some $ \sigma $ and $ \alpha $,
      because \( \Gamma \vdash^\varepsilon N \app B : \tau \).
    \item By the induction hypothesis,
      $ \Gamma \vdash^\varepsilon
        \Lambda \alpha . N' : \forall \alpha . \sigma $.
      Therefore
      $ \Gamma \vdash^\varepsilon N' : \sigma $ and
      $ \alpha \notin \FMV(\Gamma) $.
    \item By Substitution Lemma (Lemma~\ref{lem:subst})
      and $ \alpha \notin \FMV(\Gamma) $,
      $ \Gamma \vdash^\varepsilon N' \Subst{\alpha := B} : \tau $.
    \item By the induction hypothesis,
      $ \Gamma \vdash^\varepsilon M : \tau $ and
      $ M \in V^\varepsilon $.
    \end{enumerate}
  \item Subcase
    \infrule{
      \bred{A}{N}{N'}
    }{
      \bred{A}{N \app B}{N' \app B}
    } ($ A \neq \varepsilon $)
    \begin{enumerate}[(1)]
    \item We have
      $ \tau = \sigma \Subst{\alpha := B} $ and
      $ \Gamma \vdash^A N : \forall \alpha . \sigma $
      for some $ \sigma $ and $ \alpha $,
      because \( \Gamma \vdash^A N \app B : \tau \).
    \item By the induction hypothesis,
      $ \Gamma \vdash^A N' : \forall \alpha . \sigma $ and
      $ N' \in V^A $.
    \item Therefore,
      $ \Gamma \vdash^A N' \app B : \tau $ by (\textsc{Ins}) and
      $ N' \app B \in V^A $.
    \end{enumerate}
  \item Subcase
    \infrule{
      \bred{\varepsilon}{N}{N'} \AND
      N' \neq \Lambda \alpha . N''
    }{
      \bred{\varepsilon}{N \app B}{\err}
    }
    \begin{enumerate}[(1)]
    \item We have
      $ \tau = \sigma \Subst{\alpha := B} $ and
      $ \Gamma \vdash^\varepsilon N : \forall \alpha . \sigma $
      for some $ \sigma $ and $ \alpha $,
      because \( \Gamma \vdash^\varepsilon N \app B : \tau \).
    \item By the induction hypothesis,
      $ \Gamma \vdash^\varepsilon N' : \forall \alpha . \sigma $ and
      $ N' \in V^A $, which is a contradiction.
%     \item Because $ N' \neq \Lambda \alpha . R $ and $ N' \in V^\varepsilon $,
%       $ N' = n, \textbf{true}, \textbf{false}, \lam x : \tau' . R $ or $ \Next{\beta}{R} $.
%       But none of them can have the type $ \forall \alpha . \sigma $.
    \end{enumerate}
   \end{enumerate}
\end{enumerate}

Now, we prove the second part.
% Assume $ \Gamma \vdash^\varepsilon M : \Seal{\alpha}{\tau} $ and
% $ \Gamma $ is $ \varepsilon $-free and
% $ \bred{\varepsilon}{M}{v} $.
By the first part,
$ \Gamma \vdash^\varepsilon v : \Seal{\alpha}{\tau} $ and
$ v \in V^\varepsilon $.
Therefore $ v = \Next{\alpha}{v'} $ and
$ v' \in V^\alpha $.
From the typing rules, we have
$ \Gamma \vdash^\alpha v' : \tau $.
Then, we have $ \Gamma^{- \alpha} \vdash^\varepsilon v' : \tau $
by Lemma~\ref{lem:sound2}. \qed

\subsection{Proof of Theorem~\ref{thm:erasure}}
\label{sec:proof-erasure}
We first prove type soundness for the new type system
(Lemma~\ref{lem:staged-preservation}).  Although the new type system
identifies a subset of well-typed \(\sname\) terms, it has to be
proved again, since we want to gurantee that the evaluation of a term
typed under the new type system results in (if converges) a value that
is also typed under it.  We start with proving various substitution lemmas.

% two kinds of Substitution Lemmas (for substitution of terms and transition variables) and
% Preservation.

% First we show Substitution Lemma of terms and transition variables.
% It is easy to show Substitution Lemma of terms.

\begin{lem}[Term Substitution Preserves Typing]
\label{lem:subst-staged-transition}
If \( \Gamma, x : \sigma @ B; \Delta \svdash^A M : \tau \) and
\( \Gamma; \Delta \svdash^B M : \sigma \), then \( \Gamma; \Delta \svdash^A M \Subst{x := N} : \tau \).
\begin{proof}
Similar to the proof of Lemma~\ref{lem:subst}.
\end{proof}
\end{lem}

Next, we show that substitution for transition variables preserves
various kinds of well-formedness and typing.

% ,
% dividing the proof into a series of proofs of lemmas.
% The goal is Lemma~\ref{lem:transition-subst-term}.
% The ordering of lemmas is followed to the ordering of the definition of well-formedness.
\begin{lem}
\label{lem:transition-subst-transition}
If \( \Delta, \alpha @ B \svdash A \) and  \( \Delta \svdash B C \),
then \( \Delta \Subst{\alpha := C} \svdash A \Subst{\alpha := C} \).
\proof
By induction of the length of \( A \).
The base case is trivial because \( \Delta \svdash \varepsilon \)
for any \( \Delta \).
So, we show the induction step.

Let \( A = \alpha_1 \dots \alpha_n \).
Assume \( \Delta, \alpha @ B \svdash A \).
It is easy to see that \( \Delta, \alpha @ B \svdash \alpha_1 \dots \alpha_n \)
implies \( \Delta, \alpha @ B \svdash \alpha_1 \dots \alpha_{n - 1} \).
Then \( \Delta \Subst{\alpha := C} \svdash (\alpha_1 \dots \alpha_{n - 1})
\Subst{\alpha := C} \)
by the induction hypothesis.
We have \( \alpha_n @ \alpha_1 \dots \alpha_{n - 1} \in
(\Delta, \alpha @ B) \) from the definition of
\( \Delta, \alpha @ B \svdash A \).
The proof is divided into two cases.\medskip

\begin{asparaitem}
\item Case \( \alpha_n \neq \alpha \): It is the case that \( \alpha_n
  @ \alpha_1 \dots \alpha_{n - 1} \in \Delta \).  The conclusion is
  trivial, because \( \alpha_n @ (\alpha_1 \dots \alpha_{n - 1})
  \Subst{\alpha := C} \in \Delta \Subst{\alpha := C} \).\medskip
\item Case \( \alpha_n = \alpha \): It is the case that
\( \alpha_1 \dots \alpha_{n - 1} = B \).
Let \( C = \beta_1 \dots \beta_m \).
Because \( \Delta \svdash B C \), 
we have
\[
  \Delta_0 = \{ \alpha_1 @ \varepsilon, \alpha_2 @ \alpha_1, \dots,
    \alpha_{n - 1} @ \alpha_1 \dots \alpha_{n - 2}, \beta_1 @ B, \beta_2 @ B \beta_1, \dots,
    \beta_m @ B \beta_1 \dots \beta_{m - 1} \}
  \subseteq \Delta.
\]
Let \( \Delta_1 = \Delta - \Delta_0 \).
Note that \( \alpha \notin \{ \alpha_1, \dots, \alpha_{n - 1}, \beta_1, \dots, \beta_m \}
= \FMV(\Delta_0) \) and \( \Delta_0 \svdash B C \).
Therefore we have \( \Delta \Subst{\alpha := C} = \Delta_0 \Subst{\alpha := C} \cup \Delta_1 \Subst{\alpha := C}
= \Delta_0 \cup \Delta_1 \Subst{\alpha := C} \).
Because \( (B \alpha) \Subst{\alpha := C} = B C \) and \( \Delta_0 \svdash B C \),
we obtain \( \Delta \Subst{\alpha := C} \svdash A \Subst{\alpha := C} \).
\qed
\end{asparaitem}
\end{lem}

\begin{lem}
\label{lem:transition-subst-trans-env}
If \( \svdash \Delta, \alpha @ B \) and \( \Delta \svdash B C \),
then \( \svdash \Delta \Subst{\alpha := C} \).
\begin{proof}
Let \( \beta @ A \in \Delta \Subst{\alpha := C} \).
We show \( \Delta \Subst{\alpha := C} \svdash A \).

There exists some \( A' \) such that \( \beta @ A' \in \Delta \) and \( A = A' \Subst{\alpha := C} \)
because \( \beta @ A \in \Delta \Subst{\alpha := C} \).
Since \( \svdash \Delta, \alpha @ B \), we have \( \Delta, \alpha @ B \svdash A' \).
So by Lemma~\ref{lem:transition-subst-transition}, we have
\( \Delta \Subst{\alpha := C} \svdash A' \Subst{\alpha := C} \) as required.
\end{proof}
\end{lem}

\begin{lem}
\label{lem:transition-subst-type}
If \( \Delta, \alpha @ B \svdash^A \tau \) and \( \Delta \svdash B C \),
then \( \Delta \Subst{\alpha := C} \svdash^{A \Subst{\alpha := C}} \tau \Subst{\alpha := C} \).
\proof
By induction on the structure of the type \( \tau \).\medskip

\begin{asparaitem}
\item Case \( \tau = \textbf{int} \): Because \( \Delta, \alpha @ B
  \svdash^A \textbf{int} \), we have \( \Delta, \alpha @ B \svdash A
  \).  So by Lemma~\ref{lem:transition-subst-transition}, we obtain \(
  \Delta \Subst{\alpha := C} \svdash A \Subst{\alpha := C} \).
  Therefore we obtain \( \Delta \Subst{\alpha := C} \svdash^{A
    \Subst{\alpha := C}} \textbf{int} \) by applying
  \textsc{(ST-Base)}.  \medskip

\item Case \( \tau = \Seal{\beta}{\sigma} \):
Because \( \Delta, \alpha @ B \svdash^A \Seal{\beta}{\sigma} \),
we have \( \Delta, \alpha @ B \svdash^{A \beta} \sigma \).
By the induction hypothesis, we have \( \Delta \Subst{\alpha := C} \svdash^{(A \beta) \Subst{\alpha := C}}
\sigma \Subst{\alpha := C} \).
Therefore by applying \textsc{(ST-Code)} as many times as needed,
we obtain \( \Delta \Subst{\alpha := c} \svdash^{A \Subst{\alpha := C}}
\Seal{\beta \Subst{\alpha := C}}{(\sigma \Subst{\alpha := C})} \).
\medskip

\item Case \( \tau = \forall \beta @ A'. \sigma \):
Because \( \Delta, \alpha @ B \svdash^A \forall \beta @ A'. \sigma \),
we have \( \Delta, \alpha @ B, \beta @ A A' \svdash^A \sigma \).
By the induction hypothesis, we have
\( \Delta \Subst{\alpha := C}, \beta @ (A A') \Subst{\alpha := C} \svdash^{A \Subst{\alpha := C}}
\sigma \Subst{\alpha := C} \).
Hence by applying \textsc{(ST-Univ)},
we obtain \( \Delta \Subst{\alpha := C} \svdash^{A \Subst{\alpha := C}}
\forall \beta @ (A' \Subst{\alpha := C}). (\sigma \Subst{\alpha := C}) \). \qed
\end{asparaitem}
\end{lem}

\begin{lem}
\label{lem:transition-subst-env}
If \( \Delta, \alpha @ B \svdash \Gamma \) and \( \Delta \svdash B C \),
then \( \Delta \Subst{\alpha := C} \svdash \Gamma \Subst{\alpha := C} \).
\begin{proof}
Take \(x:\tau@A \in \Gamma\).  By the well-formedness of \(\Gamma\),
we have \(\Delta, \alpha@B \svdash^A \tau\).
By Lemma~\ref{lem:transition-subst-type},
we have \( \Delta \Subst{\alpha := C} \svdash^{A \Subst{\alpha := C}} \tau \Subst{\alpha := C} \), as required.
\end{proof}
\end{lem}

Now we show the substitution lemma for transition variables.
  However, transition substitution does \emph{not} preserve typing in
  general, and, even worse, transition substitution is undefined in
  some cases.  This is because of the distinction between two kinds of
  applications \( \SIns{M}{\alpha} \) and \( \Ins{M}{A} \).  For
  example, if we substitute \( A \) for \( \alpha \) into \(
  \SIns{M}{\alpha} \), the result is \( \SIns{M}{A} \), which is not a
  valid term.

So, the substitution lemma we will prove is a restricted one.
Transition substitution preserves typing in the following two cases:
\begin{enumerate}[(1)]
\item Substitution of a single transition variable, \( M \Subst{\alpha := \beta} \) for any \( M \).
\item Substitution of an arbitrary transition \(A\) at the initial
  stage in a value, \( v \Subst{\alpha := A} \),
  where \( v \) is a value (an element of \(V^\varepsilon\)) and \(
  \alpha @ \varepsilon \).
\end{enumerate}
They are what we need to prove type soundness.

\begin{lem}
\label{lem:transition-subst-term}\ 
\begin{enumerate}[\em(1)]
\item If \( \Gamma; \Delta, \alpha @ B \svdash^A M : \tau \) and \( \beta @ B \in \Delta \),
then \( \Gamma \Subst{\alpha := \beta}; \Delta \Subst{\alpha := \beta} \svdash^{A \Subst{\alpha := \beta}}
M \Subst{\alpha := \beta} : \tau \Subst{\alpha := \beta} \).
\item 
For any \( v \in V^A \), 
if \( \Gamma; \Delta, \alpha @ \varepsilon \svdash^A v : \tau \) and \( \Delta \svdash B \), then \( \Gamma \Subst{\alpha := B}; \Delta \Subst{\alpha := B} \svdash^{A \Subst{\alpha := B}}
v \Subst{\alpha := B} : \tau \Subst{\alpha := B} \).
\end{enumerate}
\proof
We can prove (1) by induction on the type derivation.
\begin{asparaitem}
\item Case \textsc{(S-Var)}: We have \( M = x \) and \( \Delta, \alpha @ B \svdash \Gamma, x : \tau @ A \).
By Lemma~\ref{lem:transition-subst-env}, we have \( \Delta \Subst{\alpha := \beta} \svdash
\Gamma \Subst{\alpha := \beta}, x : \tau \Subst{\alpha := \beta} @ A \Subst{\alpha := \beta} \).
By applying the \textsc{(S-Var)} rule, we obtain \\
\( \Gamma \Subst{\alpha := \beta}; \Delta \Subst{\alpha := \beta} \svdash^{A \Subst{\alpha := \beta}}
x : \tau \Subst{\alpha := \beta}\).
\\
\item Case \textsc{(S-Gen)}: We have \( M = \Gen{\gamma}{M'} \) and \( \tau = \forall \gamma @ C. \sigma \)
and \( \Gamma; \Delta, \alpha @ B, \gamma @ AC \svdash^A M' : \sigma \).  By the induction hypothesis, we have
\( \Gamma \Subst{\alpha := \beta}; \Delta \Subst{\alpha := \beta}, \gamma @ (AC) \Subst{\alpha := \beta}
\svdash^{A \Subst{\alpha := \beta}} M' \Subst{\alpha := \beta} : \sigma \Subst{\alpha := \beta} \)
(note that \( \alpha \neq \gamma \)).
By applying the \textsc{(S-Gen)} rule, we obtain \\
\( \Gamma \Subst{\alpha := \beta}; \Delta \Subst{\alpha := \beta} \svdash^{A \Subst{\alpha := \beta}}
\Gen{\gamma}{(M' \Subst{\alpha := \beta})} : \forall \gamma @ (C \Subst{\alpha := \beta}). \sigma \Subst{\alpha := \beta} \).
\\
\item Case \textsc{(S-Ins2)}: We have \( M = \SIns{M'}{\gamma_1} \),
and there is some \( \sigma \) and \( \gamma_0 \) and \( C \) such that
\( \tau = \sigma \Subst{\gamma_0 := \gamma_1} \) and
\( \Gamma; \Delta, \alpha @ B \svdash^A M' : \forall \gamma_0 @ C. \sigma \) and
\( \gamma_1 @ A C \in \Delta \).
We can assume without loss of generality \( \gamma_0 \notin \FMV(C) \cup \{ \alpha, \beta \} \).
By the induction hypothesis, we have
\( \Gamma \Subst{\alpha := \beta}; \Delta \Subst{\alpha := \beta} \svdash^{A \Subst{\alpha := \beta}}
M' \Subst{\alpha := \beta} : (\forall \gamma_0 @ C. \sigma) \Subst{\alpha := \beta} \).
Since \( (\forall \gamma_0 @ C. \sigma) \Subst{\alpha := \beta}
  = \forall \gamma_0 @ (C \Subst{\alpha := \beta}). (\sigma \Subst{\alpha := \beta}) \)
and \( \gamma_1 \Subst{\alpha := \beta} @ (AC) \Subst{\alpha := \beta} \in \Delta \Subst{\alpha := \beta} \),
we can apply \textsc{(S-Ins2)} to obtain
\( \Gamma \Subst{\alpha := \beta}; \Delta \Subst{\alpha := \beta} \svdash^{A \Subst{\alpha := \beta}}
(\SIns{M'}{\gamma_1}) \Subst{\alpha := \beta} : \sigma \Subst{\alpha := \beta} \Subst{\gamma_0 := \gamma_1 \Subst{\alpha := \beta}} \).
Because \( \gamma_0 \notin \{ \alpha, \beta \} \), we have
\begin{eqnarray*}
  \sigma \Subst{\alpha := \beta} \Subst{\gamma_0 := \gamma_1 \Subst{\alpha := \beta}}
    & = &
      \sigma \Subst{\gamma_0 := \gamma_1} \Subst{\alpha := \beta} \\
    & = &
      \tau \Subst{\alpha := \beta}.
\end{eqnarray*}
\end{asparaitem}

The proof of (2) is the same as the proof of (1), except for the case \textsc{(S-Ins2)}.
We prove this case.  Let \( v = \SIns{v'}{\gamma_1} \).  To prove this case,
it is important to notice that \( \gamma_1 \neq \alpha \).
This fact can be shown by the following observation.
Assume that \( \gamma_1 = \alpha \).
Then, the stage \(A\) of \( v \) must be \( \varepsilon \)
because the stage of \( \alpha \) is \( \varepsilon \),
but  \( v \in V^\varepsilon\) contradicts the fact that \(v\)
is a transition application.
Therefore \( \gamma_1 \neq \alpha \).
The rest of the proof is the same as (1). \qed
\end{lem}

Then we show the preservation of \( \svdash \) by evaluation.
\begin{lem}[Type Soundness (2)]
\label{lem:staged-preservation}
If \( \Gamma; \Delta \svdash^A M \) and
\( \bred{A}{M}{R} \), then  \( R = v^A\) and \( \Gamma; \Delta \svdash^A v^A \)
for some \(v^A \in V^A\).
\begin{proof}
% An easy induction on the derivation \( \bred{A}{M}{R} \),
% using Substitution Lemma (Lemma~\ref{lem:subst-staged-transition} and
% Lemma~\ref{lem:transition-subst-term}).
Similar to the proof of Theorem~\ref{thm:ml-sound}.
\end{proof}
\end{lem}

We also have the counterpart of Lemma~\ref{lem:sound2} for the new type system.
% The last lemma is the typability of generated code fragments.
We define \( \Delta^{- \alpha} = \{ \beta @ B ~|~ \beta @ \alpha B \in \Delta \} \). %similar to the definition of \( \Gamma^{- \alpha} \).
\begin{lem}
\label{lem:staged-sound2}
If \( \Gamma; \Delta \svdash^{\alpha A} M : \tau \) and \( M \in
V^{\alpha A} \), then we have \( \Gamma^{- \alpha}; \Delta^{- \alpha}
\svdash^A M : \tau \).
\begin{proof}
Similar to the proof of Lemma~\ref{lem:sound2}.
\end{proof}
\end{lem}

Finally, we prove Theorem~\ref{thm:erasure} via a more general property
for which proof by induction works.
% We prove a more general statement,
% because the statement of Theorem~\ref{thm:erasure} is too strong to prove using induction.
\begin{thm}%[Erasure Property for Arbtrary Stage]
\label{thm:general-erasure}
If  \(\Gamma\) is \(\varepsilon\)-free and \( \Gamma; \Delta \svdash^A M : \tau \)
and \( n \) is the length of \( A \), then
\begin{enumerate}[(1)]
\item if \( \bred{A}{M}{N} \), then \( \ered{n}{\flat(M)}{\flat(N)} \); and 
\item if \( \ered{n}{\flat(M)}{N} \), there is some \( N' \) such that
\( \bred{A}{M}{N'} \) and \( N = \flat(N') \).
\end{enumerate}
\proof
By induction on the derivation of \( \bred{A}{M}{N} \) (for (1)) or
\( \ered{n}{\flat(M)}{N} \) (for (2)) with case analysis on the last
rule used in the derivation.  
Actually, the only interesting case is when 
\( M = \Ins{M_0}{B} \) and \( \bred{A}{\Ins{M_0}{B}}{N} \).
  \\
\begin{enumerate}[$\bullet$]
\item Proof of (1): Assume \( M = \Ins{M_0}{B} \) and \( \bred{A}{\Ins{M_0}{B}}{N} \).
The case \( A \neq \varepsilon \) is easy.  Assume \( A = \varepsilon \).
So, the last derivation step for \( \bred{\varepsilon}{\Ins{M_0}{B}}{N} \) is of the form
\vspace{.5em}
\begin{center}
\infrule{
  \bred{\varepsilon}{M_0}{\Gen{\alpha}{M_1}} \AND
  \bred{\varepsilon}{M_1 \Subst{\alpha := B}}{N}
}{
  \bred{\varepsilon}{\Ins{M_0}{B}}{N}
}.
\end{center}

Since \( \Gamma; \Delta \svdash^\varepsilon \Ins{M_0}{B} : \tau \),
we have \( \Gamma; \Delta \svdash^\varepsilon M_0 : \forall \alpha@\varepsilon. \Seal{\alpha}{\sigma} \)
and \( \tau = \sigma \Subst{\alpha := B} \) for some \( \sigma \).
By Lemma~\ref{lem:staged-preservation},
we have \( \Gamma; \Delta \svdash^\varepsilon \Gen{\alpha}{M_1} : \forall \alpha @ \varepsilon. \Seal{\alpha}{\sigma} \)
and \( \Gen{\alpha}{M_1} \in V^\varepsilon \).
Therefore, \( \Gamma; \Delta, \alpha @ \varepsilon \svdash^\varepsilon M_1 : \Seal{\alpha}{\sigma} \) and
\( M_1 \in V^\varepsilon \).
Because \( M_1 \in V^\varepsilon \), we have \( M_1 = \Next{\alpha}{M_2} \) with \( M_2 \in V^\alpha \).
So \( \Gamma; \Delta, \alpha @ \varepsilon \svdash^\alpha M_2 : \sigma \).
By Lemma~\ref{lem:staged-sound2}, \( \Gamma^{- \alpha}; (\Delta, \alpha @ \varepsilon)^{- \alpha}
\svdash^\varepsilon M_2 : \sigma \).
Because \( (\Delta, \alpha @ \varepsilon)^{- \alpha} = \Delta^{- \alpha} \),
we have \( \alpha \notin \FMV((\Delta, \alpha @ \varepsilon)^{- \alpha}) \),
and therefore \( \alpha \notin \FMV(M_2) \).
So \( (\Next{\alpha}{M_2}) \Subst{\alpha := B} = \Next{B}{M_2} \).
Then, we know that the last derivation step must have a more specific shape:
\vspace{.5em}
\begin{center}
\infrule{
  \bred{\varepsilon}{M_0}{\Gen{\alpha}{\Next{\alpha}{M_2}}} \AND
  \bred{\varepsilon}{\Next{B}{M_2}}{N}
}{
  \bred{\varepsilon}{\Ins{M_0}{B}}{N}
}
\end{center}
\vspace{.5em}

By the induction hypothesis, we have \( \ered{0}{\flat(M_0)}{\Lambda \Next{}{\flat(M_2)}} \) and
\( \ered{0}{\overbrace{\Next{}{} \dots \Next{}{}}^m \flat(M_2)}{\flat(N)} \),
where \( m \) is the length of \( B \), because
\( \flat(\Next{B}{M_2}) = \overbrace{\Next{}{} \dots \Next{}{}}^m \flat(M_2) \).
Finally, we obtain \( \ered{0}{\flat(M_0) \app m}{\flat(N)} \) from them.
\\
\item Proof of (2):
Assume \( M = \Ins{M_0}{B} \) and \( \ered{n}{\flat(\Ins{M_0}{B})}{N^\flat} \).
The case \( n \neq 0 \) is easy.  We assume \( n = 0 \).  Then,
 the last derivation step is of the form:
\vspace{.5em}
\begin{center}
\infrule{
  \ered{0}{\flat(M_0)}{\Lambda \Next{}{M_2^\flat}} \AND
  \ered{0}{\overbrace{\Next{}{} \dots \Next{}{}}^m M_2^\flat}{N^\flat}
}{
  \ered{0}{\flat(M_0) \app m}{N^\flat}
}
\end{center}
where \( m \) is the length of \( B \).  Since \( \Gamma; \Delta
\svdash^\varepsilon \Ins{M_0}{B} : \tau \), we have \( \Gamma; \Delta
\svdash^\varepsilon M_0 : \forall \alpha @
\varepsilon. \Seal{\alpha}{\sigma} \) for some \( \sigma \).

By the induction hypothesis, there is some \( M_1 \) such that \(
\bred{\varepsilon}{M_0}{M_1} \) and \( \Lambda \Next{}{M_2^\flat} =
\flat(M_1) \).  So, \( M_1 = \Lambda \alpha. \Next{\beta}{M_2} \) for
some \( \alpha \) and \( \beta \) and \(M_2\) such that \(\flat(M_2) =
M_2^\flat\).  Actually, \( \alpha = \beta \) since \( M_0 \) has the type
\( \forall \alpha @ \varepsilon. \Seal{\alpha}{\sigma} \)
and evaluation preserves types.
% \old{Actually, \( \alpha = \beta \) since \( M_1 \) is
% typable by Lemma~\ref{lem:staged-preservation}.}
By the same argument
as in the proof of (1), we have \( \alpha \notin \FMV(M_2) \).
Therefore, \( (\Next{\alpha}{M_2}) \Subst{\alpha := B} = \Next{B}{M_2}
\).  Because \( \flat(\Next{B}{M_2}) = \overbrace{\Next{}{} \dots
  \Next{}{}}^m M_2^\flat \), by the induction hypothesis, we have \(
\bred{\varepsilon}{\Next{B}{M_2}}{N} \) for some \( N \) with \(
\flat(N) = N^\flat \). \qed
\end{enumerate}
\end{thm}

\paragraph{\emph{Proof of Theorem~\ref{thm:erasure}}}
Theorem~\ref{thm:erasure} is a special case of Theorem~\ref{thm:general-erasure}.
\qed

\section{Proof of Completeness (Theorem~\ref{thm:completeness})}
\label{sec:proof-of-completeness}
\subsection{Completeness in the Quantifier-Free Setting}
We first prove completeness in the quantifier-free setting:
the set of propositions has no quantifiers and
the deduction rules and semantics has no rule for quantifiers.

Formally,
the set $ \TypeSet_{\bot}^{- \forall} $ of
quantifier-free propositions is defined by the following grammar:
\begin{center}
\(\begin{array}{@{}lr@{}r@{}l}
    \textrm{\it Propositions}\quad
      & \phi \in \TypeSet_{\bot}^{- \forall} &{}::={}&
        b \OR
        \bot \OR
	\phi \to \phi \OR
	\Seal{\alpha}{\phi}
  \end{array}\ .
\)
\end{center}
The natural deduction system consists of
(\textsc{Var}),
(\textsc{Abs}),
(\textsc{App}),
(\textsc{$\Next{}{}$}),
(\textsc{$\Prev{}{}$}) and
(\textsc{$\bot$-E}).
The Kripke semantics for this fragment is
essentially the same as the full logic,
but the proposition $ \phi $ in
the satisfaction relation $ {\mathcal T}, v, \rho; s \Vdash \phi $
is restricted to be quantifier-free one.

An assumption $ \Gamma $ is
a (possibly infinite) set $ \{ \phi_i @ A_i \} $
where $ \phi_i \in \TypeSet_\bot^{- \forall} $ and
$ A_i \in \Gene^{*} $.
We say an assumption $ \Gamma $ is \emph{consistent}
if $ \Gamma \nvdash^{\varepsilon} \bot $.
(When \(\Gamma\) is infinite, \(\Gamma \vdash^{A} \phi\)
means that there exists a finite assumption \(\Gamma' \subset \Gamma\)
such that \(\Gamma' \vdash^{A} \phi\).)
We say $ \Gamma $ is \emph{maximally consistent}
when $ \Gamma $ is consistent and
maximal (under the ordinary set inclusion ordering).

\begin{lem}
\label{lem:maximally-consistent}
For any consistent assumption $ \Gamma $,
there is a maximally consistent assumption $ \Gamma' $
s.t. $ \Gamma' \supseteq \Gamma $.
\begin{proof}
We apply Zorn's Lemma to the set
$ \Theta = \{ \Delta ~|~ \Gamma \subseteq \Delta
  \textrm{ and $ \Delta $ is consistent} \} $.
To apply Zorn's lemma,
we prove that every totally ordered subsets $ \Xi \subseteq \Theta $
has its upper bound in $ \Theta $,
by showing $ \bigcup \Xi \in \Theta $.
It is clear that $ \Gamma \subseteq \bigcup \Xi $.

We show that $ \bigcup \Xi $ is consistent.
Assume $ \bigcup \Xi $ is inconsistent,
i.e. $ \bigcup \Xi \vdash^{\varepsilon} \bot $.
Then there is a finite subset $ \Delta \subseteq \bigcup \Xi $
such that $ \Delta \vdash^{\varepsilon} \bot $.
Let $ \{ \phi_i @ A_i ~|~ 0 \le i \le n \} = \Delta $.
Because $ \Delta \subseteq \bigcup \Xi $,
for every $ \phi_i @ A_i \in \Delta $,
there is an assumption $ \Xi_i \in \Xi $ which satisfies
$ \phi_i @ A_i \in \Xi_i $.
Because the number $ n $ of elements in $ \Delta $ is finite
and $ \Xi $ is totally ordered,
there is a maximal assumption
$ \Xi_j \supseteq \Xi_i \textrm{ ($ 0 \le i \le n $)} $.
Therefore $ \Delta \subseteq \Xi_j $ and $ \Xi_j $ is inconsistent,
and this causes a contradiction.
\end{proof}
\end{lem}

Let $ \Gamma $ be maximally consistent.
We construct a transition system $ {\mathcal T}_\Gamma $
and valuations $ \rho_\Gamma $ and $ v_\Gamma $ from it.
The transition system\footnote{If we want to prove completeness
of the logic in Sec.~\ref{sec:another-semantics},
we define $ {\mathcal T} $ as
$ (\{ A \in \Gene^{*} ~|~ \bot @ A \notin \Gamma \}, \Gene,
   \{ \stackrel{\alpha}{\rightarrow} ~|~ \alpha \in \Gene \} ) $,
where $ A \stackrel{\alpha}{\rightarrow} A \alpha $
if $ \bot @ A \alpha \notin \Gamma $ and otherwise undefined.}
$ {\mathcal T}_{\Gamma} $ is
$ (\Gene^*, \Gene, \{ \stackrel{\alpha}{\rightarrow} ~|~ \alpha \in \Gene \}) $,
where $ A \stackrel{\alpha}{\rightarrow} A \alpha $,
and valuation $ \rho_\Gamma $ for transition variables is defined by
$ \rho_{\Gamma} (\alpha) = \alpha $.
We define valuation $ v_\Gamma $ for propositional variables as
$ v_{\Gamma}(A, p) = 1 $ if and only if $ p @ A \in \Gamma $.

\begin{lem}
\label{lem:kripke-model}
$ {\mathcal T}_{\Gamma}, v_{\Gamma}, \rho_{\Gamma}; \varepsilon
  \Vdash \Seal{A}{\phi} $
if and only if $ \phi @ A \in \Gamma $.
\begin{proof}
Easy induction with respect to the proposition $ \phi $.
\end{proof}
\end{lem}

\begin{thm}[Completeness for the quantifier-free fragment]
If $ \Gamma \Vdash \phi $ then $ \Gamma \vdash^{\varepsilon} \phi $.
\begin{proof}
We prove by contraposition.
Assume $ \Gamma \nvdash^{\varepsilon} \phi $.
Then $ \Gamma, (\phi \to \bot) @ \varepsilon $ is consistent.
By Lemma~\ref{lem:maximally-consistent}
there is a maximally consistent assumption
$ \Gamma' \supseteq \Gamma \cup \{ (\phi \to \bot) @ \varepsilon \} $.
By Lemma~\ref{lem:kripke-model},
$ {\mathcal T}_{\Gamma'}, v_{\Gamma'}, \rho_{\Gamma'}; \varepsilon \Vdash \Gamma $
and
$ {\mathcal T}_{\Gamma'}, v_{\Gamma'}, \rho_{\Gamma'}; \varepsilon
  \Vdash \phi \to \bot $.
So, $ {\mathcal T}_{\Gamma'}, v_{\Gamma'}, \rho_{\Gamma'}; \varepsilon
  \nVdash \phi $ and, therefore, $ \Gamma \nVdash \phi $.
\end{proof}
\end{thm}

\subsection{Completeness of the Full Logic}
We prove completeness for the full logic.
The way to prove completeness is
similar to that for first-order predicate logic.

We abbreviate
$ \phi \to \bot $ to $ \neg \phi $,
$ (\neg \phi) \to \psi $ to $ \phi \vee \psi $,
$ \neg (\phi \to \neg \psi) $ to $ \phi \wedge \psi $ and
$ \neg \forall \alpha . \neg \phi $ to $ \exists \alpha . \phi $.
The deduction rules for these connectives can be given
as in ordinary classical logic,
e.g. as follows:
\begin{center}
\infrule[$\wedge$-I]{
  \Gamma \vdash^{A} \phi \AND
  \Gamma \vdash^{A} \psi
}{
  \Gamma \vdash^{A} \phi \wedge \psi
} \\
\infrule[$\exists$-E]{
  \Gamma \vdash^{A} \exists \alpha . \phi \AND
  \Gamma, \phi @ A \vdash^{A} \psi \AND
  \alpha \notin \FMV(\Gamma, \psi @ A)
}{
  \Gamma \vdash^{A} \psi
} .
\end{center}

We abbreviate $ \forall \alpha_1 \forall \alpha_2 \dots \forall \alpha_n \phi $
to $ \forall^n \bar{\alpha} \phi $.
We also write $ \phi $ as $ \forall^0 \bar{\alpha} \phi $.
We say a proposition is in \emph{prenex normal form}
when it has the form
\[
 \forall^{n_0} \bar{\alpha}_0 \exists \beta_1
\forall^{n_1} \bar{\alpha}_1 \dots
\exists \beta_m \forall^{n_{m + 1}} \bar{\alpha}_{m + 1} \psi 
\]
and
$ \psi $ has no quantifier.
It is easy to see that
for any proposition $ \phi $
there is an equivalent proposition $ \phi' $
(i.e., $ \vdash^\varepsilon \phi \leftrightarrow \phi' $)
which is in prenex normal form.
Therefore we assume without loss of generality
that all propositions in an assumption are in prenex normal form.

We construct a quantifier-free assumption $ \Delta(\Gamma) $
from an assumption $ \Gamma $ with quantifiers.
Assume that all binding transition variables are
different from other binding variables and free variables
in $ \Gamma $.

First we construct ``Herbrand universe'' from $ \Gamma $.
Herbrand universe is a term algebra,
which is freely generated algebra from a signature.
Let $ \Gene_0 $ be the set of free transition variables in $ \Gamma $
and $ \Gene_1 $ be the set of binding transition variables in $ \Gamma $.
The set $ S_\Gamma $ is defined as follows:
\begin{equation*}
  S_\Gamma = \{ (\beta_k, \Sigma_{i = 1}^{k - 1} n_i) ~|~
    \forall^{n_0} \bar{\alpha}_0 \exists \beta_1
    \dots \exists \beta_k \dots
    \forall^{n_{m}} \bar{\alpha}_{m} \psi @ A \in \Gamma \}
\end{equation*}
We regard $ \textrm{dom}(S_\Gamma) $ as the set of pairs of
function symbols and its arities.
If there is no constant (i.e., arity 0 function) in $ S_\Gamma $
and $ \Gene_0 = \Gene_1 = \emptyset $,
then we add a constant $ (c, 0) $ to $ S_\Gamma $.
If there is no function symbol whose arity is not 0,
then we add a function symbol $ (f, 1) $ to $ S_\Gamma $.
We define the set $ H_\Gamma $ as the least set
which satisfies the following condition:
\begin{enumerate}[$\bullet$]
\item if $ (\beta, n) \in S_\Gamma $ and
  $ \alpha_i \in H_\Gamma \cup \Gene_0 \cup \Gene_1 $($ 1 \le i \le n $),
  then $ \beta(\alpha_1, \dots, \alpha_n) \in H_\Gamma $
\end{enumerate}

Second we regard $ \Gene $ as the ``Herbrand universe''
by following observation.
Because $ \Gene $ is countably infinite and
$ \Gene_0 $ and $ \Gene_1 $ is countable,
we assume without loss of generality
that $ \Gene \backslash (\Gene_0 \cup \Gene_1) $ is countably infinite.
It is easy to see that $ H_\Gamma $ is countably infinite.
Therefore there is an bijection
$ H_\Gamma \to (\Gene \backslash (\Gene_0 \cup \Gene_1)) $.
By using this bijection,
we regard any element in $ H_\Gamma $ as an element in $ \Gene $.
So $ \Gene_0 \cup \Gene_1 \cup H_\Gamma = \Gene $ and
$ \Gene_0 \cap H_\Gamma = \Gene_1 \cap H_\Gamma = \emptyset $.

% First we define the set of transition variables for $ \Delta(\Gamma) $,
% which includes all free transition variables in $ \Gamma $
% and ``Skolem functions''.
% We define a function $ E_\Gamma $ from binding variables
% with the existential quantifier to natural numbers called arities,
% as follows:
% \begin{equation*}
%   E_\Gamma = \{ (\beta_k, \Sigma_{i = 1}^{k - 1} n_i) ~|~
%     \forall^{n_0} \bar{\alpha}_0 \exists \beta_1
%     \dots \exists \beta_k \dots
%     \forall^{n_{m}} \bar{\alpha}_{m} \psi @ A \in \Gamma \}
% \end{equation*}
% We define the set $ \Gene' $ of transition variables as
% the least set which satisfies following conditions:
% \begin{enumerate}
% \item $ \FMV(\Gamma) \subseteq \Gene' $
% \item if $ \beta \in \textrm{dom}(E_\Gamma) $ has arity $ n $ and
%   $ \alpha_i \in \Gene' $ ($ 1 \le i \le n $),
%   then there exists a unique transition variable $ \gamma \in \Gene' $
% \end{enumerate}
% We write such $ \gamma $ as $ \beta(\alpha_1, \dots, \alpha_n) $ to
% make the dependency no $ \beta, \alpha_1, \dots, \alpha_n $ explicit.

Then we define $ \Delta(\Gamma) $ as follows:
\begin{eqnarray*}
  \Delta(\Gamma) = & \{ \psi & \Subst{
    \bar{\alpha}_0, \beta_1, \bar{\alpha}_1, \beta_2, \dots
    \bar{\alpha}_m :=
    \bar{\gamma}_0,
    \beta_1(\bar{\gamma}_0),
    \bar{\gamma}_1,
    \beta_2(\bar{\gamma}_0, \bar{\gamma}_1),
    \dots,
    \bar{\gamma}_m} @ A \\
    & & ~|~
    \forall^{n_0} \bar{\alpha}_0 \exists \beta_1 \dots
    \forall^{n_m} \bar{\alpha}_m \psi @ A \in \Gamma,
    \gamma_{i j} \in \Gene \}
\end{eqnarray*}

\begin{lem}
\label{thm:lemma1}
If $ \Gamma \nvdash^{\varepsilon} \bot $,
then $ \Delta(\Gamma) \nvdash^{\varepsilon} \bot $
in quantifier-free logic.
\begin{proof}
We prove this lemma by contraposition.

Assume that $ \Delta(\Gamma) \vdash^{\varepsilon} \bot $.
Then there is a finite subset
$ \{ \phi_i @ A_i ~|~ 0 \le i \le n \} \subseteq \Delta $
such that $ \{ \phi_i @ A_i ~|~ 0 \le i \le n \} \vdash^{\varepsilon} \bot $.
From the construction of $ \Delta(\Gamma) $,
for every $ \phi_i $,
there are $ \theta_i $ and $ \gamma_{i j} \in \Gene $
such that $ \theta_i @ A \in \Gamma $ and
$ \theta_i = \forall^{n_0} \bar{\alpha}_0 \exists \beta_1 \dots
    \forall^{n_m} \bar{\alpha}_m \psi_i $ and
$ \phi_i = \psi_i \Subst{
    \bar{\alpha}_0, \beta_1, \bar{\alpha}_1, \beta_2, \dots
    \bar{\alpha}_m :=
    \bar{\gamma}_0,
    \beta_1(\bar{\gamma}_0),
    \bar{\gamma}_1,
    \beta_2(\bar{\gamma}_0, \bar{\gamma}_1),
    \dots,
    \bar{\gamma}_m} $.
Then we can get
$ \{ \theta_i @ A_i ~|~ 0 \le i \le n \} \vdash^{\varepsilon} \bot $
by using (\textsc{Ins}) and ($\exists$-E)
in an appropriate order.
Therefore
$ \Gamma \vdash^\varepsilon \bot $.
\end{proof}
\end{lem}

\begin{lem}
\label{thm:lemma2}
If $ {\mathcal T}, v, \rho; \varepsilon \Vdash \Delta(\Gamma) $,
then $ {\mathcal T}, v, \rho; \varepsilon \Vdash \Gamma $.
\begin{proof}
Let $ \forall^{n_0} \bar{\alpha}_0 \exists \beta_1 \dots
  \forall^{n_m} \bar{\alpha}_m \psi @ A \in \Gamma $.
We can get the following proposition by easy induction on the number of
existential quantifiers.
\begin{equation*}
  \textrm{for all $ \gamma_{i, j} \in \Gene $}
  {\mathcal T}, v, \rho; \varepsilon \Vdash
  \exists \beta_k \forall^{n_k} \bar{\alpha}_k \dots
  \forall^{n_m} \bar{\alpha}_n \psi
  \Subst{\bar{\alpha}_0, \beta_1, \dots, \bar{\alpha}_k :=
    \bar{\gamma}_0, \beta_1(\bar{\gamma}_0), \dots, \bar{\gamma}_k}
\end{equation*}
This lemma is an easy consequence of this proposition.
\end{proof}
\end{lem}

\paragraph{\emph{Proof of Theorem~\ref{thm:completeness}}}
We prove by contraposition.
Assume $ \Gamma \nvdash^{\varepsilon} \phi $.
Then $ \Gamma, \neg \phi @ \varepsilon \nvdash^{\varepsilon} \bot $.
By Lemma~\ref{thm:lemma1},
$ \Delta(\Gamma \cup \{ \neg \phi @ \varepsilon \} ) $ is a consistent assumption
in the quantifier-free fragment.
Therefore,
from completeness of the quantifier-free fragment,
there is a model such that
$ {\mathcal T}, v, \rho; \varepsilon \Vdash \Delta(\Gamma \cup \{ \neg \phi @ \varepsilon \}) $.
By Lemma~\ref{thm:lemma2},
$ {\mathcal T}, v, \rho; \varepsilon \Vdash \Gamma $
and $ {\mathcal T}, v, \rho; \varepsilon \Vdash \neg \phi $.
Therefore $ {\mathcal T}, v, \rho; \varepsilon \nVdash \phi $.
So $ \Gamma \nVdash \phi $.
\qed

\end{document}